\documentclass[aps,prl,reprint,superscriptaddress]{revtex4-2}
\usepackage[T1]{fontenc}
\usepackage[utf8]{inputenc}
\usepackage{amsmath,amssymb,amsthm,bm,mathtools}
\usepackage{array,booktabs,graphicx,microtype,xcolor}
\usepackage[hidelinks]{hyperref}
\newcommand{\tr}{\operatorname{tr}}
\newcommand{\Law}{\operatorname{Law}}
\newcommand{\D}{\mathcal{D}}
\newcommand{\Hc}{\mathcal{H}}
\newcommand{\C}{\mathcal{C}}
\newcommand{\M}{\mathcal{M}}
\newcommand{\A}{\mathfrak{A}}
\newcommand{\Wone}{W_{1,\mathrm{tr}}}
\newcommand{\dd}{\,\mathrm{d}}
\newcommand{\E}{\mathbb{E}}
\newcommand{\Prob}{\mathbb{P}}
\newcommand{\CP}{\mathbb{CP}}
\newcommand{\RP}{\mathbb{RP}}
\newcolumntype{P}[1]{>{\raggedright\arraybackslash}p{#1}}
\newtheorem{theorem}{Theorem}
\newtheorem{lemma}{Lemma}
\newtheorem{proposition}{Proposition}
\newtheorem{corollary}{Corollary}
\begin{document}

\title{Distinct Feedback-Strength Requirements for Quantum-State Ensemble Preparation under Channel-Equivalent Monitoring}
\author{Haitao Huang}
\email{huanghaitao4396@gmail.com}
\affiliation{School of Computer Science and Engineering, Macau University of Science and Technology, Macau 999078, China}
\author{Qinglin Zhao}
\email[Corresponding author: ]{qlzhao@must.edu.mo}
\affiliation{School of Computer Science and Engineering, Macau University of Science and Technology, Macau 999078, China}
\affiliation{Zhuhai MUST Science and Technology Research Institute, Zhuhai, China}
\author{Xiaoyu Li}
\email{xiaoyuuestc@uestc.edu.cn}
\affiliation{School of Information and Software Engineering, University of Electronic Science and Technology of China, Chengdu, China}

\begin{abstract}
Different measurements of the same environment can preserve the average quantum channel yet change the feedback strength needed to prepare a state ensemble.  We establish this distinction for a multiqubit register under complete Pauli monitoring, targeting the uniform ensemble of real-amplitude pure states.  With unit detection efficiency, public records, and bounded Hamiltonian feedback, fixed transverse noise requires a minimum peak strength proportional to $\epsilon^{-1}$ at distributional error $\epsilon$.  Tangent noise instead admits a construction requiring only $O[\sqrt{\log(1/\epsilon)}]$.  At fixed preparation time, we determine the optimum over all allowed feedback policies and derive a geometric lower bound beyond radial symmetry.  The separation identifies the measurement realization as part of the control-resource specification for monitored ensemble preparation.
\end{abstract}

\maketitle

A quantum-state ensemble is specified by its constituent states and their sampling probabilities.  During continuous monitoring, this distribution is assembled from trajectories conditioned on the measurement record.  Different readouts of the same environment can leave the discarded-record channel unchanged while producing different conditioned trajectories \cite{WisemanDiosi2001,ChiaWiseman2011,Pinol2024Unravellings,GaonaReyes2025UnravellingLimits}.  When monitoring must remain active, measurement design requires distinguishing a poorly chosen controller from a budget that no allowed controller can overcome.  We make this distinction for the \emph{optimal control resources} of a prescribed terminal ensemble.

The separation is explicit for a multiqubit register targeting the uniform ensemble of real-amplitude pure states.  The resource is the maximum Hamiltonian strength during preparation.  At fixed dimension and time, target-preserving monitoring admits an $O[\sqrt{\log(1/\epsilon)}]$ construction, whereas fixed transverse noise imposes a $\Theta(\epsilon^{-1})$ minimum peak strength at terminal ensemble error $\epsilon$.  After successful capture, one readout permits sampling within the family without normal feedback; the other requires sustained confinement (Fig.~\ref{fig:instrument}).  An exact reduction yields the all-controller optimum and a budget window favoring tangent monitoring; a geometric lower bound extends beyond radial symmetry.

Monitoring affects estimation and purification \cite{Jacobs2003Purification,RuskovKorotkovMolmer2010,RuskovCombesMolmerWiseman2012,JiangWangMartinWhaley2020,HacohenGourgy2016Noncommuting}, and its influence on optimal feedback cost is already established.  Wiseman and Doherty optimize the environment measurement for stationary linear-quadratic-Gaussian feedback at fixed system--environment dynamics \cite{WisemanDoherty2005}.  Ohzeki and Jordan construct a measured Schr\"odinger bridge under a chosen monitored reference diffusion \cite{OhzekiJordan2026MeasuredBridge}.  Here the new result is a peak-strength--accuracy separation for a nonlinear terminal ensemble task, with a finite-time optimum over all allowed causal feedback.  Its same-budget window certifies reachability for one readout and impossibility for another with the identical average channel.  The readouts vary backaction geometry and record information jointly.

\begin{figure*}[t]
 \includegraphics[width=\textwidth]{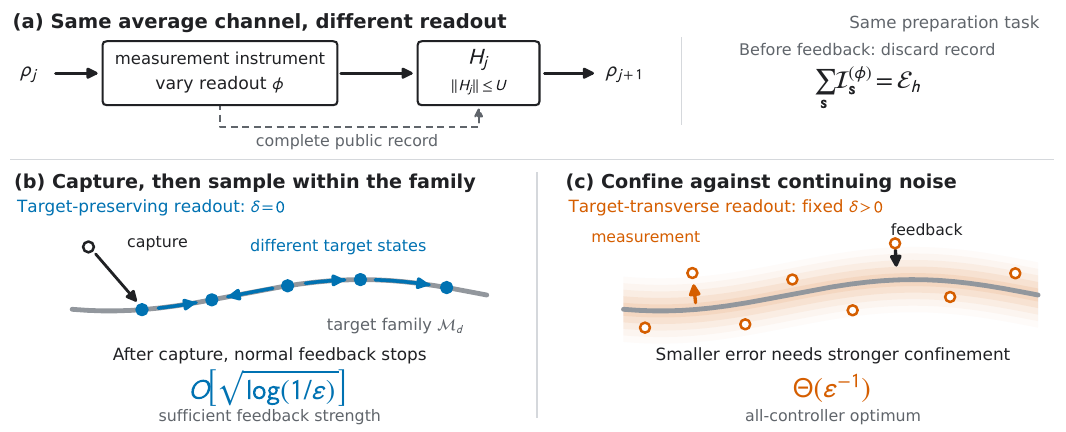}
 \caption{\textbf{Same average channel, different feedback scaling.}
 (a) Readout phases vary at fixed discarded-record channel.  The solid line carries the system; the dashed complete public record conditions the subsequent Hamiltonian.
 (b) At $\delta=0$, successful capture is followed by monitored motion within $\M_d$; normal feedback stops.
 (c) At fixed $\delta>0$, measurement continually displaces states from $\M_d$.  Inward feedback maintains a characteristic distance $\ell_\perp\sim\kappa\delta/U$ at large $U$.  Gray curves depict the same target-family slice in pure-state space, rather than a trajectory.  Dots and shading illustrate states and confinement, not sampled data or hard support.  The scalings show a sufficient construction in (b) and the all-controller optimum in (c), as $\epsilon\to0$ at fixed dimension and time.  Task and control conditions are matched; geometry and record information vary jointly.}
 \label{fig:instrument}
\end{figure*}

\paragraph{Task and channel-equivalent monitoring.}
Let $d=2^n$.  One run starts from a public complex-Haar state and returns an $n$-qubit register together with the public measurement record.  At unit efficiency the record and declared policy determine the conditional state $\rho_T$, so repeated runs produce $\Law(\rho_T)$.  The target is the invariant law $\nu_d$ on
$\M_d=\RP^{d-1}\subset\CP^{d-1}$: pure states with real amplitudes in a fixed basis, up to global phase.  For $n\ge2$ this family contains entangled states and has real dimension $d-1$ inside a $2(d-1)$-dimensional pure-state space.  Calibration uses record-conditioned measurements across repeated preparations or prespecified record bins \cite{WisemanVaccaro2001PREnsembles}.
Real states and real-Kraus operations also form the operational resource theory of imaginarity \cite{Wu2021OperationalImaginarity}.  Our tangent instrument preserves this established real-state structure; the resource considered here is the strength needed to prepare its ensemble under ongoing monitoring.

The prescribed monitoring remains active at fixed strength throughout $[0,T]$, including the terminal interval.  The public filtration $\mathcal F_t^{\rm obs}$ contains the initial label, an optional independent seed, and all records through $t$.  A predictable controller uses only information available before each new measurement increment.  With $\hbar=1$, we allow
\begin{equation}
 \A_U=\{H_t=H_t^\dagger:H_t\ \text{predictable},\ \|H_t\|_{\rm op}\le U\}.
\label{eq:controls}
\end{equation}
Thus $U$ is the essential supremum of the Hamiltonian strength over time and public records; it is the single feedback resource varied below.  The control enters through $-i[H_t,\rho_t]\dd t$.  Fixed fresh-ancilla preparation, readout, and reset are included in the matched instrument that realizes Eq.~(\ref{eq:kraus}).  The control class excludes same-increment current feedback, hidden records, additional reset channels, postselection, and a free terminal pulse.  We compare ensemble laws with trace-Wasserstein distance $\Wone$: the smallest mean trace distance $d_{\rm tr}(\rho,\sigma)=\|\rho-\sigma\|_1/2$ over couplings of the two laws.  For the readout parameter $\delta$ defined below, the optimal error and minimum peak strength are
\begin{align}
 E^*_{d,\delta}(U,T)&=\inf_{H\in\A_U}
 \Wone(\Law^{\delta,H}[\rho_T],\nu_d),\nonumber\\
 \C_{d,\delta}(\epsilon)&=\inf\{U:E^*_{d,\delta}(U,T)\le\epsilon\}.
\label{eq:cost}
\end{align}

For every nonidentity Hermitian Pauli string $P$, use the binary instrument with outcome $s=\pm1$,
\begin{equation}
 K_{s,h}^{(P,\phi)}=
 \frac{\sqrt{1-p_h}I+s\sqrt{p_h}e^{-i\phi}P}{\sqrt2},
 \quad p_h=\frac{1-e^{-2\Gamma h/d^2}}2.
\label{eq:kraus}
\end{equation}
Set $e^{-i\phi_P}=1$ for $P^{\mathsf T}=P$ and
$e^{-i\phi_P}=\sqrt\delta-i\sqrt{1-\delta}$ for $P^{\mathsf T}=-P$.
The discarded-record factors commute and give, for every $\delta\in[0,1]$,
\begin{equation}
 \mathcal E_h^{(\delta)}=
 \exp\!\left[h\Gamma\left(\frac{I}{d}\tr-\mathrm{id}\right)\right].
\label{eq:macrochannel}
\end{equation}
Equation~(\ref{eq:macrochannel}) is the channel-matching constraint used throughout.
The public alphabets and record rates are also matched.  The readout phase changes both the local Fisher information of the public record and the conditional covariance.  Equal record access therefore allows different information content.  In the continuous limit, the conditional state follows a diffusive quantum filter driven by independent public innovations.  The Supplemental Material gives this filter and the exact Fisher quadratic form; equality of the latter forces equal $\delta$ within this family.

\paragraph{Control-strength separation.}
Here tangent and normal refer to the target family within pure-state space, rather than to an individual trajectory.  We compare transverse measurement diffusion; curvature-induced It\^o drift and inward Hamiltonian feedback play different roles.  At the target, $\delta$ sets the transverse diffusion strength: $\delta=0$ preserves the family once reached, while every $\delta>0$ produces transverse fluctuations requiring continued repair (Fig.~\ref{fig:instrument}).  This distinction determines the accuracy dependence of the required feedback.

\begin{theorem}[channel-equivalent control-strength separation]
\label{thm:multi}
Fix $n\ge2$, $d=2^n$, and $T,\Gamma>0$ under the contract above.  For every fixed $0<\delta\le1$,
\begin{equation}
 \C_{d,\delta}(\epsilon)=\Theta(\epsilon^{-1})
 \quad(\epsilon\downarrow0).
\label{eq:fixedcost}
\end{equation}
For the tangent instrument $\delta=0$, first-hit capture followed by stopping the normal controller gives
\begin{equation}
 \C_{d,0}(\epsilon)\le
 \frac{\pi}{4T}+\sqrt{\frac{2\kappa}{T}\log\frac1\epsilon},
 \quad \kappa=\frac{\Gamma}{2d}.
\label{eq:tangentresult}
\end{equation}
This is a constructive $O[\sqrt{\log(1/\epsilon)}]$ upper bound; a matching tangent lower asymptotic is open.  When $\delta$ varies with accuracy, a budget of order $\sqrt{\log(1/\epsilon)}$ suffices exactly when
$\delta=O[\epsilon\sqrt{\log(1/\epsilon)}]$.
\end{theorem}
The lower bound includes feedback with arbitrary memory of the public record.  The calibration condition specifies the readout accuracy needed to retain the tangent construction's resource scaling.

\paragraph{Why the bound is unavoidable and attainable.}
The proof reduces both the unavoidable error and an attaining feedback to the distance from the target.  For $z=\psi^{\mathsf T}\psi$, that distance is
\begin{equation}
 q=\tfrac12\arccos|z|\in[0,\pi/4],\qquad
 d_{\rm tr}([\psi],\M_d)=\sin q.
\label{eq:distance}
\end{equation}
After choosing a global phase, write $\psi=\cos q\,e+i\sin q\,f$ with real orthonormal $e,f$.  The closest target ray is $[e]$.  The controlled radial coordinate obeys
\begin{align}
 \dd q_t={}&[b_{d,\delta}(q_t)+v_t]\dd t
 +\sqrt{\kappa f_\delta(q_t)}\dd W_t,\nonumber\\
 |v_t|\le{}&U,\qquad f_\delta(q)=\delta+(1-\delta)\sin^2(2q).
\label{eq:radial}
\end{align}
The drift $b_{d,\delta}$ is fixed by the instrument; its expression and derivation are given in the Supplemental Material.  In particular, $f_\delta(0)=\delta$: transverse fluctuations vanish on the target only for the tangent instrument.  Hamiltonian feedback changes the drift by at most $U$; it cannot cancel the instantaneous quadratic variation at a fixed state.

Let $R_0$ have Haar radial density
$p_{\rm H}(q)=2(d-1)\sin^{d-2}(2q)\cos(2q)$, and let
\begin{equation}
 \dd R_t=[b_{d,\delta}(R_t)-U]\dd t
 +\sqrt{\kappa f_\delta(R_t)}\dd W_t.
\label{eq:reference}
\end{equation}
\emph{All-controller bound.}  The reference uses the largest allowed inward drift.  It is absorbed at zero for $\delta=0$; for $n\ge2$ and $\delta>0$, neither endpoint is reached.  A localized one-dimensional comparison, performed in each controller's own public filtration, gives $R_t\le q_t$ because $v_t\ge-U$.  Every terminal coupling to a target-supported law therefore costs at least $\E\sin R_T$.  This argument includes arbitrary causal memory of the measurement record.

\emph{Attaining feedback and the ensemble law.}  In the real frame above, the public Hamiltonian
\begin{equation}
 H_N=U(ef^{\mathsf T}+fe^{\mathsf T})
 \quad(0<q<\pi/4).
\label{eq:HN}
\end{equation}
has norm $U$ and realizes $v=-U$.  Its phase-independent form and endpoint construction are given in the Supplemental Material.  At $\delta=0$ it is stopped after the first target hit.  The prior, measurement law, and this feedback are covariant under real orthogonal transformations, $O(d)$.  Consequently nearest-real-state projection has exactly law $\nu_d$, and coupling each state to its projection costs $\E\sin R_T$.  This supplies the correct relative sampling weights as well as proximity to the target family.  Thus, for every $\delta\in[0,1]$ and $U\ge0$,
\begin{equation}
 E^*_{d,\delta}(U,T)=\E[\sin R_T].
\label{eq:finiteexact}
\end{equation}

\emph{Accuracy cost.}  For $\delta>0$, let $L_{d,\delta}(U)$ denote the invariant expectation of $\sin R$ for Eq.~(\ref{eq:reference}).  A stationary comparison gives, for every $H\in\A_U$ and every target law supported on $\M_d$,
\begin{align}
 \Wone(\Law^{\delta,H}[\rho_T],\nu)&\ge L_{d,\delta}(U),\nonumber\\
 L_{d,\delta}(U)&=\frac{\kappa\delta(d-1)}{2U}
 [1+O_d(\kappa^2\delta/U^2)].
\label{eq:lower}
\end{align}
Here $L_{d,\delta}$ is a stationary lower bound at every finite $T$, and $E^*_{d,\delta}(U,T)\to L_{d,\delta}(U)$ as $T\to\infty$.  The invariant density is derived in the Supplemental Material.

For residual normal noise, capture must be followed by continued inward control.  Its error budget is $\Pr(\tau_*>t_c)+B_d\kappa\delta/U$, where $\tau_*$ is the first entry into a target tube, $t_c<T$ is a fixed capture deadline, and $B_d$ depends only on dimension.  The first term bounds failed capture; the second bounds the mean distance after capture uniformly in the remaining time.  At sufficiently large $U/(\kappa\sqrt\delta)$, this construction matches the inverse-accuracy lower bound.  For $\delta=0$, exact capture is permanent and only its failure probability remains.

\paragraph{Common-budget comparison and numerical evaluation.}
The bounds give a finite-accuracy comparison under a common budget.  For
$0<\epsilon<L_{d,\delta}(0)$ define
\begin{align}
 U_{\rm N}^{\rm lb}(\epsilon;\delta)
 &=\inf\{U:L_{d,\delta}(U)\le\epsilon\},\nonumber\\
 U_{\rm T}^{\rm ub}(\epsilon)
 &=\frac{\pi}{4T}+\sqrt{\frac{2\kappa}{T}\log\frac1\epsilon}.
\label{eq:budgetthresholds}
\end{align}
$L_{d,\delta}$ is continuous and strictly decreasing, so the first
quantity is the unique solution of $L_{d,\delta}(U)=\epsilon$.  Whenever
$U_{\rm T}^{\rm ub}<U_{\rm N}^{\rm lb}$, every peak budget in
$[U_{\rm T}^{\rm ub},U_{\rm N}^{\rm lb})$ admits the tangent capture
strategy with error at most $\epsilon$, whereas every allowed controller
for the $\delta$ instrument has error greater than $\epsilon$ (Fig.~\ref{fig:resource}).  These are sufficient and necessary thresholds, not exact finite-accuracy costs.
For $d=4$, $\Gamma=4$, $T=1$, and $\epsilon=0.01$, the tangent sufficient threshold is $2.9314$, while the $\delta=1$ necessary threshold is approximately $74.975$.  Thus $U=4$, well inside this window, suffices for tangent preparation and is excluded for every normal-instrument controller.

\begin{figure*}[t]
 \includegraphics[width=\textwidth]{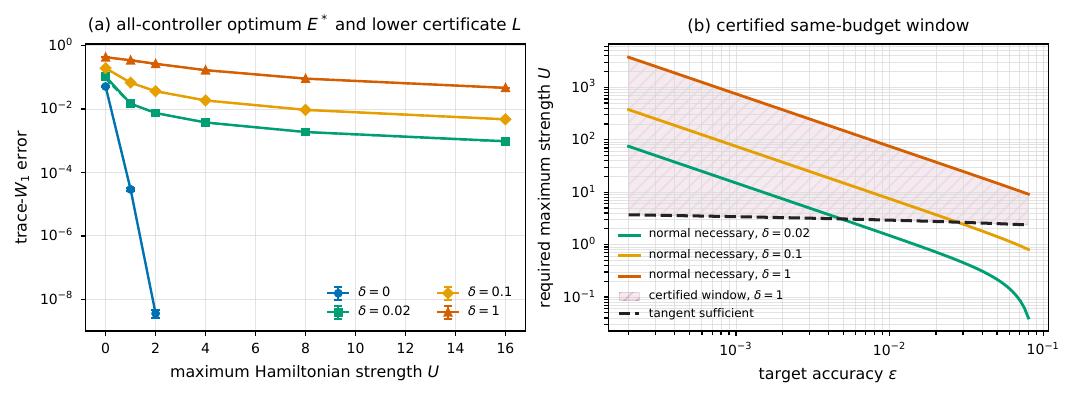}
 \caption{\textbf{Continuous-time resource separation for $d=4$.}
 Here $\Gamma=4$, $\kappa=0.5$, and $T=1$.
 (a) Deterministic evaluations of the finite-time all-controller optimum $E^*(U,T)$ (markers and solid lines).  Dashed curves are the analytic stationary lower certificates $L(U)$ for $\delta>0$; the resolved $\delta=0$ points show the tangent case.  Error bars are twice the sum of successive spatial and temporal differences and are resolution proxies, not rigorous error bounds.
 (b) Inverting $L(U)=\epsilon$ gives necessary budgets for the normal instruments (solid).  The black dashed curve is the sufficient tangent construction in Eq.~(\ref{eq:tangentresult}).  The shaded $\delta=1$ region illustrates the analytically certified same-budget window: the tangent construction succeeds while every allowed $\delta=1$ policy fails.  Boundaries are drawn by floating-point quadrature and interpolation, not interval arithmetic; neither is an exact finite-accuracy optimum.  Both panels use the same prior, target, horizon, and control contract.}
 \label{fig:resource}
\end{figure*}

Figure~\ref{fig:resource} evaluates the scalar optimum and the common-budget bounds for $d=4$.  At $U=8$, $E^*_{4,0.02}=0.00189$ with resolution proxy $3.0\times10^{-5}$, whereas $E^*_{4,1}=0.09121$ with proxy $2.0\times10^{-4}$.  Their excesses above $L(U)$ are smaller than these proxies, so the additional transient costs are unresolved.  Independent discretizations support the resolved separation.  The Supplemental Material gives the main-figure method and pointwise refinement values; auxiliary finite-step data and failed checks remain available in the reproduction files.

\paragraph{Robustness and geometric scope.}
The lower bound also controls perturbations of the target law.  If $\nu_d^\sigma$ obeys
$\Wone(\nu_d^\sigma,\nu_d)\le\sigma$, then every admissible policy satisfies
\begin{equation}
 \Wone(\Law^{\delta,H}[\rho_T],\nu_d^\sigma)
 \ge [L_{d,\delta}(U)-\sigma]_+.
\label{eq:thicktarget}
\end{equation}
Thus error at most $\epsilon$ requires $L_{d,\delta}(U)\le\epsilon+\sigma$.  For fixed $d,\kappa,\delta>0$ and $\epsilon+\sigma\downarrow0$, the certified necessary peak is
$\kappa\delta(d-1)/[2(\epsilon+\sigma)]\,[1+O_d((\epsilon+\sigma)^2/\delta)]$; at fixed nonzero $\sigma$, this lower bound no longer forces a $1/\epsilon$ divergence.
Here $\sigma$ measures distance between laws, not a pointwise tube width.  Finite-accuracy separation persists when $0\le\sigma<\epsilon$ and
$U_{\rm T}^{\rm ub}(\epsilon-\sigma)\le U<U_{\rm N}^{\rm lb}(\epsilon+\sigma;\delta)$:
the tangent construction succeeds for the perturbed target, while every normal policy fails.  The Supplemental Material derives this window by the triangle inequality; it does not solve a general fixed-width target problem.

For complete-Pauli monitoring with detection efficiency $\eta<1$, the public conditional state is generally mixed.  The Supplemental Material proves a controller-uniform positive trace-$W_1$ floor to every pure target family at each fixed $\eta<1$.  Hence high-accuracy preparation requires $1-\eta=O(\epsilon)$.  At $d=4$, $\Gamma=4$, $T=1$, and $\eta=0.99$, this floor is $3.7487\times10^{-3}$, excluding error $10^{-3}$ at any peak.  This bounds the applicability of the ideal separation; it is not a sufficient finite-efficiency preparation condition.

The operator norm is one realization of a more intrinsic resource.  In the Fubini--Study convention used here, Hamiltonian motion has metric speed
$|\dot X_t|_{\rm FS}=\Delta_{\psi_t}H_t\le\|H_t\|_{\rm op}$, so Eq.~(\ref{eq:controls}) bounds the finite-variation speed available to repair target-normal displacement.  The lower theorem below uses only that induced state-space speed bound; constructive attainability still requires a Hamiltonian that realizes the desired vector field, as Eq.~(\ref{eq:HN}) does in the present geometry.

The normal-noise mechanism does not require radial symmetry.  Let $\M$ be a compact embedded target in a compact state manifold.  Assume the local tubular and global embedded-increment bounds stated in the Supplemental Material: uniform finite-variation rate at most $B_0+c_HU$, bounded covariance, and target-normal quadratic-variation rate at least $\lambda_\perp>0$.  The global bounds control arrival from outside the tube.  Then, for sufficiently large $U$,
\begin{equation}
 W_{1,g}(\Law[X_T],\nu)\ge
 c\frac{\lambda_\perp}{B_0+c_HU},
 \qquad \operatorname{supp}\nu\subset\M.
\label{eq:general}
\end{equation}
During a terminal window $h\asymp\lambda_\perp/(B_0+c_HU)^2$, normal martingale displacement is of order $\sqrt{\lambda_\perp h}$, whereas bounded repair is of order $(B_0+c_HU)h$.  The Supplemental Material states the stopping, covariance, tube, and terminal-law assumptions needed to make this comparison uniform over all causal controls.

A pair of oppositely oriented spin-coherent states gives a second realization, on an orbit of intermediate K\"ahler angle (a proper-slant orbit).  For
$[|j_+,\bm n\rangle\otimes|j_-,-\bm n\rangle]$ with $0<j_-<j_+$, monitoring $K_a$ or $iK_a$ gives the same generator $\gamma\sum_a\D[K_a]$.  The first instrument has constant normal rate $4\gamma j_+j_-/(j_++j_-)$, while the second is tangent.  From the same orbit point to the same smooth positive target density on the orbit, Eq.~(\ref{eq:general}) gives an all-controller $\Omega(1/\epsilon)$ normal cost, whereas tangent heat-kernel mixing and bounded transport realize the target at a finite peak independent of $\epsilon$.  This extension uses an orbit-supported prior; ambient-Haar capture and a matching normal upper bound are separate questions.

The Supplemental Material also states single-excitation and bounded-support graph realizations, each with its explicit prior, locality, and control contract.

\paragraph{Implications for quantum diffusion.}
The instrument in Eq.~(\ref{eq:kraus}), followed by feedback in Eq.~(\ref{eq:controls}), defines a constrained quantum diffusion layer: fixed weak system--ancilla interactions, a public measurement record, and a bounded record-conditioned Hamiltonian.  For every fixed bounded Lipschitz Markov controller, the discrete chain converges weakly to the continuous filter studied here, including its terminal $\Wone$ law (Supplemental Material).  The trainable reverse transport, when used, is the bounded system Hamiltonian.  Monitored diffusion generators with Hamiltonian denoising provide related algorithmic settings \cite{Liu2025MBQDM,Bompais2026Reverse,Gabbassov2026ReverseSSE}; theorem coverage requires the stated instrument and control restrictions.  Unrestricted system--ancilla QuDDPM circuits lie outside this class \cite{Zhang2024QuDDPM}.  The measurement-current Hamiltonian of Dubey and John contains a same-increment stochastic term \cite{DubeyJohn2026ScoreFeedback}, as does no-knowledge cancellation \cite{Szigeti2014NoKnowledge}; these use a different resource from bounded $H_t\dd t$.  Joint measurement--control optimization addresses complementary endpoint objectives \cite{KarmakarJordan2026Pontryagin}.

For continuously monitored ensemble preparation, the feedback-free channel alone does not determine the peak-control accuracy cost.  The same average noise channel can support either inverse-accuracy feedback costs or a tangent construction with only square-root-logarithmic growth.  This distinction follows from a bound over all allowed controls and an explicit ensemble-generating feedback.  Measurement implementation therefore enters the resource specification alongside the channel and the preparation task.  The present readouts vary geometry and record information jointly; separating their contributions at exactly matched information remains open.

\nocite{AnnbyAndersson2022FiniteBandwidth,BoutenVanHandelJames2007,Brown2025GaugeUnravellings,Chen1990SlantImmersions,ChenBaoChoi2024,KruhnerXu2023Density,Lecocq2021EfficientMeasurement,Saiphet2021DelayedFeedback,WisemanBouten2008Purification,LionsSznitman1984,Zvonkin1974,LiangAminiMason2022GHZ,PiccittoRussomannoRossini2022,BartheEtAl2025ContinuousGenerators,KobayashiYamamoto2019ControlLimit,OConnorMaGenoni2025FidelityBound}
\bibliography{refs}

\clearpage
\onecolumngrid
\hypersetup{pageanchor=false}
\setcounter{page}{1}
\renewcommand{\thepage}{S\arabic{page}}
\setcounter{section}{0}
\renewcommand{\thesection}{S\Roman{section}}
\renewcommand{\theHsection}{supp.\arabic{section}}
\setcounter{equation}{0}
\renewcommand{\theequation}{S\arabic{equation}}
\renewcommand{\theHequation}{supp.\arabic{equation}}
\setcounter{figure}{0}
\renewcommand{\thefigure}{S\arabic{figure}}
\renewcommand{\theHfigure}{supp.\arabic{figure}}
\setcounter{table}{0}
\renewcommand{\thetable}{S\arabic{table}}
\renewcommand{\theHtable}{supp.\arabic{table}}
\setcounter{lemma}{0}
\renewcommand{\thelemma}{S\arabic{lemma}}
\renewcommand{\theHlemma}{supp.\arabic{lemma}}
\setcounter{proposition}{0}
\renewcommand{\theproposition}{S\arabic{proposition}}
\renewcommand{\theHproposition}{supp.\arabic{proposition}}
\setcounter{corollary}{0}
\renewcommand{\thecorollary}{S\arabic{corollary}}
\renewcommand{\theHcorollary}{supp.\arabic{corollary}}
\hypersetup{pageanchor=true}

\begin{center}
{\large\bfseries Supplemental Material for ``Distinct Feedback-Strength Requirements for Quantum-State Ensemble Preparation under Channel-Equivalent Monitoring''\par}
\vspace{0.7em}
{Haitao Huang$^{1}$, Qinglin Zhao$^{1,2}$, and Xiaoyu Li$^{3}$\par}
\vspace{0.4em}
{\small\itshape
$^{1}$School of Computer Science and Engineering, Macau University of Science and Technology, Macau 999078, China\par
$^{2}$Zhuhai MUST Science and Technology Research Institute, Zhuhai, China\par
$^{3}$School of Information and Software Engineering, University of Electronic Science and Technology of China, Chengdu, China\par}
\vspace{0.3em}
{\small Corresponding author: Qinglin Zhao (\href{mailto:qlzhao@must.edu.mo}{qlzhao@must.edu.mo})\par}
\end{center}
\vspace{0.8em}
The material follows the proof of the main result.  First,
\hyperref[sec:S-shared]{the common control and record definitions} specify
the task.  The \hyperref[sec:S-multiqubit]{complete-Pauli multiqubit section}
then gives the model, radial reduction, all-controller comparison,
attaining ensemble law, and capture--maintenance estimates in one sequence.
The \hyperref[sec:S-qubit]{single-qubit specialization} treats the reflected
boundary and gives explicit crossover formulas.  The
\hyperref[sec:S-general-window]{geometric terminal-window principle}
and its realizations follow the main proof.  Nonuniform targets,
\hyperref[sec:S-efficiency]{finite efficiency}, and
\hyperref[sec:S-instruments]{finite-step instruments} state their additional
assumptions.  Finally,
\hyperref[sec:S-diagnostics]{numerical diagnostics and data provenance}
retain all reported checks and limitations.
Supplemental equations are numbered independently of the main text.

\paragraph*{Costs, errors, and empirical laws.}
$\C_\delta$ is the continuous-time infimum over all allowed controls;
$E_\delta^*(U,T)$ and $E_{d,\delta}^*(U,T)$ are the exact qubit and complete-basis finite-time optima for their uniform targets;
$L_\delta(U)$ is an analytic lower certificate;
$E_\delta^H(T)$ is the finite-time error of one specified controller;
$\widehat\mu_T$ is a finite empirical law; and
$\C_{\delta,h}$ is the discrete cost under the declared post-triple control slot.
The finite-time theorem identifies $E_\delta^*$ with one reference-process expectation; no empirical or finite-step object is equated with it.
Throughout, $\epsilon$ is the requested terminal accuracy, $\varepsilon$ is a
smoothing parameter, $\delta$ is the target-normal instrument weight, $\eta$
is the record efficiency, and $h$ is a finite instrument step.

\section{Common task, public records, and admissible feedback}
\label{sec:S-shared}

Work on a filtered probability space
$(\Omega,\mathcal F,(\mathcal F_t)_{0\le t\le T},\Prob)$ satisfying the usual conditions.
At time zero it carries a normalized state $\psi_0$ drawn from complex Haar measure and, optionally, a declared classical random seed $\xi$ independent of $\psi_0$ and of future measurement innovations.
For an $n$-qubit register let $d=2^n$ and let $\mathcal P_n^\circ$
index the $d^2-1$ nonidentity Hermitian Pauli strings.  Write
$\gamma_P$ for the rate of channel $P$ and
$C_{P,\delta}=e^{-i\phi_{P,\delta}}P$ for its dimensionless readout operator.
Complete-Pauli monitoring uses $\gamma_P=\Gamma/d^2$; the corresponding
radial scale is $\kappa=\Gamma/(2d)$.  For a single qubit these rates
coincide: $\gamma_P=\kappa=\Gamma/4$.
The public filtration is
\begin{equation*}
 \mathcal F_t^{\rm obs}=\sigma\{\psi_0,\xi,Y_{P,s}:P\in\mathcal P_n^\circ,0\le s\le t\}^{\rm aug},
\end{equation*}
where $Y_P$ are the observed records and ``aug'' denotes the usual completion and right-continuous augmentation.
The continuous records and innovations are related by
\begin{equation}
 \dd Y_{P,t}=\sqrt{\gamma_P}\,\tr[(C_{P,\delta}+C_{P,\delta}^\dagger)\rho_t]\dd t
 +\dd W_{P,t}.
\label{eq:S-record}
\end{equation}
At unit efficiency, the conditional state is reconstructed by
\begin{align}
 \dd\rho_t={}&\sum_{P\in\mathcal P_n^\circ}\gamma_P\D[P]\rho_t\dd t
 -i[H_t,\rho_t]\dd t\nonumber\\
 &+\sum_{P\in\mathcal P_n^\circ}\sqrt{\gamma_P}\,
 \Hc_{C_{P,\delta}}(\rho_t)\dd W_{P,t},
\label{eq:S-filter-common}
\end{align}
where $\D[L]\rho=L\rho L^\dagger-\{L^\dagger L,\rho\}/2$ and
$\Hc_C(\rho)=C\rho+\rho C^\dagger-\tr[(C+C^\dagger)\rho]\rho$.
A pure initial condition remains pure under this filter, so
$\rho_t=|\psi_t\rangle\langle\psi_t|$.

\begin{lemma}[public innovations]
\label{lem:S-innovations}
Under the physical measure of the unit-efficiency filter, the processes $W_P$ defined by Eq.~(\ref{eq:S-record}) are independent standard Brownian motions relative to $(\mathcal F_t^{\rm obs})$.
Every coefficient obtained by applying a Borel function to the conditional state and a predictable feedback is consequently predictable in that filtration.
\end{lemma}
\begin{proof}
The conditional-mean term in Eq.~(\ref{eq:S-record}) is the predictable compensator of the observed quadrature.
The innovations are therefore continuous $\mathcal F_t^{\rm obs}$-local martingales.
Independent output quadratures give
$[W_P,W_Q]_t=\delta_{PQ}t$.
The multidimensional L\'evy characterization then makes $(W_P)_{P\in\mathcal P_n^\circ}$ a standard Brownian motion in the augmented public filtration \cite{BoutenVanHandelJames2007}.
In particular its increments after time zero are independent of $\mathcal F_0^{\rm obs}$.
\end{proof}

For $U\ge0$, $\A_U$ is the set of $\mathcal F_t^{\rm obs}$-predictable Hermitian processes $H_t$ satisfying $\|H_t\|_{\rm op}\le U$ almost surely for Lebesgue-almost every $t$.
Monitoring remains active at its fixed strength throughout $[0,T]$, including the terminal interval; the model has no noise-free final correction slot.
The control acts through $H_t\dd t$; hence $A_t=\int_0^tH_s\dd s$ has bounded variation on the finite horizon.
The declared instrument includes the fixed preparation, measurement, and reset of each fresh ancilla used in its dilation.  The control class excludes Hamiltonians proportional to the contemporaneous increment $\dd Y_t$, additional system-reset or adaptive-ancilla channels, postselection, and hidden records.

For instrument label $\delta$, control $H\in\A_U$, and a fixed initial law $\mu_0$ independent of the later policy choice, write $\mu_T^{\delta,H}=\Law^{\delta,H}(\rho_T)$.
The ground metric and its Wasserstein lift are
\begin{align*}
 d_{\rm tr}(\rho,\sigma)&=\tfrac12\|\rho-\sigma\|_1,\\
 W_{1,\rm tr}(\mu,\nu)&=\inf_{\Pi\in\Gamma(\mu,\nu)}
 \int d_{\rm tr}(\rho,\sigma)\,\Pi(\dd\rho,\dd\sigma).
\end{align*}
The optimal peak cost is
\begin{equation}
 \C_\delta(\epsilon;\nu,\mu_0,T)
 =\inf\{U\ge0:\exists H\in\A_U,
 \ W_{1,\rm tr}(\mu_T^{\delta,H},\nu)\le\epsilon\}.
\label{eq:S-cost}
\end{equation}
For one specified controller define
\begin{equation}
 E_\delta^H(T;\nu):=W_{1,\rm tr}(\mu_T^{\delta,H},\nu).
\label{eq:S-policyerror}
\end{equation}
For the single-qubit uniform real-state target $\nu_R$ on $\mathbb{RP}^1$, define
\begin{equation}
 E_\delta^*(U,T):=\inf_{H\in\A_U}E_\delta^H(T;\nu_R).
\label{eq:S-finitevalue}
\end{equation}
For $N$ stored terminal states, $\widehat\mu_T=N^{-1}\sum_{i=1}^N\delta_{\rho_T^{(i)}}$ denotes the empirical law.
Neither $E_\delta^H$ nor a statistic of $\widehat\mu_T$ is the infimum in Eq.~(\ref{eq:S-cost}).
$E_\delta^*$ is the infimum over controllers; specified or learned controller errors are denoted separately.
The phrase ``all-feedback lower bound'' means that the inequality holds for every member of $\A_U$.
Equal record access below means that both controllers receive every public outcome at the same update times.
It does not mean equal state information: the information content of the record is one of the instrument properties being compared.

\subsection{Delivered object and verification}

One run delivers a classical--quantum pair: the terminal system and the public data
\begin{equation*}
 \mathcal R_T^{\rm pub}=(\psi_0,\xi,
 \{Y_{P,s}:P\in\mathcal P_n^\circ,0\le s\le T\},\delta,\mathcal P_H),
\end{equation*}
where $\mathcal P_H$ is the public feedback policy; if that policy is not supplied, the actual applied-control log must replace it.
The instrument calibration, measurement record, and applied control map
$\mathcal R_T^{\rm pub}$ to the conditional label $\rho_T$.
The law in Eq.~(\ref{eq:S-cost}) is the distribution of these record-conditioned labels over repeated runs.
A verifier can replay the controlled filter on all records.
Calibration of the reported labels must be tested across repeated runs, for example by record-dependent terminal measurements or by tomography within predeclared record bins followed by a conditional-likelihood test.
No protocol here assigns complete tomography to one system carrying a unique continuous record label.
Classical goodness-of-fit or optimal-transport estimators may then test the labelled empirical law against the announced target.
The record is therefore part of the delivered resource, rather than an unobserved mathematical device.

This distinction is operationally necessary here.
For the complex-Haar prior and the orthogonally invariant law $\nu_d$ on $\mathbb{RP}^{d-1}$,
\begin{equation*}
 \int\rho\,\mu_0(\dd\rho)=\int\rho\,\nu_d(\dd\rho)=I/d.
\end{equation*}
Consequently a receiver given only unlabelled, independently prepared terminal registers cannot distinguish these two ensembles from their one-copy average state.
No nonlinear conditional-state statistic is being presented as an ordinary observable without its instrument and record \cite{GaonaReyes2025UnravellingLimits}.
The free initial description and optional seed also mean that Eq.~(\ref{eq:S-cost}) measures monitored-ensemble control, not the production of irreducibly quantum randomness.

If $\nu$ is supported on a closed set $\M$, then every $\Pi\in\Gamma(\mu,\nu)$ obeys
\begin{equation*}
 \int d_{\rm tr}(\rho,\sigma)\Pi(\dd\rho,\dd\sigma)
 \ge\int d_{\rm tr}(\rho,\M)\mu(\dd\rho).
\end{equation*}
Taking the infimum gives
\begin{equation}
 W_{1,\rm tr}(\mu,\nu)\ge\E_\mu d_{\rm tr}(\rho,\M).
\label{eq:S-support}
\end{equation}

\section{Complete-Pauli multiqubit theorem on \texorpdfstring{$\mathbb{RP}^{d-1}$}{RP(d-1)}}
\label{sec:S-multiqubit}

The main theorem concerns $n\ge2$ system qubits; measurement ancillas are additional registers.
Use the common filter~(\ref{eq:S-filter-common}) with $d=2^n$ and $\gamma_P=\Gamma/d^2$.
Set
\begin{equation}
 L_P=\frac{\sqrt\Gamma}{d}P,
 \qquad \kappa=\frac{\Gamma}{2d}.
\label{eq:S-multi-L}
\end{equation}
Write $\mathcal S=\{P:P^{\mathsf T}=P\}$ and
$\mathcal A=\{P:P^{\mathsf T}=-P\}$.
The informational quadrature is used for $P\in\mathcal S$, while for
$P\in\mathcal A$ we use
\begin{equation}
 C_P=(\sqrt\delta-i\sqrt{1-\delta})L_P,
 \qquad 0\le\delta\le1.
\label{eq:S-multi-phase}
\end{equation}
The phase does not change $\D[L_P]$, and the Pauli twirl gives
\begin{equation}
 \sum_{P\in\mathcal P_n^\circ}\D[L_P]\rho
 =\Gamma\left(\frac{I}{d}\tr\rho-\rho\right)
\label{eq:S-multi-channel}
\end{equation}
for every $\delta$.
Thus the two endpoint instruments share one feedback-free depolarizing channel.

Here $C_P$ includes the coupling amplitude, unlike the dimensionless $C_{P,\delta}$ in Eq.~(\ref{eq:S-filter-common}).  Set $C_P=L_P$ on $\mathcal S$ and use Eq.~(\ref{eq:S-multi-phase}) on $\mathcal A$.
The public records and innovations are
\begin{align}
 \dd Y_{P,t}&=\tr[(C_P+C_P^\dagger)\rho_t]\dd t+\dd W_{P,t},\nonumber\\
 \mathcal F_t^{\rm obs,(d)}&=
 \sigma\{\psi_0,\xi,Y_{P,s}:P\in\mathcal P_n^\circ,0\le s\le t\}^{\rm aug}.
\label{eq:S-multi-record}
\end{align}
Under the physical unit-efficiency filter, the $d^2-1$ processes $W_P$ are independent standard Brownian motions in this augmented filtration.
Throughout this section, $\A_U$ means the Hermitian, $\mathcal F_t^{\rm obs,(d)}$-predictable controls with $\|H_t\|_{\rm op}\le U$; hence the multiqubit controller uses all and only the declared public records.

The instruments are channel matched, but they are not information matched.
For a local state perturbation $\rho_\theta=\rho+\theta X+o(\theta)$,
with $X=X^\dagger$ and $\tr X=0$, the drift of record $P$ and its
fixed-state, per-unit-time local classical Fisher quadratic form are
\begin{align}
 \mu_P(\rho)&=\frac{2\sqrt\Gamma}{d}\cos\phi_P\,\tr(P\rho),\nonumber\\
 \mathcal I_\rho[X]
 &=\frac{4\Gamma}{d^2}\sum_{P\ne I}\cos^2\phi_P\,[\tr(PX)]^2.
\label{eq:S-record-Fisher}
\end{align}
For the transpose split used here this becomes
\begin{equation}
 \mathcal I_{\rho,\delta}[X]
 =\frac{4\Gamma}{d^2}\left\{
 \sum_{P\in\mathcal S}[\tr(PX)]^2
 +\delta\sum_{P\in\mathcal A}[\tr(PX)]^2\right\}.
\label{eq:S-record-Fisher-delta}
\end{equation}
The orthogonality of Pauli strings shows that equality of the full quadratic
form for every $X$ forces the same weights and hence the same $\delta$ within
this one-parameter family.  Allowing opposite phase signs can preserve
Eq.~(\ref{eq:S-record-Fisher}) while reversing a tangent--normal cross
covariance, but it leaves the diagonal normal quadratic variation unchanged.
Accordingly, the theorem compares complete conditional instruments---record
information and backaction geometry jointly---and does not claim an
information-matched normal-rate separation.

The target is the full real-state family
\begin{equation}
 \M_d=\mathbb{RP}^{d-1}\subset\mathbb{CP}^{d-1},
 \qquad \nu_d=\text{its invariant probability law}.
\label{eq:S-multi-target}
\end{equation}
For a normalized pure state $\psi$, define
\begin{equation}
 z=\psi^{\mathsf T}\psi,\qquad
 q=\frac12\arccos|z|\in[0,\pi/4].
\label{eq:S-multi-q}
\end{equation}
After a global phase choice,
$\psi=\cos q\,e+i\sin q\,f$ with real orthonormal $e,f$.
The closest real ray is $[e]$ and
\begin{equation}
 d_{\rm tr}([\psi],\M_d)=\sin q.
\label{eq:S-multi-distance}
\end{equation}
For $n\ge2$, $\M_d$ has real dimension $d-1$, contains entangled states, and sits inside the full $2(d-1)$-dimensional space $\mathbb{CP}^{d-1}$ occupied by the prior below.

\begin{lemma}[Multiqubit radial reduction]
\label{lem:S-multi-radial}
Start from complex Haar measure on $\mathbb{CP}^{d-1}$ and use any public predictable Hamiltonian satisfying $\|H_t\|_{\rm op}\le U$.
On the interior, the complete filter induces
\begin{align}
 \dd q_t={}&[b_{d,\delta}(q_t)+v_t]\dd t
 +\sqrt{\kappa f_\delta(q_t)}\dd W_t,
 \qquad |v_t|\le U,\nonumber\\
 f_\delta(q)={}&\delta+(1-\delta)\sin^2(2q),\nonumber\\
 b_{d,\delta}(q)={}&\kappa\{\delta(d-2)\cot(2q)-\tan(2q)
 -(1-\delta)\sin(2q)\cos(2q)\}.
\label{eq:S-multi-radial}
\end{align}
The Haar radial density is
\begin{equation}
 p_{d,\rm H}(q)=2(d-1)\sin^{d-2}(2q)\cos(2q).
\label{eq:S-multi-Haar}
\end{equation}
The feedback
\begin{equation}
 H_N(\psi)=
 \frac{2U}{\sqrt{1-|z|^2}}
 \operatorname{Im}_{\rm entry}\!\left(
 \frac{\bar z}{|z|}\psi\psi^{\mathsf T}\right)
\label{eq:S-multi-HN}
\end{equation}
is smooth for $0<q<\pi/4$, obeys $\|H_N\|_{\rm op}=U$, and realizes $v_t=-U$.
\end{lemma}

\begin{proof}
The identity may be added to the Pauli basis because its conditional and unconditional fields vanish.
For Hilbert--Schmidt orthogonal bases of the real symmetric and purely imaginary antisymmetric subspaces,
\begin{align}
 \sum_{P\in\mathcal S}P_{ij}P_{kl}
 &=\frac d2(\delta_{il}\delta_{jk}+\delta_{ik}\delta_{jl}),\nonumber\\
 \sum_{P\in\mathcal A}P_{ij}P_{kl}
 &=\frac d2(\delta_{il}\delta_{jk}-\delta_{ik}\delta_{jl}).
\label{eq:S-multi-completeness}
\end{align}
Apply It\^o's formula first to $z=\psi^{\mathsf T}\psi$ in fixed complex coordinates and only then choose the pointwise phase with $z=|z|$.
The drift and quadratic covariations are
\begin{align}
 \operatorname{drift}(z)&=-2\kappa(d-1)
 [\delta-i\sqrt{\delta(1-\delta)}]|z|,\nonumber\\
 \frac{\dd\langle\operatorname{Re}z\rangle}{\dd t}
 &=4\kappa(1-|z|^2)[1-(1-\delta)|z|^2],\nonumber\\
 \frac{\dd\langle\operatorname{Im}z\rangle}{\dd t}
 &=4\kappa(1-|z|^2),\qquad
 \frac{\dd\langle\operatorname{Re}z,\operatorname{Im}z\rangle}{\dd t}=0.
\label{eq:S-multi-zmoments}
\end{align}
The transformations $s=|z|$ and $q=\arccos(s)/2$ give Eq.~(\ref{eq:S-multi-radial}).
The Hamiltonian part has Fubini--Study speed at most $\|H_t\|_{\rm op}$, hence $|v_t|\le U$.
In the displayed real frame, Eq.~(\ref{eq:S-multi-HN}) becomes
$U(ef^{\mathsf T}+fe^{\mathsf T})$ and direct differentiation gives $\dot q=-U$.
Finally, a $2\times2$ real Wishart reduction of the real and imaginary Gaussian parts of a Haar vector gives Eq.~(\ref{eq:S-multi-Haar}).
\end{proof}

For $0<\delta\le1$, define
\begin{equation}
 J_\delta(q)=\int_0^q\frac{\dd r}{f_\delta(r)}
 =\frac1{2\sqrt\delta}
 \arctan\frac{\tan(2q)}{\sqrt\delta}
\label{eq:S-multi-J}
\end{equation}
and the normalized density
\begin{equation}
 \pi_{d,\delta,U}(q)\propto
 \sin^{d-2}(2q)\cos(2q)f_\delta(q)^{-(d+2)/2}
 \exp[-2UJ_\delta(q)/\kappa].
\label{eq:S-multi-pi}
\end{equation}
Let
\begin{equation}
 L_{d,\delta}(U)=\int_0^{\pi/4}\sin q\,
 \pi_{d,\delta,U}(q)\dd q.
\label{eq:S-multi-Lbound}
\end{equation}

\begin{lemma}[Uniform truncated boundary-layer moments]
\label{lem:S-multi-truncated-moments}
Fix $d\ge2$, let $p\in\{1,2\}$, and set
$A_d=\max\{1,d-2\}$, $\ell=\kappa\delta/U$,
$r=U/(\kappa\sqrt\delta)$, and $q_*=A_d\ell$.
There are finite constants $r_{d,p}$ and $B_{d,p}$, depending only on
$d$ and $p$, such that, whenever $0<\delta\le1$, $r\ge r_{d,p}$,
and $q_*\le\pi/4$,
\begin{equation}
 \int_{q_*}^{\pi/4}\sin^p q\,
 \pi_{d,\delta,U}(\dd q\mid q\ge q_*)
 \le B_{d,p}\ell^p.
\label{eq:S-multi-truncated-moments}
\end{equation}
\end{lemma}
\begin{proof}
Use the exact change of variables
\begin{equation*}
 \theta=\arctan\!\frac{\tan(2q)}{\sqrt\delta},\qquad
 q_\delta(\theta)=\frac12\arctan(\sqrt\delta\tan\theta),
\end{equation*}
with the continuous endpoint value $q_\delta(\pi/2)=\pi/4$.
Writing $D_\delta(\theta)=\cos^2\theta+\delta\sin^2\theta$, the
unnormalized measure in Eq.~(\ref{eq:S-multi-pi}) becomes, up to a
$q$-independent factor,
\begin{equation}
 g_{d,\delta}(\theta)e^{-r\theta}\dd\theta,\qquad
 g_{d,\delta}=\sin^{d-2}\theta\cos\theta\sqrt{D_\delta(\theta)}.
\label{eq:S-multi-theta-density}
\end{equation}
Let $\theta_*$ be the image of $q_*$.  For all sufficiently large
$r$ (depending only on $d$), elementary bounds for $\tan$ and
$\arctan$ give
$A_d/r\le\theta_*\le4A_d/r$ and
$\theta_*+r^{-1}\le\pi/4$.  Therefore the truncated normalizer obeys
\begin{equation}
 \int_{\theta_*}^{\pi/2}g_{d,\delta}(\theta)e^{-r\theta}\dd\theta
 \ge c_d r^{-(d-1)},
\label{eq:S-multi-truncated-normalizer}
\end{equation}
because on $[\theta_*,\theta_*+r^{-1}]$ one has
$g_{d,\delta}(\theta)\ge c_d'\theta^{d-2}$ and $r\theta=O_d(1)$.

On $[\theta_*,\pi/4]$, $\sin q_\delta(\theta)\le
q_\delta(\theta)\le\sqrt\delta\,\theta$ and
$g_{d,\delta}(\theta)\le\theta^{d-2}$.  Hence its $p$th-moment
numerator is at most
$C_{d,p}\delta^{p/2}r^{-(d-1+p)}$.
On the remaining interval $\theta\in[\pi/4,\pi/2]$ put
$s=\cos\theta$.  Since $\sin\theta\ge1/\sqrt2$ and
$|\dd\theta|=|\dd s|/\sin\theta$, the bounds
\begin{equation*}
 \sin q_\delta(\theta)\le
 C\min\{1,\sqrt\delta/s\},\qquad
 g_{d,\delta}(\theta)|\dd\theta|
 \le C s\sqrt{s^2+\delta}|\dd s|
\end{equation*}
show, after splitting at $s=\sqrt\delta$, that the outer numerator is
at most $C_{d,p}\delta^{p/2}e^{-\pi r/4}$.
Dividing by Eq.~(\ref{eq:S-multi-truncated-normalizer}) and absorbing
$r^{d-1+p}e^{-\pi r/4}$ into the constant proves
Eq.~(\ref{eq:S-multi-truncated-moments}) uniformly in
$0<\delta\le1$.
\end{proof}

For the uniform target define
\begin{equation}
 E_{d,\delta}^*(U,T)=\inf_{H\in\A_U}
 W_{1,\mathrm{tr}}(\Law^{\delta,H}[\rho_T],\nu_d).
\label{eq:S-multi-value}
\end{equation}
Let $R_0$ have density~(\ref{eq:S-multi-Haar}) and, on the same scalar state space, set
\begin{equation}
 \dd R_t=[b_{d,\delta}(R_t)-U]\dd t
 +\sqrt{\kappa f_\delta(R_t)}\dd W_t.
\label{eq:S-multi-reference}
\end{equation}
For $\delta>0$, zero is reflected when $d=2$ and is inaccessible when $d\ge3$; the upper endpoint is inaccessible for every $d\ge2$.  For $\delta=0$, solve to the first hit of zero and keep $R$ there afterward.

\begin{theorem}[Finite-time complete-Pauli separation]
\label{thm:S-multiqubit}
Fix $n\ge2$, $d=2^n$, $T,\Gamma>0$, the complex-Haar prior, the unit-efficiency public record, and the peak class $\A_U$.  For every $\delta\in[0,1]$ and $U\ge0$,
\begin{equation}
 E_{d,\delta}^*(U,T)=\E\sin R_T.
\label{eq:S-multi-finite-exact}
\end{equation}
For every $0<\delta\le1$, every target law $\nu$ supported on $\M_d$, and every $H\in\A_U$,
\begin{equation}
 \Wone(\Law^{\delta,H}[\rho_T],\nu)
 \ge L_{d,\delta}(U).
\label{eq:S-multi-lower}
\end{equation}
Uniformly for $r=U/(\kappa\sqrt\delta)\to\infty$,
\begin{equation}
 L_{d,\delta}(U)=
 \frac{\kappa\delta(d-1)}{2U}[1+O_d(r^{-2})].
\label{eq:S-multi-lower-asymp}
\end{equation}
Furthermore $E_{d,\delta}^*(U,T)\to L_{d,\delta}(U)$ as $T\to\infty$.

For the uniform target, let $A_d=\max\{1,d-2\}$.  For each fixed $0<t_c<T$ there are finite $r_d,B_d$, depending only on $d$, such that the public controller~(\ref{eq:S-multi-HN}) obeys
\begin{align}
 \Wone(\Law^{\delta,H_N}[\rho_T],\nu_d)
 \le{}&\exp\!\left[-\frac{(Ut_c/2-\pi/4)_+^2}{2\kappa t_c}\right]
 +B_d\frac{\kappa\delta}{U}
\label{eq:S-multi-upper}
\end{align}
when $r\ge r_d$ and $U\ge4A_d\kappa\delta/\pi$.  Consequently, for fixed $d,T,\kappa$ and $\delta>0$,
\begin{equation}
 \C_{d,\delta}(\epsilon;\nu_d)
 =\Theta_{d,T,\kappa,\delta}(1/\epsilon).
\label{eq:S-multi-fixedcost}
\end{equation}
At $\delta=0$, stopping $H_N$ at the first target hit gives
\begin{equation}
 \C_{d,0}(\epsilon;\nu_d)
 \le\frac{\pi}{4T}+\sqrt{\frac{2\kappa}{T}\log\frac1\epsilon}.
\label{eq:S-multi-tangent}
\end{equation}
No matching lower asymptotic is asserted for Eq.~(\ref{eq:S-multi-tangent}).
\end{theorem}

\begin{proof}
For $\delta>0$, the scale density of the uncontrolled radial diffusion is, up to a positive constant,
\begin{equation}
 \mathfrak s'_0(q)=
 \frac{f_\delta(q)^{d/2}}{\sin^{d-2}(2q)\cos(2q)}.
\label{eq:S-multi-scale}
\end{equation}
It behaves as $q^{-(d-2)}$ at zero and as $(\pi/4-q)^{-1}$ at the upper endpoint.  Thus zero is regular and reflected for $d=2$, logarithmically inaccessible for $d=3$, and power-law inaccessible for $d\ge4$; the upper endpoint is inaccessible in all cases.  A bounded radial drift changes the scale density only by a locally bounded positive factor, equivalently by a finite-horizon Girsanov density when $\delta>0$.

Fix an arbitrary $H\in\A_U$.  Project its actual public innovations onto the normalized radial martingale direction; L\'evy's characterization gives a Brownian motion in that controller's public filtration.  Drive Eq.~(\ref{eq:S-multi-reference}) from the same random value $R_0=q_0$ with this Brownian motion.  Put $D_t=R_t-q_t$ and $a(q)=\sqrt{\kappa f_\delta(q)}$.  Let $\tau_m$ stop either process on leaving $[1/m,\pi/4-1/m]$.  On this interval $b_{d,\delta}$ and $a$ are Lipschitz, and the occupation-density formula gives $L^0_{t\wedge\tau_m}(D)=0$: the quadratic-variation density of $D$ is at most a constant times $D^2$.  Tanaka's formula and $v_t\ge-U$ therefore give
\begin{align*}
 D_{t\wedge\tau_m}^+
 &\le C_m\int_0^{t\wedge\tau_m}D_s^+\dd s+M_t^{(m)},\\
 M_t^{(m)}
 &=\int_0^{t\wedge\tau_m}\mathbf1_{\{D_s>0\}}
       [a(R_s)-a(q_s)]\dd W_s.
\end{align*}
The stopped stochastic integral is square integrable and has zero expectation; sharing a Brownian motion does not cancel its state-dependent integrand.  Thus
\begin{equation*}
 \E D_{t\wedge\tau_m}^+
 \le C_m\int_0^t\E D_{s\wedge\tau_m}^+\dd s.
\end{equation*}
Gr\"onwall and $D_0=0$ give $\E D_{t\wedge\tau_m}^+=0$.  Applying this at rational times and using continuity gives the stopped pathwise order.  Endpoint inaccessibility removes the localization for $\delta>0$.  For $\delta=0$ the order holds to the first lower-boundary hit; if $q$ hits first, continuity forces $R$ to hit simultaneously.  Thereafter $R=0\le q$.  Weak uniqueness of the scalar equation identifies the reference marginal without assuming that different controllers share an external Brownian motion.  Hence
\begin{equation}
 E_{d,\delta}^H(T;\nu_d)\ge\E\sin q_T\ge\E\sin R_T.
\label{eq:S-multi-finite-lower}
\end{equation}

For $d\ge4$ and $\delta>0$, Eq.~(\ref{eq:S-multi-HN}) is smooth along every interior finite path, is public, has norm $U$, and realizes $v=-U$; standard localization therefore gives a strong closed-loop solution with radial law~(\ref{eq:S-multi-reference}).  At $\delta=0$ the same statement holds to the target hit, after which the controller is switched off and the tangent measurement keeps the target invariant.  The optional $d=2$ specialization has a reflected lower boundary; its separate smooth-approximation argument is given in Theorem~\ref{thm:S-finite} and is not needed for the present $n\ge2$ result.

Conjugation by $O\in O(d)$ acts by separate orthogonal transformations on the real-symmetric and imaginary-antisymmetric monitor spaces, and $H_N(O\psi)=OH_N(\psi)O^{\mathsf T}$.  The Haar prior and whole closed-loop law are therefore $O(d)$ invariant.  Nearest real-state projection is equivariant off the zero-mass cut locus and has the unique invariant law $\nu_d$.  Coupling each trajectory to its projection costs $\sin R_T$ and proves the reverse inequality in Eq.~(\ref{eq:S-multi-finite-exact}).

For the stationary certificate, density~(\ref{eq:S-multi-pi}) is the zero-current invariant law of the maximal-inward process.  Its likelihood ratio with respect to~(\ref{eq:S-multi-Haar}) decreases because $f_\delta$ and $J_\delta$ increase.  A quantile coupling therefore initializes a stationary $Q_0\le q_0$; the same filtration-by-filtration comparison proves $Q_T\le q_T$ for every history-dependent controller and gives Eq.~(\ref{eq:S-multi-lower}).  For $\delta>0$, nondegenerate interior noise makes the scalar diffusion irreducible, while the reflected/inaccessible boundary classification and finite normalized zero-current speed density make it positive recurrent with a unique invariant law.  Standard one-dimensional diffusion convergence then gives the stated $T\to\infty$ limit.  The change $q=\kappa\delta x/(2U)$ and endpoint Laplace expansion produce Eq.~(\ref{eq:S-multi-lower-asymp}).

For the upper estimate put
\begin{equation}
 q_*:=A_d\frac{\kappa\delta}{U},
 \qquad \tau_*:=\inf\{t:q_t\le q_*\}.
\label{eq:S-multi-qstar}
\end{equation}
The condition on $U$ ensures $q_*\le\pi/4$.  Since
$\cot(2q)\le(2q)^{-1}$ and the remaining two terms in
$b_{d,\delta}$ are nonpositive,
\begin{equation}
 b_{d,\delta}(q)
 \le \frac{\kappa\delta(d-2)}{2q}
 \le \frac U2,
 \qquad q\ge q_*.
\label{eq:S-multi-inward-drift}
\end{equation}
On $\{\tau_*>t_c\}$ the martingale stopped at $\tau_*$ must therefore
exceed $Ut_c/2-\pi/4$.  Its quadratic variation is at most
$\kappa t_c$, because $f_\delta\le1$.  The exponential martingale
inequality gives
\begin{equation}
 \Prob(\tau_*>t_c)
 \le \exp\!\left[-\frac{(Ut_c/2-\pi/4)_+^2}
 {2\kappa t_c}\right].
\label{eq:S-multi-capture}
\end{equation}

It remains to control every post-hit duration, rather than only an
equilibrium limit.  Start the same maximal-inward scalar diffusion at
$q\le q_*$ and compare it with the diffusion having identical
coefficients and an upward reflecting boundary at $q_*$.  The latter
has invariant law
$\bar\pi_{d,\delta,U}^{(*)}
=\pi_{d,\delta,U}(\cdot\mid q\ge q_*)$.
Pathwise one-dimensional reflection comparison, first from the hit
point to $q_*$ and then from $q_*$ to a stationary reflected copy,
yields
\begin{equation}
 \sup_{s\ge0}\sup_{q_{\rm in}\le q_*}
 \E_{q_{\rm in}}\sin q_s
 \le \int_{q_*}^{\pi/4}\sin q\,
 \bar\pi_{d,\delta,U}^{(*)}(\dd q).
\label{eq:S-multi-truncated-order}
\end{equation}
Taking $p=1$ in Lemma~\ref{lem:S-multi-truncated-moments} gives finite
constants depending only on $d$ such that
\begin{equation}
 \int_{q_*}^{\pi/4}\sin q\,
 \bar\pi_{d,\delta,U}^{(*)}(\dd q)
 \le B_d\frac{\kappa\delta}{U},
 \qquad r\ge r_d,
\label{eq:S-multi-truncated-moment}
\end{equation}
uniformly over $0<\delta\le1$.  Conditioning at $\tau_*$, assigning
unit transport cost to the failure event, and using the $O(d)$
projection coupling combine
Eqs.~(\ref{eq:S-multi-capture})--(\ref{eq:S-multi-truncated-moment})
to prove Eq.~(\ref{eq:S-multi-upper}).  At $\delta=0$,
$b_{d,0}\le0$ and the stopped martingale bracket is at most $\kappa t$;
taking the capture window to be $T$ gives
Eq.~(\ref{eq:S-multi-tangent}).
\end{proof}

\begin{corollary}[Certified same-budget separation]
\label{cor:S-budget-window}
Fix $0<\delta\le1$ and $0<\epsilon<L_{d,\delta}(0)$, and define
\begin{align}
 U_{\rm N}^{\rm lb}(\epsilon;\delta)
 &=\inf\{U\ge0:L_{d,\delta}(U)\le\epsilon\},\nonumber\\
 U_{\rm T}^{\rm ub}(\epsilon)
 &=\frac{\pi}{4T}
 +\sqrt{\frac{2\kappa}{T}\log\frac1\epsilon}.
\label{eq:S-budget-thresholds}
\end{align}
The first threshold is the unique solution of
$L_{d,\delta}(U)=\epsilon$.  If
$U_{\rm T}^{\rm ub}<U_{\rm N}^{\rm lb}$, then every
\begin{equation}
 U\in[U_{\rm T}^{\rm ub},U_{\rm N}^{\rm lb})
\label{eq:S-budget-window}
\end{equation}
admits the stopped first-hit tangent controller with terminal error at
most $\epsilon$, while every admissible controller for the
$\delta$ instrument has terminal error strictly greater than
$\epsilon$.
\end{corollary}

\begin{proof}
In Eq.~(\ref{eq:S-multi-pi}) write the $U$-independent positive factor
as $a_{d,\delta}(q)$.  Differentiation under the integral gives
\begin{equation}
 \frac{\dd}{\dd U}L_{d,\delta}(U)
 =-\frac{2}{\kappa}
 \operatorname{Cov}_{\pi_{d,\delta,U}}
 [\sin q,J_\delta(q)]<0.
\label{eq:S-L-monotone}
\end{equation}
Both functions in the covariance are strictly increasing on
$(0,\pi/4)$, and the invariant density is positive there; the strict
sign follows from the independent-copy covariance identity.  Dominated
convergence gives continuity, and
Eq.~(\ref{eq:S-multi-lower-asymp}) gives $L_{d,\delta}(U)\to0$.
Thus the normal threshold exists and is unique.  For every
$U<U_{\rm N}^{\rm lb}$, Eq.~(\ref{eq:S-multi-lower}) places every
normal-instrument terminal law above $\epsilon$.  For every
$U\ge U_{\rm T}^{\rm ub}$, the capture estimate used in
Eq.~(\ref{eq:S-multi-tangent}) is at most $\epsilon$; after the first
hit, the tangent instrument preserves the target and $O(d)$ covariance
gives the target law.  Combining the two one-sided statements proves
Eq.~(\ref{eq:S-budget-window}).
\end{proof}

\begin{corollary}[Wasserstein-perturbed target]
\label{cor:S-thick-target}
Fix $0<\delta\le1$ and let $\nu_d^\sigma$ be any probability law satisfying
$\Wone(\nu_d^\sigma,\nu_d)\le\sigma$.  Every $H\in\A_U$ obeys
\begin{equation}
 \Wone(\Law^{\delta,H}[\rho_T],\nu_d^\sigma)
 \ge [L_{d,\delta}(U)-\sigma]_+.
\label{eq:S-thick-target}
\end{equation}
Consequently, terminal error at most $\epsilon$ requires
$L_{d,\delta}(U)\le\epsilon+\sigma$.  If
$\epsilon+\sigma<L_{d,\delta}(0)$, the certified necessary budget is
$U\ge U_{\rm N}^{\rm lb}(\epsilon+\sigma;\delta)$, with the inverse defined
in Eq.~(\ref{eq:S-budget-thresholds}).  At fixed $d,\kappa,\delta>0$,
\begin{equation}
 U_{\rm N}^{\rm lb}(\epsilon+\sigma;\delta)
 =\frac{\kappa\delta(d-1)}{2(\epsilon+\sigma)}
 \left[1+O_d\!\left(\frac{(\epsilon+\sigma)^2}{\delta}\right)\right]
\label{eq:S-thick-target-asymp}
\end{equation}
as $\epsilon+\sigma\downarrow0$.  For fixed $\sigma>0$ this statement does
not imply a $1/\epsilon$ divergence.
If $0\le\sigma<\epsilon$, $\epsilon+\sigma<L_{d,\delta}(0)$, and the interval
\begin{equation}
 U_{\rm T}^{\rm ub}(\epsilon-\sigma)\le U
 <U_{\rm N}^{\rm lb}(\epsilon+\sigma;\delta)
\label{eq:S-perturbed-budget-window}
\end{equation}
is nonempty, the tangent capture construction has error at most $\epsilon$ to this same $\nu_d^\sigma$, whereas every allowed normal-instrument controller has error greater than $\epsilon$.  Here $\sigma$ bounds a distance between probability laws, not the displacement of each sample from $\M_d$.
\end{corollary}
\begin{proof}
The Wasserstein triangle inequality and
Eq.~(\ref{eq:S-multi-lower}) give
\begin{align*}
 \Wone(\Law^{\delta,H}[\rho_T],\nu_d^\sigma)
 &\ge \Wone(\Law^{\delta,H}[\rho_T],\nu_d)
      -\Wone(\nu_d^\sigma,\nu_d)\nonumber\\
 &\ge L_{d,\delta}(U)-\sigma.
\end{align*}
Nonnegativity yields Eq.~(\ref{eq:S-thick-target}).  Strict monotonicity of
$L_{d,\delta}$ gives the inverse statement.  Inverting
Eq.~(\ref{eq:S-multi-lower-asymp}) gives
Eq.~(\ref{eq:S-thick-target-asymp}); the exact inverse is retained when
$\delta$ varies with the requested accuracy.
For Eq.~(\ref{eq:S-perturbed-budget-window}), the tangent construction at tolerance $\epsilon-\sigma$ and the triangle inequality give error at most $(\epsilon-\sigma)+\sigma=\epsilon$ to $\nu_d^\sigma$.  Strict monotonicity gives $L_{d,\delta}(U)>\epsilon+\sigma$ on the normal side, so Eq.~(\ref{eq:S-thick-target}) gives error greater than $\epsilon$.  Both conclusions refer to the same target law; no pointwise tube-support assumption is used.
\end{proof}

The cost inversion uses strict margins rather than an unattained equality at the infimum.  The classes $\A_U$ are nested, hence $E_{d,\delta}^*$ is nonincreasing.  For an upper bound choose $U$ so that the analytic controller estimate is below $\epsilon-\eta$; when an approximating feedback is needed, weak terminal-law convergence supplies an admissible smooth policy within the remaining $\eta$.  For the lower bound apply Eq.~(\ref{eq:S-multi-lower}) to every feasible policy before taking the infimum over its peak.  This proves the asymptotic cost statements without asserting attainment at a budget where an infimum merely equals $\epsilon$.

The theorem also gives the multiqubit calibration statement.
With $S(\epsilon)=\sqrt{\log(1/\epsilon)}$,
\begin{equation}
 \C_{d,\delta(\epsilon)}(\epsilon;\nu_d)=O(S)
 \quad\Longleftrightarrow\quad
 \delta(\epsilon)=O[\epsilon S(\epsilon)].
\label{eq:S-multi-window}
\end{equation}
The necessity follows from Eq.~(\ref{eq:S-multi-lower-asymp}) in the resulting large-$r$ regime, and sufficiency from Eq.~(\ref{eq:S-multi-upper}).
Equation~(\ref{eq:S-multi-window}) is a budget-class statement; it is not a hidden claim that the tangent inverse in Eq.~(\ref{eq:S-multi-tangent}) has a matching lower asymptotic.

For fixed total depolarizing rate $\Gamma$, the leading lower coefficient is
$\Gamma\delta(d-1)/(4d)$ and approaches $\Gamma\delta/4$ as $d$ grows.
This does not make the construction scalable: it uses $4^n-1$ monitored Pauli channels, complete record-conditioned state reconstruction, and unrestricted global Hamiltonians.
The peak $U$ is therefore one resource coordinate, not a proxy for total hardware complexity.

\section{Single-qubit specialization: target geometry and Haar law}
\label{sec:S-qubit}

This section and the next five specialize to one system qubit, with three
Pauli channels of equal rate $\kappa=\Gamma/4$ and uniform target
$\nu_R$ on $\mathbb{RP}^1$.  They establish the reflected-boundary
construction and the explicit crossover formulas used by the later
single-site and finite-clock diagnostics.  The multiqubit proof above
does not require these additional boundary arguments.

Let $\psi\in\mathbb C^2$ be normalized and $z=\psi^{\mathsf T}\psi$.
After multiplying $\psi$ by a global phase, take $z=|z|\ge0$ and write $\psi=a+ib$ with real vectors $a,b$.
The imaginary part of $\psi^{\mathsf T}\psi$ gives $a^{\mathsf T}b=0$, while
\begin{equation*}
 \|a\|^2+\|b\|^2=1,\qquad
 \|a\|^2-\|b\|^2=|z|.
\end{equation*}
Hence for real orthonormal vectors $e,f$,
\begin{equation}
 \psi=\cos q\,e+i\sin q\,f,\qquad
 |z|=\cos(2q),\qquad 0\le q\le\pi/4.
\label{eq:S-normalform}
\end{equation}
For a normalized real vector $r$, the overlap satisfies
$|\langle r,\psi\rangle|^2\le\cos^2q$, with equality at $r=e$.
Therefore
\begin{equation}
 d_{\rm FS}([\psi],\mathbb{RP}^1)=q,\qquad
 d_{\rm tr}([\psi],\mathbb{RP}^1)=\sqrt{1-\cos^2q}=\sin q.
\label{eq:S-metric}
\end{equation}
This proves the metric identities used in the support bound directly.

For a pure qubit with Bloch vector $(x,y,z_B)$, one checks
$|\psi^{\mathsf T}\psi|^2=1-y^2$.
Thus $|y|=\sin(2q)$.
Complex Haar measure is uniform on the Bloch sphere, so $y$ is uniform on $[-1,1]$.
Adding the two signs and changing variables gives
\begin{equation}
 p_H(q)=2\cos(2q),\qquad
 \int_0^{\pi/4}p_H(q)\dd q=1.
\label{eq:S-haar}
\end{equation}

For later use, any two pure states satisfy
\begin{equation*}
 d_{\rm tr}([\psi],[\varphi])=\sin d_{\rm FS}([\psi],[\varphi]),
 \qquad \frac{2}{\pi}d_{\rm FS}\le d_{\rm tr}\le d_{\rm FS}.
\end{equation*}
The two corresponding Wasserstein metrics therefore have identical asymptotic cost classes.

\section{From the Pauli filter to the radial diffusion}

Let $r=(x,y,z_B)$ be the pure-state Bloch vector.
For the $X$ and $Z$ informational channels and the interpolated $Y$ readout
$e^{-i\phi_\delta}=\sqrt\delta-i\sqrt{1-\delta}$, the uncontrolled $y$ coordinate obeys
\begin{equation}
 \dd y=-4\kappa y\dd t
 -2\sqrt\kappa\,y(x\dd W_X+z_B\dd W_Z)
 +2\sqrt{\kappa\delta}(1-y^2)\dd W_Y.
\label{eq:S-y}
\end{equation}
The random-unitary component of the $Y$ innovation rotates the $x$--$z_B$ plane and has no $y$ component.
Using $x^2+z_B^2=1-y^2$, Eq.~(\ref{eq:S-y}) has quadratic variation
\begin{equation}
 \dd[y]_t=4\kappa(1-y^2)[\delta+(1-\delta)y^2]\dd t.
\label{eq:S-yqv}
\end{equation}

Away from $y=0$, set $q=\tfrac12\arcsin|y|$ and apply It\^o's formula.
The diffusion coefficient becomes
\begin{equation*}
 \frac{1}{4(1-y^2)}\frac{\dd[y]_t}{\dd t}
 =\kappa[\delta+(1-\delta)\sin^2(2q)]
 =\kappa f_\delta(q),
\end{equation*}
while the drift is
\begin{align}
 b_\delta(q)
 &=\frac{-2\kappa|y|}{\sqrt{1-y^2}}
 +\frac{\kappa|y|}{\sqrt{1-y^2}}
 [\delta+(1-\delta)y^2]\nonumber\\
 &=-\kappa\tan(2q)-\kappa(1-\delta)\sin(2q)\cos(2q).
\label{eq:S-b}
\end{align}
\begin{lemma}[radial semimartingale under admissible feedback]
\label{lem:S-radial}
For every $H\in\A_U$, the target distance has a continuous adapted version satisfying
\begin{equation}
 \dd q_t=[b_\delta(q_t)+v_t]\dd t
 +\sqrt{\kappa f_\delta(q_t)}\dd W_t+\dd\ell_t^0,
 \qquad |v_t|\le U,
\label{eq:S-radial}
\end{equation}
where $v$ is predictable, $W$ is Brownian in the public filtration, and $\ell^0$ is nondecreasing and supported on $\{q=0\}$.
For $\delta=0$, the local time vanishes and zero is absorbing after the Hamiltonian is stopped.
\end{lemma}
\begin{proof}
Write the radial martingale on the interior as
$\dd M_q=\sum_Pg_{P,t}\dd W_{P,t}$, where
$\sum_Pg_{P,t}^2=\kappa f_\delta(q_t)$.
Lemma~\ref{lem:S-innovations} and adaptation of the conditional state make the $g_P$ predictable.
Where $f_\delta>0$, set
$\beta_{P,t}=g_{P,t}/\sqrt{\kappa f_\delta(q_t)}$; where it vanishes, choose any predictable unit vector $\beta_t$.
Then
$W_t=\sum_P\int_0^t\beta_{P,s}\dd W_{P,s}$
is a continuous public-filtration local martingale with $[W]_t=t$, hence is Brownian by L\'evy's characterization.

Use the symmetric local-time convention
\begin{equation*}
 |y_t|=|y_0|+\int_0^t\operatorname{sgn}(y_s)\dd y_s+L_t^0(y).
\end{equation*}
Applying the semimartingale It\^o--Tanaka formula first to $|y|$ and then to $q=\tfrac12\arcsin|y|$ gives the drift in Eq.~(\ref{eq:S-b}) and
$\ell_t^0=\tfrac12L_t^0(y)$.
Finite-variation Hamiltonian motion contributes an absolutely continuous predictable term $v_t\dd t$ away from the measure-zero crossing times.
The instantaneous Fubini--Study speed is $\Delta_\psi H$, and the metric derivative of distance to a closed set is no larger than the ambient speed, so
\begin{equation}
 |v_t|\le \Delta_{\psi_t}H_t
 \le\sqrt{\langle H_t^2\rangle}\le\|H_t\|_{\rm op}\le U.
\label{eq:S-speed}
\end{equation}
The finite-variation term does not contribute to local time.
For $\delta>0$, nonzero quadratic variation at zero gives the nonsticky Skorokhod reflection.
For $\delta=0$, every measurement martingale coefficient in the normal coordinate is proportional to $y$, so the occupation-density formula gives $L_t^0(y)=0$; with $H=0$, both normal drift and noise vanish at zero.
\end{proof}
The speed bound is saturated pointwise away from the target and cut locus.
In the normal form (\ref{eq:S-normalform}), define
\begin{equation}
 H_N=U(ef^{\mathsf T}+fe^{\mathsf T}).
\label{eq:S-Hnormal}
\end{equation}
Then $H_N$ is Hermitian, $\|H_N\|_{\rm op}=U$, and Schr\"odinger evolution gives $\dot q=-U$.
Without choosing $e,f$, the same operator is
\begin{equation}
 H_N(\psi)=\frac{2U}{\sqrt{1-|z|^2}}
 \operatorname{Im}_{\rm entry}
 \left(\frac{\bar z}{|z|}\psi\psi^{\mathsf T}\right),
\label{eq:S-globalH}
\end{equation}
where $\operatorname{Im}_{\rm entry}$ is the real matrix of entrywise imaginary parts.
Substituting Eq.~(\ref{eq:S-normalform}) into Eq.~(\ref{eq:S-globalH}) returns Eq.~(\ref{eq:S-Hnormal}), proving Hermiticity, norm, and velocity.
It is invariant under global phase and satisfies
$H_N(O\psi)=OH_N(\psi)O^{\mathsf T}$ for every real orthogonal $O$.
The formula is needed only for $0<|z|<1$; the complex-Haar prior gives the exceptional upper endpoint $z=0$ zero mass.
On the open domain it is a Borel function of the conditional ray.
For $\delta=0$ it is stopped permanently at the first hit $|z|=1$.
For $\delta>0$ define $H_N=0$ on $|z|=1$ and use Eq.~(\ref{eq:S-globalH}) whenever $q>0$.
The reflected process is nonsticky because its boundary quadratic variation is nonzero, so the value on $q=0$ affects no Lebesgue-time drift and the radial velocity equals $-U$ almost everywhere off the boundary.
The inaccessible upper endpoint may be assigned any bounded value without changing the law.
The next section uses this Borel selector only to identify a unique limiting weak law.  Public admissibility for $\delta>0$ is obtained with smooth covariant approximants, so no strong-solution claim for the discontinuous selector is needed.

\section{Exact finite-horizon optimum for the uniform target}

Let $R_0=q_0$ and, for $0<\delta\le1$, define on the same filtered space
\begin{equation}
 \dd R_t=[b_\delta(R_t)-U]\dd t
 +\sqrt{\kappa f_\delta(R_t)}\dd W_t+\dd L_t^0.
\label{eq:S-finite-reference}
\end{equation}
For $\delta=0$, solve the same equation until
$\tau_R=\inf\{t:R_t=0\}$ and set $R_t=0$ afterward.

\begin{theorem}[finite-horizon optimality]
\label{thm:S-finite}
For every $\delta\in[0,1]$, $U\ge0$, and $T>0$ under the declared operational contract,
\begin{equation}
 E_\delta^*(U,T)=\E\sin R_T.
\label{eq:S-finite-optimum}
\end{equation}
For $\delta=0$ the stopped maximum-inward feedback in Eq.~(\ref{eq:S-globalH}) attains the infimum.  For $\delta>0$ a sequence of smooth public Markov feedbacks of norm at most $U$ converges to the value in Eq.~(\ref{eq:S-finite-optimum}); exact attainment by the discontinuous limiting selector is not required or asserted.
\end{theorem}
\begin{proof}
Fix any $H\in\A_U$ and use its radial Brownian motion from Lemma~\ref{lem:S-radial}.  Start $R$ from the same random value $R_0=q_0$ and drive it with that Brownian motion.  On each compact interval below $\pi/4$, the common diffusion coefficient is Lipschitz and
\begin{equation*}
 b_\delta'(q)=-2\kappa\sec^2(2q)
 -2\kappa(1-\delta)\cos(4q)\le-2\kappa\delta.
\end{equation*}
Since $v_t\ge-U$, the stopped Tanaka argument used in Theorem~\ref{thm:S-multiqubit}, retaining its zero-mean stochastic integral before taking expectations, gives $R_t\le q_t$ almost surely.  Reflection is favorable: on $\{R>q\}$ the lower regulator of $R$ cannot increase, while subtracting that of $q$ is nonpositive.  The upper localization is removed by endpoint inaccessibility.  At $\delta=0$, compare only until $R$ first hits zero; thereafter $R=0\le q$.  Weak uniqueness for the scalar equation identifies the coupled reference marginal with Eq.~(\ref{eq:S-finite-reference}).  Therefore
\begin{equation}
 E_\delta^H(T;\nu_R)\ge\E\sin q_T\ge\E\sin R_T.
\label{eq:S-finite-lower}
\end{equation}

It remains to prove that the lower bound is sharp in the public control class.  In Bloch notation Eq.~(\ref{eq:S-globalH}) is
\begin{equation}
 H_N(\bm r)=U\operatorname{sgn}(y)
 \frac{z_BX-xZ}{\sqrt{x^2+z_B^2}},
\label{eq:S-Hbloch}
\end{equation}
off $y=0$ and the cut locus.  It has operator norm $U$ and gives
$\dot y=-2U\operatorname{sgn}(y)\sqrt{1-y^2}$, hence $\dot q=-U$.

A measurable interior selector and its smooth approximations provide the required public feedback.  The innovation vector fields on the pure-state Bloch sphere are, up to the sign convention for $W_Y$,
\begin{align}
 G_X&=2\sqrt\kappa(\bm e_x-x\bm r),\qquad
 G_Z=2\sqrt\kappa(\bm e_z-z_B\bm r),\nonumber\\
 G_{Y,\delta}&=2\sqrt\kappa\{\sqrt\delta(\bm e_y-y\bm r)
 +\sqrt{1-\delta}(\bm e_y\times\bm r)\}.
\label{eq:S-Blochfields}
\end{align}
Let $s=\sqrt{1-y^2}$ and, where $s>0$, use the tangent basis
$\bm n=(\bm e_y-y\bm r)/s$ and
$\bm t=(\bm e_y\times\bm r)/s$.
The covariance of Eq.~(\ref{eq:S-Blochfields}), divided by $4\kappa$, is
\begin{equation}
 A_\delta=
 \begin{pmatrix}
 y^2+\delta s^2 & \sqrt{\delta(1-\delta)}s^2\\
 \sqrt{\delta(1-\delta)}s^2 & 1+(1-\delta)s^2
 \end{pmatrix}.
\label{eq:S-ellipticmatrix}
\end{equation}
Its determinant is
$y^2+\delta s^2+(1-\delta)y^2s^2\ge\delta$ and its trace is at most three.  Thus, for $\delta>0$, the diffusion is uniformly elliptic on the two-dimensional tangent bundle of the compact sphere, equivalently in every smooth sphere chart, with tangent covariance bounded below by $4\kappa\delta/3$.

The feedback-free coefficients are smooth and Eq.~(\ref{eq:S-Hbloch}), extended by zero on $y=0$ and boundedly at the cut locus, adds a bounded Borel tangent drift.  In each smooth chart the bounded-drift uniformly elliptic martingale problem is well posed by the Zvonkin--Stroock--Varadhan theory \cite{Zvonkin1974}; uniqueness on chart overlaps gives a global sphere-valued weak law.  Its coordinate martingales have covariance $\sum_PG_PG_P^{\mathsf T}$, so martingale representation identifies it as a weak closed-loop filter law rather than only a radial construction.  This statement supplies uniqueness in law, not public strong implementability of the discontinuous selector.

For public implementation fix $\varepsilon>0$ and use the globally smooth feedback
\begin{equation}
 H_N^{(\varepsilon)}(\bm r)=
 U\tanh(y/\varepsilon)
 \frac{z_BX-xZ}{\sqrt{x^2+z_B^2+\varepsilon^2}}.
\label{eq:S-Hsmooth}
\end{equation}
It is $O(2)$-covariant and satisfies $\|H_N^{(\varepsilon)}\|_{\rm op}\le U$.  The smooth observation-form filter has a public strong solution, so $H_N^{(\varepsilon)}(\rho_t)$ is predictable and belongs to $\A_U$.  Its drift converges almost everywhere and in every finite chartwise $L^p$ to the bounded Borel drift above.  Tightness on the compact sphere, stability of uniformly elliptic martingale problems, and uniqueness of the limiting martingale problem therefore give weak path-law convergence as $\varepsilon\downarrow0$.  In particular, the terminal laws converge in trace-$W_1$ because the state space is compact.

For $\delta=0$, Eq.~(\ref{eq:S-Hbloch}) and all filter coefficients are locally Lipschitz on either open hemisphere $y>0$ or $y<0$ before the target hit; no Lipschitz claim is made across $y=0$.  The cut locus is inaccessible.  Strong existence and uniqueness therefore hold to the hit, after which the stopped $H=0$ target-circle filter has a unique tangential law, zero normal drift and diffusion, and remains on $\M$.  This stopping-time concatenation treats absorption separately from positive-$\delta$ reflection.

Finally, conjugation by a real $O\in O(2)$ rotates the pair $\operatorname{span}\{X,Z\}$ orthogonally; applying the same rotation to $(W_X,W_Z)$ preserves its law.  Moreover $OYO^{\mathsf T}=(\det O)Y$, and $W_Y\mapsto(\det O)W_Y$ preserves both pieces of Eq.~(\ref{eq:S-Blochfields}).  The Haar prior, $q$, stopping rule, $H_N$, and every $H_N^{(\varepsilon)}$ are covariant.  Uniqueness in law makes each whole terminal law $O(2)$ invariant.  The cut locus has zero mass: for the Borel limiting law this follows from $R_t<\pi/4$ almost surely, while each smooth positive-$\delta$ law has a chartwise density.  Hence nearest-point projection $\Pi$ is defined almost surely and pushes every such terminal law to $\nu_R$.

For $\delta=0$, the stopped feedback has $q_T=R_T$ and the coupling $(\rho_T,\Pi\rho_T)$ costs exactly $\sin R_T$.  For $\delta>0$, the same coupling under $H_N^{(\varepsilon)}$ costs $\E\sin q_T^{(\varepsilon)}$, which converges to $\E\sin R_T$.  Taking $\varepsilon\downarrow0$ proves the reverse inequality to Eq.~(\ref{eq:S-finite-lower}) and hence the exact infimum, without asserting that the discontinuous selector is itself a public strong feedback.
\end{proof}

The corresponding HJB cross-check is
\begin{equation}
 \partial_tV+b_\delta V_q+\frac\kappa2f_\delta V_{qq}
 -U|V_q|=0,\qquad V(T,q)=\sin q.
\label{eq:S-HJB}
\end{equation}
Order preservation makes the mild value nondecreasing, selecting drift $-U$.  For $\delta>0$ the reflected scalar control problem has $V_q(t,0)=0$ and a natural upper endpoint.  For $\delta=0$, the semigroup of the attaining stopped reference process has $V(t,0)=0$; this is not asserted as a Dirichlet boundary condition for the unrestricted control problem, whose controls could drive a state away from the target.  The positive-$\delta$ Neumann condition and terminal derivative are incompatible at the single terminal corner, so Eq.~(\ref{eq:S-HJB}) is only a mild/viscosity corroboration.  The theorem itself follows from comparison and the admissible approximating couplings above.

\section{Residual normal noise: stationary limit and crossover}

Fix $0<\delta\le1$ and abbreviate $f=f_\delta$, $b=b_\delta$.
The maximally inward proof process is the reflected diffusion
\begin{equation}
 \dd Q_t=[b(Q_t)-U]\dd t+\sqrt{\kappa f(Q_t)}\dd W_t+\dd L_t^0.
\label{eq:S-reference}
\end{equation}
The upper endpoint is inaccessible by the one-dimensional scale test: $b(q)\to-\infty$ as $q\uparrow\pi/4$, while $f(q)\to1$.
Moreover
\begin{equation}
 b'(q)=-2\kappa\sec^2(2q)-2\kappa(1-\delta)\cos(4q)
 \le-2\kappa\delta<0.
\label{eq:S-bmonotone}
\end{equation}

\begin{lemma}[stationary residual-noise comparison]
\label{lem:S-comparison}
Let $q_0$ have the Haar density in Eq.~(\ref{eq:S-haar}).
For every $H\in\A_U$, Eq.~(\ref{eq:S-reference}) can be initialized in stationarity and coupled so that $Q_t\le q_t$ for all $0\le t\le T$ almost surely.
Its stationary density is
\begin{align}
 \pi_{\delta,U}(q)&=Z_{\delta,U}^{-1}
 \frac{\cos(2q)}{f_\delta(q)^2}
 \exp\!\left[-\frac{2U}{\kappa}J_\delta(q)\right],\nonumber\\
 J_\delta(q)&=\int_0^q\frac{\dd s}{f_\delta(s)}
 =\frac1{2\sqrt\delta}
 \arctan\frac{\tan(2q)}{\sqrt\delta}.
\label{eq:S-stationary}
\end{align}
\end{lemma}
\begin{proof}
The zero-flux formula for a reflected scalar diffusion is
$\pi\propto f^{-1}\exp\{\int 2(b-U)/(\kappa f)\}$.
Substitution of Eq.~(\ref{eq:S-b}) gives Eq.~(\ref{eq:S-stationary}).
Relative to $p_H(q)=2\cos(2q)$,
\begin{equation*}
 \frac{\pi_{\delta,U}(q)}{p_H(q)}
 \propto f_\delta(q)^{-2}
 \exp[-2UJ_\delta(q)/\kappa]
\end{equation*}
is decreasing because $f_\delta'\ge0$ and $J_\delta'=1/f_\delta>0$.
Monotone likelihood-ratio order therefore gives stochastic order, and the initial quantile coupling
$Q_0=F_{\pi}^{-1}(F_H(q_0))\le q_0$ uses no extra controller information.

On every compact subinterval below $\pi/4$, the common diffusion coefficient $\sqrt{\kappa f}$ is Lipschitz.
Equation~(\ref{eq:S-bmonotone}), $v_t\ge-U$, and the favorable ordering of the lower reflection terms give the standard one-dimensional comparison for a progressively measurable drift.
Localize, apply the Tanaka positive-part argument (whose local-time term vanishes for the Lipschitz common diffusion coefficient), and then remove the upper localization using inaccessibility.
This yields $Q_t\le q_t$ pathwise.
The future radial innovation is Brownian relative to the public filtration by Lemma~\ref{lem:S-innovations}; initializing $Q_0\sim\pi_{\delta,U}$ therefore preserves stationarity.
This is a standard reflected-diffusion comparison step \cite{LionsSznitman1984,KruhnerXu2023Density}; the quantum content enters through the exact radial coefficients and the common public innovation.
\end{proof}

For every target law supported on $\M$,
\begin{equation}
 W_{1,\rm tr}(\mu_T^{\delta,H},\nu)
 \ge\E\sin q_T\ge L_\delta(U),
 \qquad
 L_\delta(U):=\int_0^{\pi/4}\sin q\,\pi_{\delta,U}(q)\dd q.
\label{eq:S-Ldelta}
\end{equation}
This is a finite-time lower bound uniform over all predictable, history-dependent feedback in $\A_U$.
Theorem~\ref{thm:S-finite} strengthens its interpretation for the uniform target: ergodicity of the maximally inward positive-$\delta$ diffusion gives
\begin{equation}
 \lim_{T\to\infty}E_\delta^*(U,T)=L_\delta(U).
\label{eq:S-stationary-limit}
\end{equation}

Define
\begin{equation}
 r=\frac{U}{\kappa\sqrt\delta},\qquad
 \theta=\arctan\frac{\tan(2q)}{\sqrt\delta},\qquad
 q_\delta(\theta)=\tfrac12\arctan(\sqrt\delta\tan\theta).
\label{eq:S-theta}
\end{equation}
The common singular factor cancels between numerator and denominator in $L_\delta$, leaving
\begin{align}
 L_\delta(U)&=
 \frac{\int_0^{\pi/2}s_\delta(\theta)w_{\delta,r}(\theta)\dd\theta}
 {\int_0^{\pi/2}w_{\delta,r}(\theta)\dd\theta},\nonumber\\
 s_\delta(\theta)&=\sin q_\delta(\theta),\nonumber\\
 w_{\delta,r}(\theta)&=\cos\theta
 \sqrt{\cos^2\theta+\delta\sin^2\theta}\,e^{-r\theta}.
\label{eq:S-nonsingular}
\end{align}

\begin{proposition}[exact small-misalignment crossover]
\label{prop:S-crossover}
If $\delta_j\downarrow0$ and $U_j/(\kappa\sqrt{\delta_j})\to r\in[0,\infty)$, then
\begin{align}
 \frac{L_{\delta_j}(U_j)}{\sqrt{\delta_j}}&\longrightarrow F_2(r),\nonumber\\
 F_2(r)&=\frac12
 \frac{\int_0^{\pi/2}\sin\theta\cos\theta e^{-r\theta}\dd\theta}
 {\int_0^{\pi/2}\cos^2\theta e^{-r\theta}\dd\theta}\nonumber\\
 &=\frac{r[1+e^{-\pi r/2}]}
 {2[r^2+2(1-e^{-\pi r/2})]},\qquad F_2(0)=\frac1\pi.
\label{eq:S-F2}
\end{align}
The function $F_2$ is continuous and strictly decreasing.
There are constants $r_0,C<\infty$ such that, uniformly for $0<\delta\le1$ and $r\ge r_0$,
\begin{equation}
 \left|\frac{L_\delta(U)}{\kappa\delta/(2U)}-1\right|
 \le\frac{C}{r^2}.
\label{eq:S-uniformlayer}
\end{equation}
\end{proposition}
\begin{proof}
For fixed $\theta<\pi/2$,
$s_\delta(\theta)/\sqrt\delta\to\tfrac12\tan\theta$ and
$w_{\delta,r}\to\cos^2\theta e^{-r\theta}$.
A uniform dominating bound follows by writing $a=\arctan(\sqrt\delta\tan\theta)$.  For $0<\theta<\pi/2$,
\begin{equation*}
 \frac{s_\delta(\theta)}{\sqrt\delta}
 w_{\delta,r}(\theta)
 =\sin\theta\cos\theta\,
 \frac{\sin(a/2)}{\sin a}\,e^{-r\theta}
 \le\frac1{\sqrt2}\sin\theta\cos\theta,
\end{equation*}
because $\sin(a/2)/\sin a=[2\cos(a/2)]^{-1}\le1/\sqrt2$ for $0\le a\le\pi/2$.  The endpoint values follow by continuity, and the denominator is dominated by $\cos\theta$.  Dominated convergence therefore proves Eq.~(\ref{eq:S-F2}).
Writing $F_2$ as one half the expectation of $\tan\theta$ under the normalized denominator measure gives
$F_2'(r)=-\tfrac12\operatorname{Cov}_r(\tan\theta,\theta)<0$.
The closed form follows by integrating numerator and denominator separately.
For numerical evaluation near $r=0$, the denominator uses \texttt{expm1} or the continuous series rather than subtracting nearly equal exponentials.

For Eq.~(\ref{eq:S-uniformlayer}), split Eq.~(\ref{eq:S-nonsingular}) at a fixed small $\theta_0$.
Uniformly in $0<\delta\le1$,
\begin{equation*}
 w_{\delta,r}(\theta)=e^{-r\theta}[1+O(\theta^2)],
 \qquad
 s_\delta(\theta)=\tfrac12\sqrt\delta\,\theta[1+O(\theta^2)]
\end{equation*}
on the inner interval.  The proved integrable bound above makes the outer numerator $O(\sqrt\delta e^{-r\theta_0})$, while the denominator outer part is $O(e^{-r\theta_0})$; hence both are uniform in $\delta$.
The change $y=r\theta$ gives denominator $r^{-1}[1+O(r^{-2})]$ and numerator $\sqrt\delta(2r^2)^{-1}[1+O(r^{-2})]$.
Since $\sqrt\delta/r=\kappa\delta/U$, their ratio is Eq.~(\ref{eq:S-uniformlayer}).
\end{proof}

For fixed $\delta>0$, Eq.~(\ref{eq:S-uniformlayer}) implies the necessary cost
\begin{equation}
 \C_\delta(\epsilon;\nu)\ge
 \frac{\kappa\delta}{2\epsilon}[1-o(1)].
\label{eq:S-costlower}
\end{equation}
More generally, if $\C_{\delta(\epsilon)}(\epsilon;\nu)=O(S)$ with
$S=\sqrt{\log(1/\epsilon)}$, monotonicity of $L_\delta$ and the uniform expansion force
\begin{equation}
 \delta(\epsilon)=O(\epsilon S).
\label{eq:S-necessarywindow}
\end{equation}
The infimum in Eq.~(\ref{eq:S-cost}) need not be attained.
For this necessity argument choose, for each $\epsilon$, an admissible near-minimizer with peak at most $\C_{\delta(\epsilon)}(\epsilon;\nu)+\epsilon$ and apply the all-controller bound before taking the limit.

The stationary limit alone does not determine a finite-time value, a distinction now resolved by Eq.~(\ref{eq:S-finite-optimum}).
At $\delta=0$ and $H=0$, Eq.~(\ref{eq:S-y}) gives a multiplicative equation for $y_t$ whose sign is preserved.
Taking the expectation of $|y_t|$ and using $y_0\sim\mathrm{Unif}[-1,1]$ yields
\begin{equation}
 \E|y_T|=\frac12e^{-4\kappa T}.
\label{eq:S-zerocontroly}
\end{equation}
Since $|y|=\sin(2q)\le2\sin q$, every target law on $\M$ obeys
\begin{equation}
 W_{1,\rm tr}(\mu_T^{0,0},\nu)\ge\E\sin q_T
 \ge\frac14e^{-4\kappa T}>0.
\label{eq:S-zerocontrolW}
\end{equation}
Thus $L_\delta(0)\to0$ is a statement about this lower certificate, not a fixed-$T$ zero-control feasibility transition and not a counterexample to the peak-cost theorem.

\section{Tangent capture, covariance, and the uniform law}

For $\delta=0$, use Eq.~(\ref{eq:S-globalH}) until
$\tau_0=\inf\{t:q_t=0\}$ and set $H=0$ afterward.
Before the hit,
\begin{equation*}
 q_{t\wedge\tau_0}=q_0-U(t\wedge\tau_0)
 +\int_0^{t\wedge\tau_0}b_0(q_s)\dd s+M_t,
\end{equation*}
where $b_0\le0$ and
$[M]_t=\int_0^{t\wedge\tau_0}\kappa\sin^2(2q_s)\dd s\le\kappa t$.
On $\{\tau_0>t_c\}$,
\begin{equation*}
 M_{t_c}\ge Ut_c-q_0\ge Ut_c-\pi/4.
\end{equation*}
For $Ut_c>\pi/4$, the exponential martingale inequality gives the nontrivial estimate; adjoining the trivial bound one for smaller budgets yields the all-budget form
\begin{equation}
 \Prob(\tau_0>t_c)\le
 \exp\!\left[-\frac{(Ut_c-\pi/4)_+^2}{2\kappa t_c}\right].
\label{eq:S-capture}
\end{equation}
At $q=0$, both $b_0$ and the radial diffusion vanish, so stopping $H_N$ makes the target circle invariant.

It remains to identify the law along the circle.
For a real orthogonal matrix $O$, the pair of traceless real Pauli operators $(X,Z)$ rotates under $P\mapsto OPO^{\mathsf T}$, while $Y\mapsto(\det O)Y$.
In the simultaneous continuous limit, the independent pair $(W_X,W_Z)$ is rotationally invariant.
For $\det O=-1$, the $\delta=0$ $Y$ term is the random-unitary innovation and transforms by $W_Y\mapsto-W_Y$, which leaves its law unchanged.
Equation~(\ref{eq:S-globalH}) obeys
$H_N(O\psi)=OH_N(\psi)O^{\mathsf T}$, and both $q$ and the event $\{\tau_0\le t_c\}$ are invariant.
Thus the filter, Haar prior, hitting event, and controller are jointly $O(2)$ covariant.
The successful terminal measure is an $O(2)$-invariant subprobability measure supported on $\mathbb{RP}^1$.
After division by its success probability, transitivity of the $O(2)$ action makes its normalized law the unique invariant probability $\nu_R$.

The terminal law has the form
$(1-r)\nu_R+r\mu_{\rm fail}$ with $r\le\Prob(\tau_0>T)$.
Coupling the common $(1-r)\nu_R$ mass identically and using the unit trace diameter on the remainder gives
\begin{equation*}
 W_{1,\rm tr}(\mu_T^{0,H_N},\nu_R)\le r.
\end{equation*}
Taking $t_c=T$ in Eq.~(\ref{eq:S-capture}) and requiring $r\le\epsilon$ gives the explicit admissible peak
\begin{equation*}
 U=\frac{\pi}{4T}+\sqrt{\frac{2\kappa}{T}\log\frac1\epsilon},
\end{equation*}
which proves the tangent upper bound.

\section{Matching robust construction for the uniform target}

First work with the unique limiting weak law identified in Theorem~\ref{thm:S-finite}, whose radial process has velocity $-U$ for $q>0$ and a nonsticky reflecting boundary.  Write its terminal state law as $\mu_T^{\rm lim}$.  Moment estimates below concern this law; admissible public policies will be supplied by the smooth feedbacks in Eq.~(\ref{eq:S-Hsmooth}).
For $0<\delta\le1$, set
\begin{equation}
 q_*:=\frac{\kappa\delta}{U},
 \qquad \tau_*:=\inf\{t:q_t\le q_*\}.
\label{eq:S-qstar}
\end{equation}
Because $b_\delta(q)\le0$ for the qubit, on $\{\tau_*>t_c\}$ the stopped martingale of Eq.~(\ref{eq:S-radial}) must exceed $Ut_c-\pi/4$ and has quadratic variation at most $\kappa t_c$.
Hence
\begin{equation}
 \Prob(\tau_*>t_c)\le
 \exp\!\left[-\frac{(Ut_c-\pi/4)_+^2}{2\kappa t_c}\right].
\label{eq:S-tubecapture}
\end{equation}

After the hit, compare the same maximal-inward process, started at $q\le q_*$, with the diffusion having the same coefficients and a reflecting lower boundary at $q_*$.
Its invariant law is the truncation
$\bar\pi_{\delta,U}^{(*)}=\pi_{\delta,U}(\cdot\mid q\ge q_*)$.
Two applications of one-dimensional reflection comparison---first from the hit point to $q_*$, then from $q_*$ to the stationary truncated process---give
\begin{equation}
 \sup_{s\ge0}\sup_{q_{\rm in}\le q_*}
 \E_{q_{\rm in}}\sin q_s
 \le\int_{q_*}^{\pi/4}\sin q\,\bar\pi_{\delta,U}^{(*)}(\dd q).
\label{eq:S-truncatedorder}
\end{equation}
The transformation in Eq.~(\ref{eq:S-theta}), followed by $y=r\theta$, gives constants $B,r_1<\infty$ such that
\begin{equation}
 \int_{q_*}^{\pi/4}\sin q\,\bar\pi_{\delta,U}^{(*)}(\dd q)
 \le B\frac{\kappa\delta}{U},
 \qquad r\ge r_1,
\label{eq:S-truncatedmoment}
\end{equation}
uniformly over $0<\delta\le1$.
The leading coefficient is
$\tfrac12\Gamma(2,2)/\Gamma(1,2)=3/2$.

The Haar prior, $q$, the hit event, measurement law, and controller are all $O(2)$-covariant.
Consequently the nearest-point projection of the successful terminal subensemble is exactly uniform on $\M$.
Couple each successful state to that projection and retain all failure runs.
Equations~(\ref{eq:S-tubecapture})--(\ref{eq:S-truncatedmoment}) then give
\begin{equation}
 W_{1,\rm tr}(\mu_T^{\rm lim},\nu_R)
 \le \exp\!\left[-\frac{(Ut_c-\pi/4)_+^2}{2\kappa t_c}\right]
 +B\frac{\kappa\delta}{U}.
\label{eq:S-robustupper}
\end{equation}
For every fixed $U,\delta>0,T$, smooth public feedbacks of the same peak satisfy
$\Wone(\mu_T^{\delta,H_N^{(\varepsilon)}},\mu_T^{\rm lim})\to0$ as $\varepsilon\downarrow0$.
Choose the two terms above at most $\epsilon/3$ each, then choose the smoothing scale so that the terminal-law discrepancy is below $\epsilon/3$.  This gives an actual allowed policy of error at most $\epsilon$, and proves, for finite constants depending only on the fixed contract,
\begin{equation}
 \C_\delta(\epsilon;\nu_R)
 \le K_0+K_1\sqrt{\log(1/\epsilon)}
 +K_2\frac{\kappa\delta}{\epsilon}.
\label{eq:S-costupper}
\end{equation}
Together with Eq.~(\ref{eq:S-costlower}), this yields
$\C_\delta(\epsilon;\nu_R)=\Theta_{T,\kappa,\delta}(1/\epsilon)$ for fixed $\delta>0$.
Combining Eqs.~(\ref{eq:S-necessarywindow}) and (\ref{eq:S-costupper}) yields the matching calibration window
\begin{equation}
 \C_{\delta(\epsilon)}(\epsilon;\nu_R)=O(S)
 \quad\Longleftrightarrow\quad
 \delta(\epsilon)=O(\epsilon S).
\label{eq:S-windowiff}
\end{equation}
If the phase deviation from the tangent quadrature is $\alpha$, the definition
$e^{-i\phi_\delta}=\sqrt\delta-i\sqrt{1-\delta}$ gives $\delta=\sin^2\alpha\sim\alpha^2$ and hence
$|\alpha|=O[\epsilon^{1/2}\log^{1/4}(1/\epsilon)]$.

This theorem is the reflected/absorbed $d=2$ boundary case.  Theorem~\ref{thm:S-multiqubit} proves that the same-initial finite-time comparison and projection coupling persist for $d\ge3$; the model-independent tubular theorem addresses lower bounds when no scalar reduction is available.

\section{A target-manifold terminal-window principle}
\label{sec:S-general-window}
\label{sec:S-general-principle}

The exact radial model permits sharp constants, but the inverse-accuracy obstruction does not require a scalar reduction.  We state the geometric form used to separate that mechanism from the symmetry of $\mathbb{RP}^{d-1}$.
For a pure-state Hamiltonian path, the horizontal lift of
$\dot\psi=-iH\psi$ has Fubini--Study speed
\begin{equation}
 |\dot X_t|_{\rm FS}=\Delta_{\psi_t}H_t
 =\sqrt{\langle H_t^2\rangle-\langle H_t\rangle^2}
 \le\|H_t\|_{\rm op}.
\label{eq:S-intrinsic-speed}
\end{equation}
Because distance to a closed target is one-Lipschitz, its normal
finite-variation rate is bounded by the same metric speed.  Thus the lower
bound below applies to any actuator contract that directly bounds this
state-space rate, with $c_HU$ replaced by that intrinsic bound.  Operator
norm is the physical realization used in the quantum applications.  This
necessary speed bound is not a general controllability theorem: a matching
upper construction additionally needs a Hamiltonian right inverse for the
required vector fields, together with capture, invariance, and mixing
assumptions.
Let $(\mathcal X,g)$ be a compact $C^3$ Riemannian manifold and let $\mathcal M\subset\mathcal X$ be a compact embedded $C^3$ submanifold.  Consider continuous adapted weak solutions
\begin{equation}
 \dd X_t=\{b(X_t)+G(X_t)u_t\}\dd t+\sigma(X_t)\dd W_t,
 \qquad \|G(X_t)u_t\|_g\le c_HU,
\label{eq:S-general-SDE}
\end{equation}
where $u_t$ may be any predictable history-dependent feedback and $\sigma$ is continuous as a bundle map on a tubular neighborhood of $\mathcal M$.
Fix a smooth isometric embedding $\iota:\mathcal X\to\mathbb R^K$ and, in a tube $\mathcal T_r$ of $\mathcal M$, let
$F(x)=\iota(x)-\iota(\pi(x))$, with $\pi$ the nearest-point projection.
The following It\^o hypotheses are part of the statement.  For every admissible weak solution, the global decomposition of $\iota(X)$ on $[0,T]$ and the decomposition of $F(X)$ stopped on leaving $\mathcal T_r$ have square-integrable martingale parts.  Uniformly over every interval $[s,s+h]\subset[0,T]$ they satisfy
\begin{align}
 \operatorname{TV}(A^\iota;[s,s+h])&\le (B_0+c_HU)h,\nonumber\\
 \operatorname{TV}(A^F;[s,(s+h)\wedge\tau_r])
 &\le (B_0+c_HU)h,\label{eq:S-general-TV}\\
 \|D\iota\,\sigma\|_F^2&\le\Lambda_X\quad\text{globally},\nonumber\\
 \|DF\,\sigma\|_F^2&\le\Lambda_*
 \quad\text{on }\mathcal T_r .
\label{eq:S-general-upperQV}
\end{align}
Assume also that on the target
\begin{equation}
 \operatorname{tr}[P_N(x)\sigma(x)\sigma(x)^*P_N(x)]
 \ge\lambda_\perp>0.
\label{eq:S-general-normalQV}
\end{equation}
For $x\in\mathcal M$, $DF(x)=D\iota(x)P_N(x)$, so isometry of $\iota$ gives
$\|DF(x)\sigma(x)\|_F^2=\operatorname{tr}[P_N\sigma\sigma^*P_N]$.
Continuity and compactness then give a smaller tube in which $\|DF\sigma\|_F^2\ge\lambda_*$, with $\lambda_*\ge\lambda_\perp/2$.

\begin{lemma}[conditional terminal-window anticoncentration]
\label{lem:S-semimartingale-anticoncentration}
Let $(\mathcal F_s)_{0\le s\le h}$ be a filtered space with $\mathcal G\subseteq\mathcal F_0$, and let
$Y_s=Y_0+A_s+M_s$, where $Y_0$ is $\mathcal G$-measurable, $A_0=M_0=0$, $A$ has total variation at most $Bh$, and $M$ is a continuous martingale.  The almost-sure bracket upper bound below makes $M_h$ conditionally $L^4$:
\begin{equation*}
 \E[\operatorname{tr}\langle M\rangle_h\mid\mathcal G]\ge\lambda h,
 \qquad \operatorname{tr}\langle M\rangle_h\le\Lambda h
 \quad\text{a.s.}
\end{equation*}
Then
\begin{equation}
 \E[\|Y_h\|\mid\mathcal G]\ge c_0\sqrt{\lambda h}-Bh,
\label{eq:S-anticoncentration}
\end{equation}
where $c_0>0$ depends only on the dimension and $\Lambda/\lambda$.
\end{lemma}
\begin{proof}
Put $Z=Y_0+M_h$.  Conditional martingale orthogonality gives
$\E(\|Z\|^2\mid\mathcal G)\ge\|Y_0\|^2+\lambda h$.
For every $G\in\mathcal G$, $\mathbf1_GM$ is again a martingale.  Applying the ordinary fourth-moment BDG inequality to these indicator-stopped martingales and using the defining integral identity for conditional expectation yields
$\E[\|M_h\|^4\mid\mathcal G]\le C\E[(\operatorname{tr}\langle M\rangle_h)^2\mid\mathcal G]$.
Together with $\|y+m\|^4\le8\|y\|^4+8\|m\|^4$, it gives
\begin{equation*}
 \E(\|Z\|^4\mid\mathcal G)
 \le C(\Lambda/\lambda)^2
 [\|Y_0\|^2+\lambda h]^2.
\end{equation*}
Conditional Paley--Zygmund applied to $\|Z\|^2$ therefore yields
$\E(\|Z\|\mid\mathcal G)\ge c_0\sqrt{\lambda h}$.
The pathwise inequality $\|Y_h\|\ge\|Z\|-\operatorname{TV}(A)$ proves Eq.~(\ref{eq:S-anticoncentration}).
The argument permits $A$ to use the past of $M$ and hence does not impose Markov feedback.
Only the bracket upper bound is pathwise; the lower bound needed here is the displayed conditional-expectation bound, so stopping does not require a pathwise bracket lower bound.
\end{proof}

\begin{theorem}[normal quadratic variation forces inverse-accuracy peak]
\label{thm:S-general-normal}
Under Eqs.~(\ref{eq:S-general-SDE})--(\ref{eq:S-general-normalQV}), for every fixed $T>0$ there are $U_0,c>0$ such that, for every $U\ge U_0$, every admissible weak solution generated by a predictable feedback, and every initial law,
\begin{equation}
 \E d_g(X_T,\mathcal M)
 \ge c\frac{\lambda_\perp}{B_0+c_HU}.
\label{eq:S-general-distance-lower}
\end{equation}
Consequently, for every probability law $\nu$ supported on $\mathcal M$,
\begin{equation}
 W_{1,g}(\Law[X_T],\nu)
 \ge c\frac{\lambda_\perp}{B_0+c_HU},
\label{eq:S-general-W-lower}
\end{equation}
and, when $c_H>0$, the peak cost is $\Omega(1/\epsilon)$.
The constants may depend on $T$, the tube, coefficient bounds, dimension, and the normal condition number $\Lambda_*/\lambda_*$; no uniformity under a degenerating condition number is asserted.
\end{theorem}
\begin{proof}
Choose $0<r_o<r$ inside the smaller tube on which $\|DF\sigma\|_F^2\ge\lambda_*$, and choose
$0<\rho<\sup_{x\in\mathcal X}d_g(x,\mathcal M)$ so that
$c_+\rho\le c_-r_o/2$, where
$c_-d_g(x,\mathcal M)\le\|F(x)\|\le c_+d_g(x,\mathcal M)$ in the tube.
Set $B=B_0+c_HU$ and take the terminal window
\begin{equation}
 h=\theta\lambda_*/B^2\le T,
 \qquad 0<\theta<(c_0/2)^2.
\label{eq:S-terminal-window}
\end{equation}
Condition on $\mathcal F_{T-h}$.
If $d_g(X_{T-h},\mathcal M)\ge\rho$, reaching the $\rho/2$ tube requires a fixed displacement in the embedded coordinates.  Indeed, compactness, injectivity of $\iota$, and nonemptiness of the displayed compact sets give
\begin{equation*}
 \eta_\rho:=\min_{\substack{d_g(x,\mathcal M)\ge\rho\\
 d_g(y,\mathcal M)\le\rho/2}}
 \|\iota(x)-\iota(y)\|>0.
\end{equation*}
The global $A^\iota$ and $\Lambda_X$ bounds in Eqs.~(\ref{eq:S-general-TV})--(\ref{eq:S-general-upperQV}), followed by Chebyshev, give probability $O(h)$ for a displacement at least $\eta_\rho$; hence the conditional terminal distance is at least $(\rho/2)[1-O(h)]$.

If $d_g(X_{T-h},\mathcal M)<\rho$, let
$\tau_o$ be the first exit from the $r_o$ tube and $\tau=\tau_o\wedge T$.
For large $U$, $Bh\le c_-r_o/4$; exit then requires the stopped normal martingale to move by at least $c_-r_o/4$.
Doob's $L^2$ inequality gives
\begin{equation*}
 \Prob(\tau_o<T\mid\mathcal F_{T-h})
 \le \frac{16\Lambda_*}{c_-^2r_o^2}h=:C_oh.
\end{equation*}
For the shifted stopped martingale $\widehat M^F$,
\begin{align*}
 \E[\operatorname{tr}\langle\widehat M^F\rangle_h
       \mid\mathcal F_{T-h}]
 &\ge\lambda_*h\Prob(\tau_o\ge T\mid\mathcal F_{T-h})\\
 &\ge\lambda_*h(1-C_oh).
\end{align*}
Increase $U_0$ until $C_oh\le1/2$ and apply Lemma~\ref{lem:S-semimartingale-anticoncentration} with $\lambda=\lambda_*/2$.
Using Eq.~(\ref{eq:S-terminal-window}) gives, conditionally on $\mathcal F_{T-h}$,
$\E\|F(X_\tau)\|\ge c_1\lambda_*/B$.
Let $E=\{\tau_o\ge T\}$.  Continuity puts the stopped point on the $r_o$-tube boundary on $E^c$, so $\|F(X_\tau)\|\le c_+r_o$ there, whereas $X_\tau=X_T$ on $E$.  Consequently,
\begin{align}
 \E[\|F(X_T)\|\mathbf1_E\mid\mathcal F_{T-h}]
 &=\E[\|F(X_\tau)\|\mathbf1_E\mid\mathcal F_{T-h}]\nonumber\\
 &\ge \frac{c_1\lambda_*}{B}-c_+r_oC_oh.
\label{eq:S-remove-exit}
\end{align}
The last term is $O(B^{-2})$, so metric equivalence on $E$ and nonnegativity on $E^c$ yield the conditional terminal bound $c_2\lambda_*/B$ for all sufficiently large $U$.
Combining the near and far cases, integrating the conditioning, and using $\lambda_*\ge\lambda_\perp/2$ proves Eq.~(\ref{eq:S-general-distance-lower}).
For every coupling to a law supported on $\mathcal M$, the transport distance dominates $d_g(X_T,\mathcal M)$ pointwise, proving Eq.~(\ref{eq:S-general-W-lower}).
\end{proof}

Positive normal trace alone is not a substitute for the stopped-It\^o and uniform variation/bracket assumptions above.  For pure states, $(2/\pi)d_{\rm FS}\le d_{\rm tr}\le d_{\rm FS}$, so the conclusion transfers to the trace-Wasserstein metric up to a fixed factor.

The opposite direction needs more structure than vanishing normal variance.
Fix an initial law $\mu_0$ and a three-stage allocation $T=t_c+t_m+t_p$.
Suppose the target is compact, connected, and boundaryless; a record-observable capture coordinate has a stopping time $\tau$ with
\begin{equation}
 \Prob(\tau>t_c)
 \le\exp\!\left[-\frac{[(a_{\rm cap}U-\beta)t_c-R_{\max}]_+^2}
 {2\Lambda t_c}\right];
\label{eq:S-general-capture-tail}
\end{equation}
after a hit, an allowed baseline keeps the path in $\mathcal M$ until the deterministic time $t_c$; at $t_c$, every successful path is restarted under the same fixed uniformly elliptic Markov diffusion for time $t_m$.
Assume that this fixed semigroup maps every input probability law to a density $\rho_0$ satisfying common bounds
$\rho_0\ge m_m>0$ and $\|\rho_0\|_{C^{2,\gamma}}\le M_m$.
Assume also a Hamiltonian right-inverse map $R$ for tangent fields such that
$\|R(v)\|_{\rm op,\infty}\le C_R\|v\|_{C^{1,\gamma}}$, with $C_R$ uniform over the construction and independent of $\epsilon$.
For a target density $p\in C^{2,\gamma}(\mathcal M)$ satisfying $p\ge m_p>0$, the common positive mixing interval therefore maps the arbitrary successful law at $t_c$ to a uniformly controlled density $\rho_0$.
Interpolating $\rho_s=(1-a_s)\rho_0+a_sp$ and solving
\begin{equation}
 \operatorname{div}(\rho_s\nabla\phi_s)
 =\mathcal L_{\mathcal M}^*\rho_s-\partial_s\rho_s
\label{eq:S-weighted-transport}
\end{equation}
produces a uniformly $C^{1,\gamma}$ tangent field independent of $\epsilon$.
Let $\C_{\mathsf B}(\epsilon;p,\mu_0,T)$ denote the infimum peak over policies satisfying this capture, deterministic-time synchronization, invariant baseline, common mixing, and tangent-right-inverse contract, with every trajectory retained.
Then
\begin{proposition}[conditional capture--transport upper bound]
\label{prop:S-general-tangent}
Under the stated contract,
\begin{equation}
 \C_{\mathsf B}(\epsilon;p,\mu_0,T)
 \le\max\!\left\{U_{\rm succ},
 \frac1{a_{\rm cap}}\left[
 \beta+\frac{R_{\max}}{t_c}
 +\sqrt{\frac{2\Lambda}{t_c}\log\frac D\epsilon}
 \right]\right\},
\label{eq:S-general-tangent-upper}
\end{equation}
where $D=\operatorname{diam}(\mathcal X)$ and $U_{\rm succ}$ is independent of $\epsilon$.
Thus the cost is $O[\sqrt{\log(1/\epsilon)}]$.
\end{proposition}
\begin{proof}
Equation~(\ref{eq:S-general-capture-tail}) makes the failure mass at most $\epsilon/D$ at the stated peak.
On success, invariance preserves the state on $\mathcal M$ until $t_c$.  Restarting every successful trajectory at that deterministic time under the same elliptic semigroup for $t_m$ gives the stated common lower and $C^{2,\gamma}$ bounds, regardless of its hit time or location.  Schauder estimates for Eq.~(\ref{eq:S-weighted-transport}) then give a common $C^{1,\gamma}$ bound on the required tangent field during the final interval $t_p$.
The uniform Hamiltonian right inverse realizes that field while the baseline control preserves $\mathcal M$, so its contribution to $U_{\rm succ}$ is independent of $\epsilon$.
Conditional on $\{\tau\le t_c\}$, the success law is exactly $p$; the unconditional terminal law retains the failure mixture, whose arbitrary coupling costs at most $D\Prob(\tau>t_c)$.
Sequential stages take the maximum of their peaks; simultaneously applied baseline and transport Hamiltonians are included in $U_{\rm succ}$ by the norm triangle inequality.
\end{proof}
The conclusion applies under the displayed capture, invariance, ellipticity, and transport hypotheses; the $\mathbb{RP}^{d-1}$ construction in Theorem~\ref{thm:S-multiqubit} verifies them explicitly.

\section{A nonmaximal K\"ahler-angle quantum realization}

The real-state targets above are maximally real submanifolds.
To test the mechanism beyond that geometry, consider two spins
$j_+>j_->0$ and the homogeneous state family
\begin{equation}
 \M_{j_+,j_-}=\left\{
 [|j_+,\bm n\rangle\otimes|j_-,-\bm n\rangle]:
 \bm n\in S^2\right\}.
\label{eq:S-slant-family}
\end{equation}
Here $|j,\bm n\rangle$ is a spin-coherent state and the two factors are rotated together by
\begin{equation*}
 K_a=J_a^{(+)}\otimes I+I\otimes J_a^{(-)},
 \qquad a=x,y,z.
\end{equation*}
The orbit is a two-sphere in projective Hilbert space.
We use ``slant'' in the standard sense that the angle between the ambient complex structure applied to a nonzero tangent vector and the tangent plane is constant \cite{Chen1990SlantImmersions}.
Throughout this section the Fubini--Study real metric is the horizontal Hilbert metric already used above,
$g_{\rm FS}(u,v)=\operatorname{Re}\langle u,v\rangle$.
The normal quadratic-variation rate means
$\sum_a\|P_\perp\Xi_a\|_{\rm FS}^2$ for the horizontal stochastic vectors $\Xi_a$; no additional factor of two or four is absorbed into $q_\perp$.

\begin{lemma}[proper-slant spin realization]
\label{lem:S-slant-spin}
Monitor the three operators $K_a$ with either
\begin{equation}
 c_{a,N}=\sqrt\gamma K_a,
 \qquad c_{a,T}=i\sqrt\gamma K_a.
\label{eq:S-slant-phases}
\end{equation}
The two instruments have the same discarded-record generator
$\gamma\sum_a\D[K_a]$.
The family in Eq.~(\ref{eq:S-slant-family}) has constant K\"ahler angle
\begin{equation}
 \cos\theta_K=\frac{j_+-j_-}{j_++j_-},
 \qquad
 \sin^2\theta_K=\frac{4j_+j_-}{(j_++j_-)^2}.
\label{eq:S-slant-angle}
\end{equation}
The $T$ diffusion is tangent to $\M_{j_+,j_-}$, whereas the total normal quadratic-variation rate of the $N$ diffusion on the target is
\begin{equation}
 q_{\perp,N}
 =\gamma\sin^2\theta_K(j_++j_-)
 =\frac{4\gamma j_+j_-}{j_++j_-}.
\label{eq:S-slant-qv}
\end{equation}
\end{lemma}
\begin{proof}
The equality $\D[iK_a]=\D[K_a]$ proves the first statement.
By rotational covariance it suffices to work at
$|j_+,\bm e_z\rangle\otimes|j_-,-\bm e_z\rangle$.
Write $z_a=(K_a-\langle K_a\rangle)|\psi\rangle$.
The real tangent plane is spanned by $-iz_x,-iz_y$, and the coherent-state second moments give
\begin{equation*}
 \|z_x\|^2=\|z_y\|^2=\frac{j_++j_-}{2},
 \qquad
 \operatorname{Im}\langle z_x,z_y\rangle=\frac{j_+-j_-}{2}.
\end{equation*}
Thus the ratio of the induced symplectic area to the induced metric area is Eq.~(\ref{eq:S-slant-angle}).
For $c_{a,T}$ the stochastic vector is $iz_a$, a common infinitesimal rotation, so every path remains on the orbit.
For $c_{a,N}$ it is $z_a$, the ambient-complex-structure rotation of that tangent vector.
The normal fraction is therefore $\sin^2\theta_K$.
Finally
$\sum_a\|z_a\|^2=\sum_a\operatorname{Var}(K_a)=j_++j_-$,
which proves Eq.~(\ref{eq:S-slant-qv}).
\end{proof}

For $j_-=0$, Eq.~(\ref{eq:S-slant-angle}) reduces to the complex spin-coherent sphere and both quadratures are tangent.
For $j_-=j_+$ it reaches the totally real endpoint.
Every $0<j_-<j_+$ is a proper-slant target with a strictly intermediate angle; for example,
$(j_+,j_-)=(3/2,1/2)$ has Hilbert dimension eight,
$\cos\theta_K=1/2$, and $q_{\perp,N}=3\gamma/2$.
Thus half-integer spin pairs realize a family of strictly intermediate normal fractions without changing the Lindbladian or the number of public records; the construction does not assume a continuously tunable spin label.

The tangent instrument also supports state-family generation on this second geometry.
Its restriction to the orbit is the rotational diffusion
$(\gamma/2)\sum_aX_a^2=(\gamma/2)\Delta_{S^2}$.
Fix $0<t_m<T$ and let the uncontrolled diffusion run from a fixed point to $t_m$.
Its heat-kernel density $\rho_m$ is smooth and strictly positive.
For a fixed normalized target density $p\in C^{2,\alpha}(S^2)$ with $\inf p>0$, $0<\alpha<1$, choose a $C^1$ interpolation $\rho_t$ from $\rho_m$ to $p$ that remains bounded below by $m_*>0$.
The right-hand side of
\begin{equation*}
 \Delta f_t=\frac\gamma2\Delta\rho_t-\partial_t\rho_t,
 \qquad \int_{S^2}f_t\dd\Omega=0,
\end{equation*}
has zero integral, so the zero-mean Poisson solution exists uniquely.
Elliptic regularity and $m_*>0$ make
$v_t=\nabla f_t/\rho_t$ bounded and continuous.
Use the collective Hamiltonian
$H_t=(\bm n\times v_t)\cdot\bm K$.
It produces the drift $v_t$ and obeys
$\|H_t\|_{\rm op}\le(j_++j_-)\|v_t\|$.
Because the unit-efficiency public filter reconstructs $\bm n_t$, this Borel state feedback is adapted and admits a predictable version.
The controlled Fokker--Planck equation then has the prescribed classical solution $\rho_t$, so the terminal law is exactly $p$ with a finite peak depending on the fixed $p,T,\gamma,j_\pm,m_*$ but not on an accuracy parameter.

The normal instrument now supplies a direct application of Theorem~\ref{thm:S-general-normal}.  Fix the same orbit point $x_0\in\M_{j_+,j_-}$ as the initial law for both instruments, the same $T,\gamma,j_\pm$, the same three public records, the full predictable Hamiltonian class $\|H_t\|_{\rm op}\le U$, and the same smooth target density $p\ge m_p>0$ on the orbit.

\begin{corollary}[proper-slant cost separation]
\label{cor:S-slant-cost}
Under this common contract, the normal instrument satisfies
\begin{equation*}
 \C_N(\epsilon;p,\delta_{x_0},T)=\Omega(1/\epsilon),
\end{equation*}
whereas the tangent instrument realizes $p$ exactly with a finite peak $U_p<\infty$ independent of $\epsilon$.
\end{corollary}
\begin{proof}
The ambient $\mathbb{CP}^{(2j_++1)(2j_-+1)-1}$ is compact and smooth, and the coherent-state orbit is a compact embedded two-sphere.  The unit-efficiency filter coefficients are smooth; after a fixed smooth isometric embedding, compactness bounds the global covariance and drift, while bounded derivatives of the stopped tube coordinate bound its covariance and include the It\^o curvature term in $B_0$.  Hamiltonian motion has Fubini--Study speed at most $\|H\|_{\rm op}$, so $c_H$ is finite.  Equation~(\ref{eq:S-slant-qv}) gives the strictly positive target-normal rate
$\lambda_\perp=4\gamma j_+j_-/(j_++j_-)$, and continuity supplies a control-independent tube.  Theorem~\ref{thm:S-general-normal} proves the normal lower bound for the fixed initial point.  Equations above construct the tangent law from that same point and give $U_p\le(j_++j_-)\|v\|_\infty$.
\end{proof}
The corollary covers the orbit-supported prior and a normal lower bound; a normal matching upper bound and ambient-Haar capture remain open.  Such an upper bound would additionally require a uniform tubular neighborhood, Hamiltonian access to every inward normal direction, controlled exit and boundary-layer moments, a specified captured prior, and a compatible near-manifold tangent mixing construction.  Positive normal quadratic variation alone supplies none of these controllability properties.  A deterministic finite-dimensional check over six spin pairs found maximum K\"ahler-angle, normal-trace, total-variance, and Lindbladian errors below $2.3\times10^{-15}$; it audits the algebra in Lemma~\ref{lem:S-slant-spin}, while the cost lower bound follows from the separately proved tube theorem.

\section{Structured multiqubit realizations}
\label{sec:S-structured}

The following realizations replace the full-space Haar task or the global action class by explicitly structured alternatives.

\subsection{Native one- and two-body single-excitation realization}
\label{sec:S-single-excitation}

The complete-Pauli theorem acts on the full $2^n$-dimensional space and uses high-weight monitors.  A physically different realization keeps a genuine multiqubit register but restricts it to the conserved single-excitation sector
\begin{equation}
 \mathcal H_1=\operatorname{span}\{|j\rangle:j=1,\ldots,m\}\simeq\mathbb C^m.
\label{eq:S-oneex-sector}
\end{equation}
With $E_{jk}=|j\rangle\langle k|$, define the $m^2$ Hermitian operators
\begin{equation}
 Q_j^D=E_{jj},\qquad
 Q_{jk}^S=\frac{E_{jk}+E_{kj}}{\sqrt2},\qquad
 Q_{jk}^A=-\frac{i(E_{jk}-E_{kj})}{\sqrt2},\quad j<k.
\label{eq:S-oneex-basis}
\end{equation}
They obey
\begin{equation}
 \tr(Q_\alpha Q_\beta)=\delta_{\alpha\beta},\qquad
 \sum_\alpha Q_\alpha XQ_\alpha=\tr(X)I_m,\qquad
 \sum_\alpha Q_\alpha^2=mI_m.
\label{eq:S-oneex-complete}
\end{equation}
On the physical qubits these are, respectively,
\begin{align}
 Q_j^D&=n_j=(I-Z_j)/2,\nonumber\\
 Q_{jk}^S&=(X_jX_k+Y_jY_k)/(2\sqrt2),\nonumber\\
 Q_{jk}^A&=-i(\sigma_j^+\sigma_k^--\sigma_j^-\sigma_k^+)/\sqrt2,
\label{eq:S-oneex-physical}
\end{align}
restricted to $\mathcal H_1$.  Hence every monitor conserves excitation number and has support at most two.

Let $L_\alpha=\sqrt{\Gamma/m}\,Q_\alpha$ and $\kappa=\Gamma/(2m)$.
Equation~(\ref{eq:S-oneex-complete}) gives the common sector generator
\begin{equation}
 \sum_\alpha\D[L_\alpha]\rho
 =\Gamma[I_m\tr(\rho)/m-\rho].
\label{eq:S-oneex-channel}
\end{equation}
Use the informational homodyne quadrature for $Q^D,Q^S$, and use
$e^{-i\phi}=\sqrt\delta-i\sqrt{1-\delta}$ for $Q^A$.
Changing this phase changes the conditional diffusion and leaves Eq.~(\ref{eq:S-oneex-channel}) invariant.
Because the $Q_\alpha$ are not involutions, this is a continuous diffusive instrument and not an application of the binary Pauli Kraus formula.

The prior is complex Haar measure on $\mathbb{CP}^{m-1}$ inside $\mathcal H_1$; the target is the uniform law $\nu_{1,m}$ on
\begin{equation}
 \mathcal M_1=\mathbb{RP}^{m-1}
 =\left\{\left[\sum_{j=1}^m x_j|j\rangle\right]:x_j\in\mathbb R\right\}.
\label{eq:S-oneex-target}
\end{equation}
Except on lower-dimensional coordinate subsets, these are entangled $W$-type states.
The allowed controls are all public-record-predictable number-preserving quadratic Hamiltonians whose restriction obeys
$\|H_t|_{\mathcal H_1}\|_{\rm op}\le U$.
Since Eq.~(\ref{eq:S-oneex-basis}) spans $\operatorname{Herm}(\mathbb C^m)$, this class contains the sector geodesic controller.

\begin{corollary}[native two-body multiqubit separation]
\label{cor:S-oneex}
For every $m\ge2$, the conditional radial coordinate
$q=\frac12\arccos|\psi^{\mathsf T}\psi|$ obeys Eq.~(\ref{eq:S-multi-radial}) with $d=m$ and $\kappa=\Gamma/(2m)$.  If $R^{(m)}$ denotes Eq.~(\ref{eq:S-multi-reference}) with these parameters, then for the uniform sector target
\begin{equation}
 \inf_H W_{1,\rm tr}(\Law[\rho_T],\nu_{1,m})
 =\E\sin R_T^{(m)}.
\label{eq:S-oneex-finite}
\end{equation}
For fixed $m,T,\Gamma$ and $0<\delta\le1$,
\begin{align}
 W_{1,\rm tr}(\Law[\rho_T],\nu_{1,m})
 &\ge L_{m,\delta}(U),\label{eq:S-oneex-lower}\\
 L_{m,\delta}(U)
 &=\frac{\kappa\delta(m-1)}{2U}
 \left[1+O_m\!\left(\frac{\kappa^2\delta}{U^2}\right)\right],\nonumber\\
 \C^{(1{\rm ex})}_{m,\delta}(\epsilon)
 &=\Theta_{m,T,\Gamma,\delta}(1/\epsilon).
\label{eq:S-oneex-cost}
\end{align}
At $\delta=0$ the stopped construction gives Eq.~(\ref{eq:S-oneex-tangent}) below, with no matching lower asymptotic.
\begin{equation}
 \C^{(1{\rm ex})}_{m,0}(\epsilon)
 \le\frac{\pi}{4T}
 +\sqrt{\frac{2\kappa}{T}\log\frac1\epsilon}.
\label{eq:S-oneex-tangent}
\end{equation}
\end{corollary}
\begin{proof}
The radial derivation uses only Hilbert--Schmidt completeness, the transpose split, and the sector operator-norm speed bound.  Equations~(\ref{eq:S-oneex-complete}) and the parities of $Q^{D,S,A}$ provide those identities with $d=m$.  Boundary handling is dimension dependent: for $m=2$ the positive-$\delta$ lower endpoint is regular and the reflected value is approached by smooth public feedbacks; for $m=3$ the scale density is proportional to $q^{-1}$ and zero is logarithmically inaccessible; for $m\ge4$ it is power-law inaccessible.  At $\delta=0$ the maximal-inward feedback is solved only to the first hit and then switched off.  The comparison, covariance, and projection arguments of Theorems~\ref{thm:S-finite} and~\ref{thm:S-multiqubit} now prove the claims.  No binary Pauli Kraus formula is used: the $Q_\alpha$ are not involutions and define the continuous diffusive sector instrument stated above.
\end{proof}

For fixed total sector depolarizing rate $\Gamma$, the leading normal coefficient is $\Gamma\delta(m-1)/(4mU)$ and stays bounded as $m$ grows.
This statement does not price total hardware: there are $m^2$ monitor directions, the pair terms are all-to-all, and the budget is the sector norm rather than the full $2^m$-dimensional operator norm or an $\ell_1$ sum of physical couplings.
The deterministic checker in the reproducibility package verifies Eqs.~(\ref{eq:S-oneex-complete})--(\ref{eq:S-oneex-channel}), transpose parity, support, normal covariance, and a bipartite entanglement witness for $2\le m\le8$; its maximum algebraic error is $3.58\times10^{-15}$.

\subsection{Locality-preserving graph-dressed corollary}
\label{sec:S-graph-local}

The exponential channel count in Theorem~\ref{thm:S-multiqubit} is not required to preserve the accuracy exponent if one accepts a structured prior, target, and action class.
Let $G=(V,E)$ be a simple graph on $n$ vertices with maximum degree $\Delta$, and define the real Clifford unitary
\begin{equation}
 V_G=\prod_{\{j,k\}\in E}CZ_{jk}.
\label{eq:S-graph-V}
\end{equation}
All factors commute and $V_G^{\mathsf T}=V_G=V_G^\dagger$.
For $a\in\{X,Y,Z\}$ set
\begin{equation}
 P_{j,a}=V_G\sigma_a^{(j)}V_G^\dagger.
\label{eq:S-graph-P}
\end{equation}
Writing $N(j)$ for the neighbors of $j$ gives
\begin{equation}
 P_{j,X}=X_j\!\prod_{k\in N(j)}Z_k,
 \quad P_{j,Y}=Y_j\!\prod_{k\in N(j)}Z_k,
 \quad P_{j,Z}=Z_j.
\label{eq:S-graph-labels}
\end{equation}
There are exactly $3n$ such observables, each supported on at most $\Delta+1$ qubits, and
$P_{j,X}^{\mathsf T}=P_{j,X}$,
$P_{j,Z}^{\mathsf T}=P_{j,Z}$, and
$P_{j,Y}^{\mathsf T}=-P_{j,Y}$.

Use $L_{j,a}=\sqrt\kappa P_{j,a}$, the informational quadrature for $a=X,Z$, and
\begin{equation}
 C_{j,Y}=(\sqrt\delta-i\sqrt{1-\delta})L_{j,Y}
\label{eq:S-graph-phase}
\end{equation}
for $a=Y$.
The common discarded-record generator is
\begin{align}
 \mathcal L_G={}&\kappa\sum_{j=1}^n\sum_{a=X,Y,Z}\D[P_{j,a}]\nonumber\\
 ={}&\operatorname{Ad}_{V_G}\circ
 \left(\kappa\sum_{j=1}^n\sum_{a=X,Y,Z}
 \D[\sigma_a^{(j)}]\right)\circ
 \operatorname{Ad}_{V_G^\dagger},
\label{eq:S-graph-channel}
\end{align}
independently of $\delta$.
This is a graph-dressed product of local depolarizing generators, not the full-system depolarizing generator in Eq.~(\ref{eq:S-multi-channel}); Eq.~(\ref{eq:S-graph-channel}) fixes the per-site exposure $\kappa$ when the two instruments are compared.

Draw independent $\psi_{j,0}$ from single-qubit complex Haar measure and set
$\Psi_0=V_G\bigotimes_j\psi_{j,0}$.
All $3n$ records and the initial labels are public.
The target family and law are
\begin{equation}
 \mathcal N_G=V_G[(\mathbb{RP}^1)^n],
 \qquad \nu_G=(\operatorname{Ad}_{V_G})_\#\nu_R^{\otimes n}.
\label{eq:S-graph-target}
\end{equation}
For a nonempty connected graph, this continuous family contains generically entangled states; exceptional product choices form a lower-dimensional subset.

The relevant local resource contract is
\begin{equation}
 \A^{G,\mathrm{loc}}_u=
 \left\{H_t=V_G\left(\sum_{j=1}^n h_{j,t}^{(j)}\right)V_G^\dagger:
 \begin{array}{l}
 h_{j,t}=h_{j,t}^\dagger,\ \tr h_{j,t}=0,\\[-2pt]
 h_{j,t}\text{ is predictable from all public records},\\[-2pt]
 \|h_{j,t}\|_{\rm op}\le u
 \end{array}\right\}.
\label{eq:S-graph-controls}
\end{equation}
Each physical summand has support at most $\Delta+1$ and the total Hamiltonian obeys $\|H_t\|_{\rm op}\le n u$.
Thus $u$ is a per-site radial-speed/operator-norm budget, not the global peak $U$ of Theorem~\ref{thm:S-multiqubit}.
Define $\C^{G,\mathrm{loc}}_\delta(\epsilon;\nu_G)$ by replacing $\A_U$ with Eq.~(\ref{eq:S-graph-controls}) and minimizing this per-site peak $u$.

\begin{corollary}[Graph-dressed local separation]
\label{cor:S-graph-local}
For every $0<\delta\le1$, every $H\in\A^{G,\mathrm{loc}}_u$, and every $j$,
\begin{equation}
 \Wone(\Law^{\delta,H}[\rho_T],\nu_G)\ge L_\delta(u).
\label{eq:S-graph-lower}
\end{equation}
For every fixed $0<t_c<T$ and fixed $\delta>0$, the single-site residual estimate supplies constants $u_0,B_2<\infty$, independent of $u$, $n$, and $\epsilon$, such that the factorized controller gives for all $u\ge u_0$
\begin{equation}
 \Wone(\Law^{\delta,H}[\rho_T],\nu_G)
 \le \sqrt{n\left\{
  e^{-(ut_c-\pi/4)_+^2/(2\kappa t_c)}
  +B_2\left(\frac{\kappa\delta}{u}\right)^2\right\}}.
\label{eq:S-graph-upper}
\end{equation}
Consequently, for fixed $n,T,\kappa$ and fixed $\delta>0$, as $\epsilon\downarrow0$,
\begin{equation}
 \C^{G,\mathrm{loc}}_\delta(\epsilon;\nu_G)
 =\Theta_{n,T,\kappa,\delta}(1/\epsilon).
\label{eq:S-graph-cost}
\end{equation}
More explicitly, its lower side is
$\kappa\delta[1-o(1)]/(2\epsilon)$ and the construction has an upper side of the form
$K_0+K_1\sqrt{\log(2n/\epsilon^2)}+K_2\sqrt n\,\kappa\delta/\epsilon$, with constants independent of $n$ and $\epsilon$ for a fixed capture window.
At $\delta=0$ the same factorized first-hit construction gives only the upper bound
\begin{equation}
 \C^{G,\mathrm{loc}}_0(\epsilon;\nu_G)
 \le\frac{\pi}{4T}
 +\sqrt{\frac{2\kappa}{T}\log\frac n\epsilon}.
\label{eq:S-graph-tangent}
\end{equation}
For every bounded-degree graph family, the number of monitors is linear and their support is bounded uniformly in $n$.  No matching thermodynamic $n$ scaling or tangent lower asymptotic is asserted.
\end{corollary}

\begin{proof}
Pull back the conditional state by $V_G^\dagger$.
The observed dynamics become $n$ complete single-qubit $X/Y/Z$ filters driven by independent innovations, plus the sum of on-site Hamiltonians in Eq.~(\ref{eq:S-graph-controls}).
Starting from the product prior, the conditional state remains a product along every sample path even when $h_j$ depends on the records of other sites; such cross-record dependence may classically correlate the site laws, so it is not used for the upper construction.

Define the logical marginal channel
\begin{equation}
 \pi_j(\rho)=\tr_{\ne j}(V_G^\dagger\rho V_G).
\label{eq:S-graph-marginal}
\end{equation}
Unitary invariance and contractivity of trace distance under partial trace make $\pi_j$ one-Lipschitz, and $(\pi_j)_\#\nu_G=\nu_R$.
The Wasserstein contraction inequality therefore gives
\begin{equation}
 \Wone(\Law[\rho_T],\nu_G)
 \ge W_{1,\mathrm{tr}}((\pi_j)_\#\Law[\rho_T],\nu_R).
\label{eq:S-graph-Wcontraction}
\end{equation}
In the joint public filtration, the radial martingale of site $j$ is a predictable unit-vector projection of its three innovations and hence remains Brownian.
Its Hamiltonian radial drift has magnitude at most $u$.
The single-qubit comparison leading to Eq.~(\ref{eq:S-Ldelta}) therefore applies even when $h_j$ uses all past public records, proving Eq.~(\ref{eq:S-graph-lower}).

For the upper bound choose each $h_j$ to depend only on its own initial label and record, with no shared control randomization.  For $\delta>0$, first estimate the product of the limiting maximal-inward weak laws.  Since $\sin^2q$ is bounded and continuous, the smooth public approximants of Theorem~\ref{thm:S-finite} converge in this second moment as well.  For every fixed $n,u,\delta,T$, choose their scales so that the additional sitewise second moments are at most a constant times $(\kappa\delta/u)^2$; absorb that constant into $B_2$.  The bound below is therefore achievable with public smooth policies, rather than requiring strong implementability of the discontinuous selector.
The terminal law then factorizes.
For pure product states, writing
$D_j=d_{\rm tr}(\rho_j,\sigma_j)$ gives
\begin{equation}
 d_{\rm tr}\!\left(\bigotimes_j\rho_j,\bigotimes_j\sigma_j\right)
 =\sqrt{1-\prod_j(1-D_j^2)}
 \le\sqrt{\sum_jD_j^2}.
\label{eq:S-graph-telescope}
\end{equation}
Here every $\rho_j$ and $\sigma_j$ is pure, so
$1-D_j^2=|\langle\psi_j|\phi_j\rangle|^2$ and product fidelity gives
the equality.
For each site, assign unit cost on capture failure and use
Lemma~\ref{lem:S-multi-truncated-moments} with $d=2,p=2$ after a
successful hit.  In this dimension $b_{2,\delta}\le0$, so the stopped
argument proving Eq.~(\ref{eq:S-capture}) gives
$p_{\rm cap}\le e^{-(ut_c-\pi/4)_+^2/(2\kappa t_c)}$ also for capture
of $q_*=\kappa\delta/u$.  Since $D_j^2\le1$ on failure and the lemma
bounds its conditional second moment on success,
Jensen's inequality gives
$\E\sqrt{\sum_jD_j^2}\le\sqrt{\sum_j\E D_j^2}$ and hence
$\E d_{\rm tr}\le\sqrt{n\{p_{\rm cap}+B_2(\kappa\delta/u)^2\}}$,
which is Eq.~(\ref{eq:S-graph-upper}).
Inverting Eqs.~(\ref{eq:S-graph-lower}) and (\ref{eq:S-graph-upper}), using Proposition~\ref{prop:S-crossover}, yields Eq.~(\ref{eq:S-graph-cost}).
At $\delta=0$, the union bound makes the probability that any site misses the target no larger than $n$ times Eq.~(\ref{eq:S-capture}); on joint success the $O(2)^n$ covariance makes the logical target law exactly $\nu_R^{\otimes n}$.
Unitary pushforward by $V_G$ then proves Eq.~(\ref{eq:S-graph-tangent}).
\end{proof}

The absence of a matching $n$ dependence in
Eqs.~(\ref{eq:S-graph-lower}) and (\ref{eq:S-graph-upper}) is not resolved by
tensorizing the single-site Wasserstein inequalities.
\begin{proposition}[Marginal Wasserstein bounds do not tensorize]
\label{prop:S-graph-no-tensorization}
Let $\nu_R$ be the single-site real-Haar target and let
$\nu_n=\nu_R^{\otimes n}$.  There are joint laws $\mu_n$ for which every
single-site marginal has the same strictly positive Wasserstein discrepancy
$a c_0$, while
\begin{equation}
 W_{1,\rm tr}(\mu_n,\nu_n)\le a,
 \qquad 0<a<1,
\label{eq:S-graph-no-tensorization}
\end{equation}
with $c_0>0$ independent of $a$ and $n$.  Consequently no universal
$c\sqrt n$ amplification follows from the marginal inequalities alone.
\end{proposition}
\begin{proof}
Fix a real ray $r_0$ and put $\sigma=\delta_{r_0}$ and
$c_0=W_{1,\rm tr}(\sigma,\nu_R)>0$.  Define
\begin{equation*}
 \mu_n=(1-a)\nu_n+a\sigma^{\otimes n}.
\end{equation*}
Each marginal is $(1-a)\nu_R+a\sigma$.  An optimal
Kantorovich--Rubinstein potential for $(\sigma,\nu_R)$ gives the lower bound
$a c_0$, and convexity gives the reverse bound, so the marginal distance is
exactly $a c_0$.  Coupling the common $(1-a)\nu_n$ component identically and
using that trace distance has diameter one on states gives
Eq.~(\ref{eq:S-graph-no-tensorization}).  Conjugation by $V_G$ preserves the
same statement for graph-dressed laws.
\end{proof}
This law-level construction does not assert that $\mu_n$ is generated by an
admissible monitored controller.  It proves that the present marginal
contraction argument cannot supply a thermodynamic lower bound without a new
dynamical anti-concentration or decorrelation theorem.  The factorized upper
controller has the required independence, whereas the all-controller lower
class in Eq.~(\ref{eq:S-graph-controls}) permits cross-record correlations.

The preceding corollary uses a common per-site cap.  To expose the many-body resource count, allow finite heterogeneous bounds
$u_j=\operatorname*{ess\,sup}_{t,\omega}\|h_{j,t}\|_{\rm op}<\infty$ for predictable dressed-local terms and impose the static additive contract
\begin{equation}
 \sum_{j=1}^n u_j\le S.
\label{eq:S-graph-total-budget}
\end{equation}
\begin{proposition}[linear total-peak requirement]
\label{prop:S-graph-total}
Under Eq.~(\ref{eq:S-graph-total-budget}), every predictable dressed-local policy satisfies
\begin{equation}
 W_{1,\rm tr}(\Law[\rho_T],\nu_G)
 \ge L_\delta(S/n).
\label{eq:S-graph-total-lower}
\end{equation}
For fixed $0<\delta\le1$, reaching error $\epsilon\downarrow0$ therefore requires
\begin{equation}
 S\ge \frac{n\kappa\delta}{2\epsilon}[1-o(1)].
\label{eq:S-graph-linear-scaling}
\end{equation}
\end{proposition}
\begin{proof}
Some site $j$ has $u_j\le S/n$.
The marginal contraction argument in Eqs.~(\ref{eq:S-graph-marginal})--(\ref{eq:S-graph-Wcontraction}) and the single-site all-controller comparison give
$W_{1,\rm tr}\ge L_\delta(u_j)$.
The stationary comparison in Lemma~\ref{lem:S-comparison} makes $L_\delta$ nonincreasing in its peak argument, so
$L_\delta(u_j)\ge L_\delta(S/n)$.
The large-peak expansion in Proposition~\ref{prop:S-crossover} yields Eq.~(\ref{eq:S-graph-linear-scaling}).
\end{proof}
Equation~(\ref{eq:S-graph-linear-scaling}) applies to the static sum-of-site-peaks contract $\sum_j\operatorname*{ess\,sup}_{t,\omega}\|h_{j,t}\|_{\rm op}\le S$.  The dynamically reallocated budget $\operatorname*{ess\,sup}_{t,\omega}\sum_j\|h_{j,t}\|_{\rm op}\le S$ does not supply the uniform site cap used by the pigeonhole step.  The result neither assumes additivity of global trace-Wasserstein error nor prices simultaneous-control electronics, graph-unitary preparation, or routing.

This corollary transfers the product construction by unitary covariance; the dressing specifies its entanglement structure.
It removes the exponential monitor count and unrestricted global actuation from one explicit contract while trading the full complex-Haar prior and full $\mathbb{RP}^{2^n-1}$ target for graph-dressed product laws.
Bounded support alone does not price parity-measurement circuits, graph preparation depth, routing, ancillas, detector efficiency, or feedback delay, and Eq.~(\ref{eq:S-graph-cost}) is a fixed-$n$ statement rather than an optimized thermodynamic scaling law.

\section{Target-circle generator and nonuniform laws}

Parameterize the real ray by
\begin{equation*}
 |\psi_\vartheta\rangle=(\cos\vartheta,\sin\vartheta)^{\mathsf T},
 \qquad \vartheta\in\mathbb R/\pi\mathbb Z.
\end{equation*}
Its Bloch vector is $(\sin2\vartheta,0,\cos2\vartheta)$.
Projecting the $X$ and $Z$ informational innovations and the $Y$ random-unitary innovation gives
\begin{equation}
 \dd\vartheta=\sqrt\kappa\cos(2\vartheta)\dd W_X
 -\sqrt\kappa\sin(2\vartheta)\dd W_Z
 +\sqrt\kappa\dd W_Y.
\label{eq:S-angle}
\end{equation}
The first two contributions have total variance rate $\kappa$ and the last has variance rate $\kappa$.
Therefore $\dd[\vartheta]_t=2\kappa\dd t$ and the generator is exactly
$\kappa\partial_\vartheta^2$.
This resolves the apparent factor-two ambiguity that results if the Bloch angle $2\vartheta$ is confused with the ray angle $\vartheta$.

A Hamiltonian $H_t=h_t(\vartheta)\sigma_y$ generates $\dot\vartheta=h_t(\vartheta)$ and has norm $|h_t|$.
The density on the circle obeys
\begin{equation}
 \partial_t\rho_t=\kappa\partial_\vartheta^2\rho_t
 -\partial_\vartheta(h_t\rho_t).
\label{eq:S-FP}
\end{equation}
Let $u=1/\pi$ and let the desired density $p$ be periodic, $C^2$, normalized, and satisfy $m_p:=\inf p>0$.
Choose $0<t_c<T$ and $a\in C^1([t_c,T])$ with
$a(t_c)=0$, $a(T)=1$, and $a'(t_c)=a'(T)=0$.
Set
\begin{equation}
 \rho_t=(1-a(t))u+a(t)p.
\label{eq:S-interpolation}
\end{equation}
Then $\rho_t\ge m_0:=\min\{u,m_p\}>0$.
The periodic current
\begin{equation}
 J_t(\vartheta)=\kappa\partial_\vartheta\rho_t(\vartheta)
 -\int_0^\vartheta\partial_t\rho_t(s)\dd s
\label{eq:S-current}
\end{equation}
satisfies $J_t(\pi)=J_t(0)$ because the mass of $\rho_t$ is constant.
Taking
\begin{equation}
 h_t(\vartheta)=\frac{J_t(\vartheta)}{\rho_t(\vartheta)}
\label{eq:S-transport}
\end{equation}
makes Eq.~(\ref{eq:S-FP}) identical to the prescribed interpolation.
Moreover
\begin{equation}
 \|h\|_\infty\le
 V_p:=\frac{\kappa\|p'\|_\infty
 +\pi\|a'\|_\infty\|p-u\|_\infty}{m_0}<\infty.
\label{eq:S-Vp}
\end{equation}

Capture by time $t_c$ with failure probability at most $\epsilon$, leave the successful uniform subensemble invariant until $t_c$, and apply Eq.~(\ref{eq:S-transport}) afterward.
The successful terminal law is exactly $p$, and the retained failures contribute at most $\epsilon$ in trace-$W_1$.
Thus
\begin{equation}
 \C_0(\epsilon;\nu_p)\le
 \max\!\left\{
 \frac{\pi}{4t_c}+\sqrt{\frac{2\kappa}{t_c}\log\frac1\epsilon},
 V_p\right\}.
\label{eq:S-nonuniformupper}
\end{equation}
The normal support lower bound is independent of the density along $\M$, so it remains $\Omega(1/\epsilon)$.

\subsection{Conditional residual-noise extension for regular nonuniform laws}

The preceding $\delta=0$ construction is complete.
For $\delta>0$, extending the matching upper bound to a nonuniform target additionally requires a tube estimate that is not proved in this manuscript; this subsection records the exact condition rather than promoting the extension to the main theorem.
Let $p$ be a fixed normalized periodic $C^{2,\gamma}$ density with $\inf p>0$.
On the target circle the projected uncontrolled generator is
\begin{equation}
 a_\delta\partial_\vartheta^2,
 \qquad a_\delta=\kappa(1-\delta/2)\in[\kappa/2,\kappa].
\label{eq:S-projectedgenerator}
\end{equation}
Replacing $\kappa$ by $a_\delta$ in Eqs.~(\ref{eq:S-current})--(\ref{eq:S-transport}) gives a tangent field with a uniform bound $V_{p,\delta}\le V_p^*$.

Fix a deterministic capture deadline $0<t_c<T$.
Use the normal controller with peak $U_N$ before $t_c$, retain every failed trajectory, and continue the inward controller after a successful hit so as to maintain only a stochastic tube layer until $t_c$.
For $\delta>0$, ``maintain'' never means exact invariance of $\mathcal M$: positive normal quadratic variation makes exact pathwise holding impossible under finite-variation Hamiltonian control.
Only at the common deterministic time $t_c$ is the tangent interpolation started.
The combined Hamiltonian is
$H=H_N+H_\parallel$ and obeys
\begin{equation}
 \|H\|_{\rm op}\le U_N+V_{p,\delta}\le U.
\label{eq:S-budgetallocation}
\end{equation}
Thus one may take $U_N=U-V_p^*$; for $U\ge2V_p^*$, $U_N\ge U/2$.

In Fermi coordinates $(\vartheta,n)$ about $\M$, smoothness of the finite-dimensional filter gives on every fixed tube $q\le q_1<\pi/4$
\begin{equation}
 \left|\mathcal K_\delta(g\circ\Pi)
 -a_\delta g''\circ\Pi\right|
 \le Cq(\|g'\|_\infty+\|g''\|_\infty).
\label{eq:S-tubegenerator}
\end{equation}
The missing uniform input is the following: for the successful, maintained process and a smooth cutoff before the projection cut locus, first and second tube moments must obey
\begin{equation}
 \sup_{t_c\le t\le T}\E[q_t\mid\tau_*\le t_c]
 \le C_1\frac{\kappa\delta}{U_N},\qquad
 \sup_{t_c\le t\le T}\E[q_t^2\mid\tau_*\le t_c]
 \le C_2\left(\frac{\kappa\delta}{U_N}\right)^2,
\label{eq:S-tubeassumption}
\end{equation}
with constants uniform over the stated $\delta,U_N$ range, and the cutoff-annulus contribution must satisfy the same order.
These estimates must hold for the combined field $H_N+H_\parallel$, including the chosen off-target extension of $H_\parallel$.
They are suggested by Eq.~(\ref{eq:S-truncatedorder}) but a full proof, including the cutoff-annulus contribution, is left open.

Conditional on Eq.~(\ref{eq:S-tubeassumption}), the backward Kolmogorov equation for the uniformly elliptic circle process gives
$\|g_t'\|_\infty\le C$ and
$\|g_t''\|_\infty\le C[1+(T-t)^{-1/2}]$.
It\^o's formula, Eq.~(\ref{eq:S-tubegenerator}), and the integrability of $(T-t)^{-1/2}$ then yield
\begin{equation}
 W_{1,\rm tr}(\Law(\Pi\rho_T\mid\tau_*\le t_c),\nu_p)
 \le C_p\frac{\kappa\delta}{U_N}.
\label{eq:S-projectedstability}
\end{equation}
Coupling to the projection and retaining the failed runs would imply
\begin{equation}
 \C_\delta(\epsilon;\nu_p)
 \le C_p'+C_1\sqrt{\log(1/\epsilon)}
 +C_{p,2}\frac{\kappa\delta}{\epsilon}.
\label{eq:S-nonuniformrobust}
\end{equation}
Equation~(\ref{eq:S-nonuniformrobust}) and its matching fixed-$\delta$ consequence are therefore conditional on Eq.~(\ref{eq:S-tubeassumption}); they are not used in the abstract, Theorem 1, or the main conclusion.

\section{Finite-efficiency all-control floor and delay boundary}
\label{sec:S-efficiency}

For efficiency $0\le\eta<1$, the accessible record reconstructs a conditional density matrix rather than a hidden pure trajectory.  The consistent output metric is therefore
\begin{equation}
 W_{1,\rm tr}\!\left(\Law[\rho_T(Y)],\iota_\#\nu\right),
 \qquad \iota([\psi])=|\psi\rangle\langle\psi|.
\label{eq:S-mixed-output}
\end{equation}
For any density matrix,
\begin{equation}
 \inf_{|\psi\rangle}d_{\rm tr}(\rho,|\psi\rangle\langle\psi|)
 =1-\lambda_{\max}(\rho),
\label{eq:S-pure-distance}
\end{equation}
so every pure-target ensemble obeys the operational lower bound
$W_{1,\rm tr}\ge\E[1-\lambda_{\max}(\rho_T)]$.

The complete-Pauli instrument admits a controller-uniform statement before any radial reduction.  Retain the normalization in Eq.~(\ref{eq:S-multi-L}), allow an arbitrary fixed quadrature phase $\phi_P$ for each nonidentity Pauli string, and write
\begin{equation}
 \dd\rho_t=\Gamma(I/d-\rho_t)\dd t-i[H_t,\rho_t]\dd t
 +\frac{\sqrt{\eta\Gamma}}{d}\sum_{P\in\mathcal P_n^\circ}
 \mathcal H_{e^{-i\phi_P}P}(\rho_t)\dd W_{P,t},
 \label{eq:S-inefficient-Pauli-SME}
\end{equation}
where
$\mathcal H_C(\rho)=C\rho+\rho C^\dagger-\tr[(C+C^\dagger)\rho]\rho$.
The $W_P$ are independent innovations of the accessible records.  The Hamiltonian is Hermitian, bounded, and predictable and enters only through the finite-variation term $H_t\dd t$; the contract excludes same-increment current feedback, additional resets outside the fixed instrument dilation, postselection, and hidden records.

\begin{theorem}[finite-efficiency complete-Pauli pure-target floor]
\label{thm:S-complete-Pauli-efficiency-floor}
Let $d=2^n$, let the pure initial state be arbitrary, and let $\rho_t$ solve Eq.~(\ref{eq:S-inefficient-Pauli-SME}) for any phases $\{\phi_P\}$ and any allowed Hamiltonian policy.  Define
\begin{equation}
 \alpha_\eta=(1-\eta)(d-1),\qquad
 \beta_\eta=d(1-\eta)+4\eta .
 \label{eq:S-efficiency-coefficients}
\end{equation}
Then every pure target law $\iota_\#\nu$ satisfies
\begin{equation}
 W_{1,\rm tr}\!\left(\Law[\rho_T],\iota_\#\nu\right)
 \ge \frac{\alpha_\eta}{2\beta_\eta}
 \left[1-\exp\!\left(-\frac{2\Gamma \beta_\eta}{d}T\right)\right].
 \label{eq:S-complete-Pauli-efficiency-floor}
\end{equation}
The bound holds for $n=1$ as well as for the multiqubit regime; the separate restriction $n\ge2$ in Theorem~\ref{thm:S-multiqubit} belongs to its pure-state radial comparison.
\end{theorem}
\begin{proof}
Put $p=\tr\rho^2$, $I_{\rm lin}=1-p$, $r_3=\tr\rho^3$, and $m_P=\tr(P\rho)$.  For each Pauli string define
\begin{equation*}
 A_P=\{P,\rho\}-2m_P\rho,
 \qquad B_P=-i[P,\rho].
\end{equation*}
Then
$\mathcal H_{e^{-i\phi_P}P}=\cos\phi_P A_P+\sin\phi_P B_P$,
$\tr(A_PB_P)=0$, and
\begin{equation}
 \tr A_P^2-\tr B_P^2
 =4\tr[(P-m_PI)\rho(P-m_PI)\rho]
 =4\tr[(\sqrt\rho(P-m_PI)\sqrt\rho)^2]\ge0.
 \label{eq:S-phase-purity-order}
\end{equation}
Thus the informational quadrature maximizes the It\^o purity correction pointwise over the phase.  Pauli completeness gives
\begin{align}
 \sum_{P\ne I}m_P^2&=dp-1,
 &\sum_{P\ne I}m_P\tr(P\rho^2)&=dr_3-p,\nonumber\\
 \sum_{P\ne I}\tr(P\rho P\rho)&=d-p,
 &\sum_{P\ne I}\tr A_P^2&=2d(1+dp-4r_3+2p^2).
 \label{eq:S-Pauli-purity-identities}
\end{align}
Applying It\^o's formula to $p$ eliminates the Hamiltonian contribution exactly.  Equation~(\ref{eq:S-phase-purity-order}), Eq.~(\ref{eq:S-Pauli-purity-identities}), and $r_3\ge p^2$ yield
\begin{equation}
 \dd p\le\frac{2\Gamma}{d}
 [1+\eta-d(1-\eta)p-2\eta p^2]\dd t+\dd M_t,
 \label{eq:S-purity-drift-bound}
\end{equation}
where $M_t$ is a true martingale on every finite horizon because the state space and coefficients are bounded.  Equivalently,
\begin{equation}
 \dd I_{\rm lin}\ge\frac{2\Gamma}{d}
 [\alpha_\eta-\beta_\eta I_{\rm lin}+2\eta I_{\rm lin}^2]\dd t+\dd N_t.
 \label{eq:S-linear-entropy-drift}
\end{equation}
For a pure initial state, $y(t)=\E I_{\rm lin}(t)$ therefore obeys
 $y'(t)\ge(2\Gamma/d)[\alpha_\eta-\beta_\eta y(t)]$ and $y(0)=0$.  Scalar comparison gives
\begin{equation*}
 \E I_{\rm lin}(T)\ge\frac{\alpha_\eta}{\beta_\eta}
 [1-e^{-2\Gamma \beta_\eta T/d}].
\end{equation*}
Finally, if $q=1-\lambda_{\max}(\rho)$, then
$I_{\rm lin}=2q-q^2-\sum_{j>1}\lambda_j^2\le2q$.
Combining this inequality with Eq.~(\ref{eq:S-pure-distance}) proves Eq.~(\ref{eq:S-complete-Pauli-efficiency-floor}).
\end{proof}

The two limiting checks are explicit.  At $\eta=1$ the bound vanishes, consistently with pure conditional trajectories.  At $\eta=0$, the entropy comparison becomes the exact identity
$\E I_{\rm lin}(T)=(1-1/d)(1-e^{-2\Gamma T})$ for an initially pure depolarizing trajectory.  Near unit efficiency,
\begin{equation}
 \epsilon_{\rm CP}(\eta)
 =\frac{d-1}{8}\left(1-e^{-8\Gamma T/d}\right)(1-\eta)
 +O((1-\eta)^2),
 \label{eq:S-complete-Pauli-efficiency-window}
\end{equation}
so terminal accuracy $W_{1,\rm tr}\le\epsilon$ requires $1-\eta=O(\epsilon)$ at fixed $d,\Gamma,T$.  The theorem supplies a positive controller-uniform obstruction for the main informational and partially tangent instruments.  It does not identify their exact mixed-state optimum, its dependence on the peak $U$, or a matching attainable policy.  Earlier finite-efficiency purification results treat a qubit under fixed complementary channels or a single monitored observable \cite{RuskovKorotkovMolmer2010,RuskovCombesMolmerWiseman2012,JiangWangMartinWhaley2020}; simultaneous noncommuting qubit monitoring has also been realized experimentally \cite{HacohenGourgy2016Noncommuting}.  Equation~(\ref{eq:S-complete-Pauli-efficiency-floor}) instead covers the complete $n$-qubit Pauli frame, arbitrary phase assignments, every finite-variation Hamiltonian in the declared class, and a law-level pure-target metric.

An isotropic no-knowledge benchmark gives a sharper exact special case.  To prevent confusion with the radial $\kappa$ used elsewhere, denote this benchmark's depolarizing rate by $\gamma_{\rm eff}$.
Let a complete traceless Hermitian basis be normalized so that, after ideal cancellation of the visible stochastic unitaries, the accessible conditional state obeys
\begin{equation}
 \dot\rho_t=-i[H_t,\rho_t]
 +(1-\eta)\gamma_{\rm eff}(I/d-\rho_t).
\label{eq:S-efficiency-master}
\end{equation}
The feedback $H_t\dd t$ is predictable and of finite variation; same-increment stochastic-current feedback has already been used only to cancel the visible part.
\begin{proposition}[exact efficiency floor after ideal visible-noise cancellation]
\label{prop:S-efficiency-floor}
For the complete isotropic no-knowledge benchmark after ideal cancellation of the visible stochastic unitaries, a pure initial state and every remaining finite-variation Hamiltonian policy in Eq.~(\ref{eq:S-efficiency-master}) satisfy
\begin{equation*}
 \rho_T=sU_T\rho_0U_T^\dagger+(1-s)I/d,
 \qquad s=e^{-(1-\eta)\gamma_{\rm eff} T},
\end{equation*}
and every pure-target law satisfies
\begin{equation}
 W_{1,\rm tr}\ge
 \epsilon_{\rm floor}(\eta)
 =\left(1-\frac1d\right)
 [1-e^{-(1-\eta)\gamma_{\rm eff} T}].
\label{eq:S-efficiency-floor}
\end{equation}
Thus $\epsilon_{\rm floor}\le\epsilon<1-1/d$ requires
\begin{equation}
 1-\eta\le-\frac1{\gamma_{\rm eff} T}
 \log\!\left(1-\frac{\epsilon}{1-1/d}\right)
 =\frac{\epsilon}{(1-1/d)\gamma_{\rm eff}T}+O(\epsilon^2).
\label{eq:S-efficiency-window}
\end{equation}
\end{proposition}
\begin{proof}
The isotropic depolarizer commutes with unitary conjugation, giving the displayed affine solution.
Its largest eigenvalue is $s+(1-s)/d$.
Equation~(\ref{eq:S-pure-distance}) is therefore the same on every record and proves Eq.~(\ref{eq:S-efficiency-floor}); algebraic inversion gives Eq.~(\ref{eq:S-efficiency-window}).
\end{proof}
This proposition uses ideal same-increment cancellation before applying the remaining finite-variation controller, so its exact value is not the minimax cost under the narrower class $\A_U$.  It is retained as a solvable benchmark and is distinct from the universal lower bound in Theorem~\ref{thm:S-complete-Pauli-efficiency-floor}.

Delay and detector bandwidth are also outside the continuous finite-variation control theorem.
For the declared sampled explicit controller, resolving the stationary normal layer requires the dimensionless condition
$hU^2/(\kappa\delta)\ll1$.
At fixed $h$, however, the allowed final correction slot gives $\C_{\delta,h}\le\pi/(2h)$, so the limits $h\to0$ and $\epsilon\to0$ do not commute.
The retained nonmonotone refinement points prevent this controller-specific scale from being promoted to a universal convergence or hardware-bandwidth lower theorem.

\section{Exact weak instruments and the macro-clock contract}
\label{sec:S-instruments}

We first give the single-qubit realization at per-channel rate $\kappa$.  The complete-Pauli version follows by using rate $\Gamma/d^2$ per string; its radial scale remains $\Gamma/(2d)$.

For $P^2=I$ and
$p_h=(1-e^{-2\kappa h})/2$, define
\begin{equation*}
 K_s=\frac{\sqrt{1-p_h}I+s\sqrt{p_h}e^{-i\phi}P}{\sqrt2},
 \qquad s=\pm1.
\end{equation*}
Expansion gives
\begin{align}
 K_s^\dagger K_s
 &=\frac12\left[I+2s\sqrt{p_h(1-p_h)}\cos\phi\,P\right],\\
 \sum_sK_s^\dagger K_s&=I,\\
 \sum_sK_s\rho K_s^\dagger
 &=(1-p_h)\rho+p_hP\rho P=e^{h\kappa\D[P]}\rho.
\label{eq:S-channel}
\end{align}
The conditional outcome mean is
\begin{equation*}
 \E[s\mid\rho]=\sqrt{1-e^{-4\kappa h}}\cos\phi\,\tr(P\rho).
\end{equation*}
Thus $\phi=0$ is informational and $\phi=\pi/2$ is a no-signal random-unitary readout.

\begin{table}[tb]
\caption{Operational comparison of the two endpoint instruments.  ``Matched'' refers to the declared resource contract, not to equal information content in the records.}
\label{tab:S-contract}
\centering
\begin{tabular}{p{0.22\textwidth}p{0.25\textwidth}p{0.25\textwidth}p{0.18\textwidth}}
\toprule
Operational item & normal instrument ($\delta=1$) & tangent instrument ($\delta=0$) & status \\
\midrule
System--ancilla interaction & $V_h$ with Pauli $X,Y,Z$ & the same $V_h$ & matched \\
Discarded-record channel & $e^{h\mathcal L_{\rm dep}}$ & $e^{h\mathcal L_{\rm dep}}$ & exactly matched \\
Outcome stream & three public binary outcomes per macro-step & the same alphabet and rate & matched \\
$Y$ readout basis & informational, $\phi=0$ & no-signal random unitary, $\phi=\pi/2$ & varied instrument property \\
State information in $Y$ record & generally nonzero & zero instantaneous signal & intentionally different \\
Prior, efficiency, horizon & known Haar prior, $\eta=1$, $T$ & the same & matched \\
Control and update & predictable $H_t\dd t$, $\|H_t\|\le U$, after a complete triple & the same & matched \\
Excluded resources & hidden records, postselection, same-increment current feedback, and additional resets outside the fixed dilation & the same & matched \\
\bottomrule
\end{tabular}
\end{table}

The comparison therefore holds access, bandwidth at the model clock, and action privileges fixed while changing the readout basis.
It does not assert equal mutual information: conditional geometry and record informativeness are jointly determined by that basis.

An explicit collision-model dilation prepares a fresh ancilla in $|0\rangle_a$, sets
$\sin^2\theta_h=p_h$, and applies
\begin{equation}
 V_{P,h}=e^{-i\theta_hP\otimes Y_a}.
\label{eq:S-collision-unitary}
\end{equation}
Its action defines the isometry
\begin{equation*}
 V_{P,h}|\psi\rangle|0\rangle_a=\sqrt{1-p_h}|\psi\rangle|0\rangle_a
 +\sqrt{p_h}P|\psi\rangle|1\rangle,
\end{equation*}
followed by the ancilla basis
$|s_\phi\rangle=(|0\rangle+s e^{i\phi}|1\rangle)/\sqrt2$.
Indeed, ${}_a\!\langle s_\phi|V_{P,h}|0\rangle_a=K_s$.
Preparing and resetting this ancilla is part of the fixed instrument in both columns of Table~\ref{tab:S-contract}; an adaptive reset of the system or an extra ancilla channel is not.

The clock convention removes a possible factor-three ambiguity.
One interval $[nh,(n+1)h)$ contains three weak interactions whose strengths are each calibrated by the same model increment $h$.
They implement a stroboscopic approximation to three channels acting simultaneously at rates $\kappa$; their laboratory pulse durations and duty cycle are absorbed into that calibration.
The controller receives all three outcomes and updates once at the macro-boundary.
It is therefore incorrect to reinterpret the mathematical clock advance as $3h$ while retaining the same $p_h$.
If laboratory hardware instead allocates three subintervals of duration $h/3$, its couplings must be rescaled so that the integrated strength of each channel remains $\kappa h$.

Because $\D[X],\D[Y],\D[Z]$ commute,
\begin{equation*}
 \prod_{P=X,Y,Z}e^{h\kappa\D[P]}
 =\exp\!\left(h\kappa\sum_P\D[P]\right)
 =e^{h\mathcal L_{\rm dep}},
\end{equation*}
where $\mathcal L_{\rm dep}\rho=4\kappa(I/2-\rho)$.
This exact equality applies to the feedback-free macro-channel.
Reading an outcome and acting between sub-instruments defines a different timing contract and is not in $\A_U$ for the finite-step comparison.

For $p_h=\kappa h+O(h^2)$, the normalized one-step conditional increment matches
\begin{equation*}
 \dd\rho=\kappa\D[P]\rho\dd t
 +\sqrt\kappa\,\Hc_{e^{-i\phi}P}(\rho)\dd W,
\qquad
 \Hc_C(\rho)=C\rho+\rho C^\dagger-\tr[(C+C^\dagger)\rho]\rho.
\end{equation*}
The three martingale arrays have vanishing cross-covariances and converge to independent innovations.
For bounded Lipschitz Markov feedback, standard weak-convergence stability gives convergence of the discrete state process and hence of its compact-state terminal $W_1$ law.
The main all-predictable theorem is instead a theorem about the limiting filter itself; no uniform discrete-to-continuous limit over all predictable controls is claimed.

\subsection{Instrument-constrained diffusion layer and relation to QuDDPM}

The finite instrument above can be used as one physically constrained generative layer without changing the main resource contract.  For $d=2^n$, set $\gamma=\Gamma/d^2$, replace $\kappa$ by $\gamma$ in the one-string construction, and enumerate the nonidentity Pauli strings as $P_1,\ldots,P_m$, where $m=d^2-1$.  Write
\begin{equation}
 \mathcal I^{(P,\phi)}_{s,h}(\rho)
 =\frac{K^{(P,\phi)}_{s,h}\rho K^{(P,\phi)\dagger}_{s,h}}
 {\tr[K^{(P,\phi)}_{s,h}\rho K^{(P,\phi)\dagger}_{s,h}]}.
\label{eq:S-normalized-instrument}
\end{equation}
For the public outcome vector $\mathbf s_n=(s_{1,n},\ldots,s_{m,n})$, one macro-step is
\begin{align}
 \rho_{n+1/2}&=
 \mathcal I^{(P_m,\phi_m)}_{s_{m,n},h}\circ\cdots\circ
 \mathcal I^{(P_1,\phi_1)}_{s_{1,n},h}(\rho_n),\nonumber\\
 H_{n+1}&=H_\theta(t_{n+1},\rho_{n+1/2},\mathbf s_{0:n},\xi),
 \qquad \|H_{n+1}\|_{\rm op}\le U,\nonumber\\
 \rho_{n+1}&=e^{-ihH_{n+1}}\rho_{n+1/2}e^{ihH_{n+1}}.
\label{eq:S-constrained-diffusion-layer}
\end{align}
The complete outcome is available before the finite-variation control slot, and the resulting Hamiltonian is used predictably on the next clock interval.  The parameter $\theta$ may label a learned controller, but neither learning nor a particular parameterization is needed for the cost theorem.  The candidate continuum process is
\begin{align}
 \dd\rho_t={}&\gamma\sum_{j=1}^{m}\D[P_j]\rho_t\dd t
 -i[H_\theta(t,\rho_t),\rho_t]\dd t\nonumber\\
 &+\sqrt\gamma\sum_{j=1}^{m}
 \Hc_{e^{-i\phi_j}P_j}(\rho_t)\dd W_{j,t}.
\label{eq:S-complete-diffusion-SME}
\end{align}

\begin{proposition}[fixed-policy continuum bridge]
\label{prop:S-discrete-continuum-bridge}
Fix a continuous phase schedule and a uniformly bounded Markov feedback $H_\theta(t,\rho)$ that is Lipschitz in time and state.  As $h\downarrow0$, the piecewise-constant interpolation of Eq.~(\ref{eq:S-constrained-diffusion-layer}) converges weakly to Eq.~(\ref{eq:S-complete-diffusion-SME}).  In particular,
\begin{equation}
 W_{1,\rm tr}\!\left(\Law(\rho^{(h)}_T),\Law(\rho_T)\right)\longrightarrow0.
\label{eq:S-terminal-W1-bridge}
\end{equation}
\end{proposition}
\begin{proof}
Write $\rho=\rho_n$ and $\widetilde\rho=\rho_{n+1/2}$ for the pre- and post-measurement states.  Uniformly on the compact state space, the conditional first two moments of $\widetilde\rho-\rho$ give the measurement drift and covariance of Eq.~(\ref{eq:S-complete-diffusion-SME}) up to $o(h)$, with conditional third absolute moment $O(h^{3/2})$.  The finite unitary slot expands at the post-measurement state as
\begin{equation*}
 \rho_{n+1}-\widetilde\rho
 =-ih[H_\theta(t_{n+1},\widetilde\rho),\widetilde\rho]+O(h^2).
\end{equation*}
Boundedness and Lipschitz continuity imply
$\E[\|\widetilde\rho-\rho\|\mid\rho]=O(\sqrt h)$.
Expressing the commutator at $(t_n,\rho)$ therefore incurs conditional $L^1$ error $O(h^{3/2})$ (and $O(h^2)$ from the time shift when the feedback is Lipschitz in time).  The control increment has second moment $O(h^2)$, and its cross moment with the measurement increment is $O(h^{3/2})$.  It contributes the required drift and no leading quadratic variation.  All these remainders are $o(h)$, so the discrete generators converge on smooth test functions and the martingale arrays satisfy the Lindeberg condition.  Here the stated Lipschitz hypothesis is in both time and state; the fixed phase schedule is taken continuous in time.  Uniqueness of the limiting martingale problem gives weak convergence of the interpolated chains.  Compactness makes trace distance bounded and continuous, which upgrades convergence of terminal laws to Eq.~(\ref{eq:S-terminal-W1-bridge}).
\end{proof}

This bridge applies to an instrument-constrained subclass of quantum diffusion layers, compared below with the original QuDDPM architecture \cite{Zhang2024QuDDPM}.  A QuDDPM reverse step learns a joint system--fresh-ancilla unitary and then measures the ancilla; the learned unitaries are fixed after training.  Equation~(\ref{eq:S-constrained-diffusion-layer}) instead fixes the weak collision and readout instrument, publishes its outcomes, and restricts the adjustable operation to a record-conditioned, system-only Hamiltonian with an explicit peak budget.  Recent measurement-based and reverse-time quantum diffusion models do use continuous monitored trajectories and state- or record-conditioned Hamiltonian denoising, making them the closer algorithmic relatives \cite{Liu2025MBQDM,Bompais2026Reverse,Gabbassov2026ReverseSSE}.

\begin{table}[tb]
\caption{Operational relation between the original QuDDPM and the constrained diffusion layer.  ``Covered'' refers to the theorem proved here, not to expressive inclusion of one architecture in the other.}
\label{tab:S-diffusion-bridge}
\centering
\begin{tabular}{p{0.24\textwidth}p{0.32\textwidth}p{0.34\textwidth}}
\toprule
Item & Original QuDDPM & Instrument-constrained layer here \\
\midrule
Per-step quantum operation & learned joint system--ancilla unitary and ancilla measurement & fixed weak collision and readout, followed by bounded system Hamiltonian \\
Online use of the record & reverse unitaries fixed after training & controller may depend causally on the public record \\
Continuum statement & no identification asserted here & fixed-policy weak convergence to Eq.~(\ref{eq:S-complete-diffusion-SME}) \\
Theorem coverage & outside the control contract & all predictable finite-variation controls for the limiting SME \\
\bottomrule
\end{tabular}
\end{table}

Consequently the present construction is a meaningful \emph{record-conditioned controlled quantum diffusion}: the weak instrument supplies the reference stochastic process, its outcomes supply the latent trajectory, and the bounded Hamiltonian transports the simple Haar prior toward a target state-family law.  It is not by itself a full learned diffusion architecture.  This paper does not specify a data-to-prior forward schedule, prove that an optimizer learns $H_\theta$, or transfer the peak-cost theorem to unrestricted trainable joint unitaries.  Its contribution is instead a resource law inside a well-defined physical diffusion class: channel-equivalent noise layers can have different generation costs because their conditional covariance points differently relative to the target family.

\section{Finite-clock limit and applicability window}

Let $\C_{\delta,h}$ denote the analogue of Eq.~(\ref{eq:S-cost}) for the exact finite-step macro-instrument, with one arbitrary single-qubit Hamiltonian slot of duration $h$ after each complete $X/Y/Z$ triple.
This discrete contract has a terminal correction that has no bounded continuous-time counterpart.

\begin{proposition}[fixed-clock terminal correction]
\label{prop:S-terminalpulse}
For every $h>0$, $\delta\in[0,1]$, target probability law on pure qubit states, and $\epsilon\ge0$,
\begin{equation}
 \C_{\delta,h}(\epsilon)\le\frac{\pi}{2h}.
\label{eq:S-terminalpulse}
\end{equation}
\end{proposition}
\begin{proof}
After the final triple, the complete record and known initial state determine the conditional pure state.
Use the declared independent seed to sample a target ray from the desired law.
Any two pure qubit rays have Fubini--Study distance at most $\pi/2$.
A Hamiltonian with norm $U$ can traverse distance $Uh$ in the final slot, so $Uh\ge\pi/2$ maps the conditional state exactly to the sampled target.
This is a causal construction under the discrete timing contract and uses neither postselection nor a hidden record.
\end{proof}

Thus fixed $h$ followed by $\epsilon\downarrow0$ gives a bounded sufficient peak, whereas the continuous residual-noise theorem gives a divergent peak for fixed $\delta>0$.
The limits are noncommuting because Eq.~(\ref{eq:S-terminalpulse}) concentrates a finite rotation into a slot whose duration vanishes as $h\to0$.
It should not be used as a counterexample to the continuous theorem or as evidence that clock rate is free.

The continuous boundary layer has width $w\asymp\kappa\delta/U$ and one finite step has normal diffusion scale $s_h\asymp\sqrt{\kappa\delta h}$.
A necessary resolution heuristic for an explicit discretization to see that layer is
\begin{equation}
 s_h\ll w
 \quad\Longleftrightarrow\quad
 \lambda_{h,\delta}:=\frac{hU^2}{\kappa\delta}\ll1.
\label{eq:S-clockparameter}
\end{equation}
Along the continuous fixed-$\delta$ accuracy law $U\asymp\kappa\delta/\epsilon$, this becomes
\begin{equation}
 h\ll\frac{\epsilon^2}{\kappa\delta}.
\label{eq:S-clockwindow}
\end{equation}
Equations~(\ref{eq:S-clockparameter})--(\ref{eq:S-clockwindow}) are an applicability scale obtained by comparing two resolved lengths.
They are not a uniform convergence theorem over all discrete feedback policies; Proposition~\ref{prop:S-terminalpulse} shows why such a statement would need an additional joint peak--bandwidth resource contract.

\section{Continuous numerics, finite-step diagnostics, and data provenance}
\label{sec:S-diagnostics}

\subsection{Main-figure continuous calculation for \texorpdfstring{$d=4$}{d=4}}
\label{sec:S-main-numerics}

Main Fig.~2(a) evaluates the radial value in Theorem~\ref{thm:S-multiqubit} for $d=4$, $\Gamma=4$, $\kappa=\Gamma/(2d)=0.5$, $T=1$, and the complex-Haar prior.  These deterministic calculations contain no trajectory-sampling error.  The numerical approximation and its resolution diagnostics are distinguished from the exact control-theoretic characterization of that value.

\paragraph{Coordinates, generator, and initial law.}
For the tangent case use $x=\sin^2(2q)\in[0,1]$.  It\^o's formula applied to Eq.~(\ref{eq:S-multi-radial}) gives the drift and variance
\begin{align*}
 c_x(x)&=4\kappa\{\delta(d-1)-(d\delta+1)x-(1-\delta)x^2\}
          -4U\sqrt{x(1-x)},\\
 a_x(x)&=16\kappa x(1-x)[\delta+(1-\delta)x].
\end{align*}
The implementation at $\delta=0$ uses nodes $x_i=i/N$.  The node $x=0$ is absorbing after first capture; the node $x=1$ has zero variance and the one-sided inward drift of this generator.  The exact initial CDF is $F_x(x)=x^{(d-1)/2}$, so node weights are its differences across Voronoi-cell edges.  For $d=4$, the density is $3\sqrt{x}/2$, not the uniform qubit density in $\sin(2q)$.

At positive $\delta$, the boundary layer is resolved in
$\theta=\arctan[\tan(2q)/\sqrt\delta]\in[0,\pi/2]$.
Set $g(\theta)=\cos^2\theta+\delta\sin^2\theta$; then
$f_\delta(q)=\delta/g(\theta)$ and
\begin{align*}
 c_\theta&=\frac{2\sqrt\delta}{f_\delta(q)}[b_{d,\delta}(q)-U]
            -\frac{\kappa\sqrt\delta\,f'_\delta(q)}{f_\delta(q)},
 &a_\theta&=4\kappa g(\theta),\\
 F_\theta(\theta)&=
 \left[\frac{\sqrt\delta\sin\theta}{\sqrt{g(\theta)}}\right]^{d-1}.
\end{align*}
There are $N$ cell centers $\theta_i=(i+1/2)\pi/(2N)$, with exact Haar cell masses from $F_\theta$.  Exterior links are omitted from the first and last discrete cells.  This is the numerical boundary closure for the diffusion's inaccessible endpoints, not an additional reflecting physical controller.  Its error is included in the successive-grid proxy and is not separately certified.

For either coordinate, the interior off-diagonal row-generator entries at grid spacing $\Delta$ are
\begin{equation*}
 Q_{i,i+1}=\frac{a_i}{2\Delta^2}+\frac{\max(c_i,0)}{\Delta},\qquad
 Q_{i,i-1}=\frac{a_i}{2\Delta^2}+\frac{\max(-c_i,0)}{\Delta},
 \qquad Q_{ii}=-\sum_{j\ne i}Q_{ij}.
\end{equation*}
Implicit Euler with $\Delta t=T/N_t$ propagates the payoff $\sin q$ under $Q$ and the initial mass under $Q^{\mathsf T}$, using sparse LU factorization.  Their dual evaluations, total mass, and minimum mass are saved with every run.

\paragraph{Displayed points and numerical resolution.}
Each of the $24$ parameter points has runs at $(N,N_t)=(512,4096)$, $(1024,2048)$, and $(1024,4096)$.  Six additionally use $(2048,4096)$: $(\delta,U)=(0,0)$, $(0.02,0)$, $(0.02,8)$, $(1,0)$, $(1,8)$, and $(1,16)$.  Table~\ref{tab:S-main-point-provenance} reports the chosen spatial grid and all values, including the three tangent points omitted from the logarithmic plot.  Let $e_x$ be the difference between the chosen spatial grid and the next coarser grid at $N_t=4096$, and $e_t$ the difference between $N_t=2048$ and $4096$ at $N=1024$.  Even when $N=2048$ is selected, the temporal diagnostic is on $N=1024$.  The plotted resolution proxy is $2(e_x+e_t)$; it is a refinement diagnostic, not a confidence interval or rigorous discretization bound.  Tangent values are displayed only when they exceed this proxy.  The original $78$ runs and all original plotted values are retained.

\begin{table}[tbp]
\caption{Main Fig.~2(a) point values and resolution proxies for $d=4$, $\kappa=0.5$, $T=1$. All chosen runs use $N_t=4096$. Columns $e_x,e_t$ are successive-grid and time differences; the displayed proxy is $2(e_x+e_t)$. The three starred tangent points are retained here but omitted from the plot because their values do not exceed that proxy. None of these proxies is a rigorous error bound.}
\label{tab:S-main-point-provenance}
\centering\small
\begin{tabular}{rrrrrrr}
\toprule
$\delta$ & $U$ & $N$ & $\widehat E^*$ & $e_x$ & $e_t$ & $2(e_x+e_t)$\\
\midrule
0 & 0 & 2048 & $0.05131$ & $0.000727$ & $2.717\times10^{-5}$ & $0.001508$ \\
0 & 1 & 1024 & $2.918\times10^{-5}$ & $7.419\times10^{-7}$ & $4.387\times10^{-7}$ & $2.361\times10^{-6}$ \\
0 & 2 & 1024 & $3.574\times10^{-9}$ & $3.107\times10^{-10}$ & $1.912\times10^{-10}$ & $1.004\times10^{-9}$ \\
0 & 4 & 1024 & $1.171\times10^{-19}\,{}^{*}$ & $5.113\times10^{-20}$ & $3.451\times10^{-20}$ & $1.713\times10^{-19}$ \\
0 & 8 & 1024 & $6.317\times10^{-51}\,{}^{*}$ & $3.934\times10^{-50}$ & $2.989\times10^{-50}$ & $1.385\times10^{-49}$ \\
0 & 16 & 1024 & $1.432\times10^{-149}\,{}^{*}$ & $8.023\times10^{-144}$ & $2.043\times10^{-143}$ & $5.691\times10^{-143}$ \\
0.02 & 0 & 2048 & $0.1061$ & $0.0001594$ & $1.953\times10^{-5}$ & $0.000358$ \\
0.02 & 1 & 1024 & $0.01479$ & $4.233\times10^{-5}$ & $6.466\times10^{-7}$ & $8.596\times10^{-5}$ \\
0.02 & 2 & 1024 & $0.007493$ & $3.128\times10^{-5}$ & $3.904\times10^{-10}$ & $6.256\times10^{-5}$ \\
0.02 & 4 & 1024 & $0.003775$ & $2.999\times10^{-5}$ & $2.725\times10^{-14}$ & $5.997\times10^{-5}$ \\
0.02 & 8 & 2048 & $0.001889$ & $1.483\times10^{-5}$ & $1.104\times10^{-14}$ & $2.967\times10^{-5}$ \\
0.02 & 16 & 1024 & $0.000967$ & $2.98\times10^{-5}$ & $2.075\times10^{-15}$ & $5.959\times10^{-5}$ \\
0.1 & 0 & 1024 & $0.1946$ & $0.0001437$ & $8.743\times10^{-6}$ & $0.0003048$ \\
0.1 & 1 & 1024 & $0.06805$ & $0.0001302$ & $1.428\times10^{-6}$ & $0.0002632$ \\
0.1 & 2 & 1024 & $0.03659$ & $8.269\times10^{-5}$ & $1.836\times10^{-9}$ & $0.0001654$ \\
0.1 & 4 & 1024 & $0.01869$ & $7.028\times10^{-5}$ & $1.438\times10^{-13}$ & $0.0001406$ \\
0.1 & 8 & 1024 & $0.009426$ & $6.738\times10^{-5}$ & $3.329\times10^{-15}$ & $0.0001348$ \\
0.1 & 16 & 1024 & $0.004752$ & $6.604\times10^{-5}$ & $7.738\times10^{-15}$ & $0.0001321$ \\
1 & 0 & 2048 & $0.4318$ & $2.006\times10^{-5}$ & $1.693\times10^{-11}$ & $4.013\times10^{-5}$ \\
1 & 1 & 1024 & $0.3449$ & $2.819\times10^{-5}$ & $4.343\times10^{-8}$ & $5.646\times10^{-5}$ \\
1 & 2 & 1024 & $0.2674$ & $9.767\times10^{-5}$ & $2.084\times10^{-8}$ & $0.0001954$ \\
1 & 4 & 1024 & $0.1687$ & $0.0001708$ & $3.538\times10^{-11}$ & $0.0003417$ \\
1 & 8 & 2048 & $0.09121$ & $9.924\times10^{-5}$ & $3.839\times10^{-13}$ & $0.0001985$ \\
1 & 16 & 2048 & $0.04664$ & $0.0001031$ & $2.188\times10^{-13}$ & $0.0002062$ \\
\bottomrule
\end{tabular}
\end{table}

The displayed tangent value at $U=2$ was also refined directly, as shown in Table~\ref{tab:S-tangent-U2-refinement}.  The four additional runs use the same solver.  The finest value, $3.2259917\times10^{-9}$, differs from the original $3.5743282\times10^{-9}$ by less than its $1.0039119\times10^{-9}$ proxy.  This check supports the small plotted value without claiming that the proxy is a proved error bound or replacing the original main-figure data.

\begin{table}[tb]
\caption{Direct refinement of the displayed tangent point $d=4,\delta=0,U=2,T=1$. Additional runs use the same positive generator as the main figure; they are not an independent solver.}
\label{tab:S-tangent-U2-refinement}
\centering
\begin{tabular}{rrr}
\toprule
$N$ & $N_t$ & $\widehat E^*$\\
\midrule
512 & 4096 & $3.8850345\times10^{-9}$ \\
1024 & 2048 & $3.7655778\times10^{-9}$ \\
1024 & 4096 & $3.5743282\times10^{-9}$ \\
1024 & 8192 & $3.4819185\times10^{-9}$ \\
2048 & 4096 & $3.4276104\times10^{-9}$ \\
2048 & 8192 & $3.3385702\times10^{-9}$ \\
4096 & 16384 & $3.2259917\times10^{-9}$ \\
\bottomrule
\end{tabular}
\end{table}

\paragraph{An observable same-budget window.}
For $\epsilon=0.01$, the analytic tangent construction gives $U_T^{\rm ub}=2.9313642$, while quadrature of the stationary lower bound at $\delta=1$ gives $U_N^{\rm lb}\simeq74.974995$.  Under a target-law perturbation $\sigma=0.002$, Eq.~(\ref{eq:S-perturbed-budget-window}) gives the sufficient/necessary thresholds $2.9827406$ and $62.469991$.  At the common budget $U=4$, the perturbed tangent error is at most $0.00203252$, whereas every normal controller has error at least approximately $0.166492$.  Independent quadratures in $q$ and in the main script's transformed coordinate agree within $2.1\times10^{-17}$ at these threshold examples.  This is floating-point consistency, not interval certification.  The underlying analytic inequalities supply the existence of the window.

\paragraph{Independent cross-discretization check.}
A second implementation does not import the upwind generator or its Haar
weights.  At $\delta=0$ it uses the regularized coordinate
$y=\sin(2q)=\sqrt{x}$, for which
\begin{equation}
 \dd y_t=\left[-4\kappa y_t-2U\sqrt{1-y_t^2}\right]\dd t
 +2\sqrt\kappa\,y_t\sqrt{1-y_t^2}\dd W_t,
 \qquad \rho_{\rm Haar}(y)=3y^2,
\label{eq:S-independent-y}
\end{equation}
and at positive $\delta$ it uses centered differences in $\theta$, sparse-LU
implicit Euler, PCHIP interpolation, and independent Haar quadrature.  The six points were fixed before comparison: $(\delta,U)=(0,0),(0,1),(0.02,0),(0.02,8),
(0.1,0),(1,8)$ and compared three-grid continuum diagnostics against the
upwind calculation.  Every cross-method difference was below
$\max(5\times10^{-6},0.25e_{\rm old})$.  The full gate, which additionally
requires monotone three-grid differences and observed order at least one,
passed at $(0,1),(0.02,8),(1,8)$.  It remains inconclusive at the other three
points because the upwind observed orders were $0.8981$, $0.9980$, and
$0.9988$.  The corresponding cross-method differences were
$7.83\times10^{-5}$, $4.01\times10^{-5}$, and $1.81\times10^{-5}$, all inside
their fixed tolerances.  All 56 raw refinement runs were finite; the complete
grid values, cutoff and time checks, software environment, script hashes, and
the three statuses remain inconclusive in the numerical record.  These
results support cross-method consistency but are not an a posteriori or
rigorous continuum error certificate, and the original Fig.~2 values remain
paired with their stated resolution proxies.

\paragraph{High-order zero-control check.}
The three unresolved $U=0$ points were subsequently recomputed by a third implementation that imports neither finite-difference generator.  It uses Chebyshev--Gauss--Lobatto method-of-lines discretization and continuous-time BDF propagation with an analytic matrix Jacobian.  For $\delta=0$ it works in $y=\sin(2q)$ with the first-hit absorbing condition at $y=0$; for $\delta=0.02,0.1$ it works directly in $x=\sin^2(2q)$ and enforces the exact degenerate endpoint equations, without a cutoff.  Independent Gauss--Jacobi quadrature and off-grid residual evaluation at seventeen positive times give the $N=256$ values
\begin{equation*}
 E_{0,0}^{(256)}=0.05227952905,\quad
 E_{0.02,0}^{(256)}=0.10594949609,\quad
 E_{0.1,0}^{(256)}=0.19444388366.
\end{equation*}
The corresponding maximum off-grid PDE residuals are
$4.67\times10^{-10}$, $5.80\times10^{-11}$, and $8.13\times10^{-11}$; strict-tolerance changes are below $1.4\times10^{-10}$, and halving the maximum BDF step changes none of the printed values.  These diagnostics support the terminal observable values, but a rigorous error bound still requires a residual-to-observable stability estimate.  The original order and residual-decrease criteria remain \textsc{inconclusive} at all three points; their thresholds, observed orders, and failed checks are retained in the reproduction notes.  Small floating-point residuals alone are not a proof of continuum accuracy.

\begin{figure}[t]
 \includegraphics[width=0.56\textwidth]{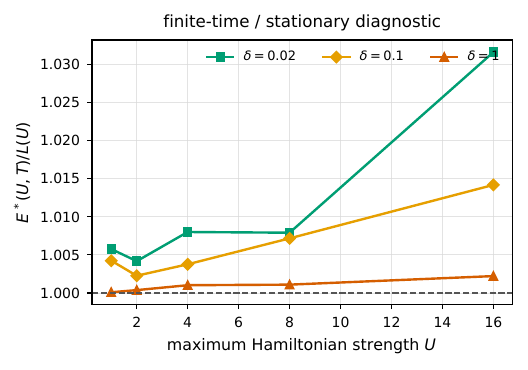}
 \caption{Continuous $d=4$ transient-to-stationary diagnostic.  The plotted quantity is the numerically evaluated finite-time optimum divided by the analytic stationary lower certificate.  Ratios close to one show relaxation at $T=1$; differences smaller than the reported discretization proxies are unresolved and are not used as evidence for the cost separation.}
 \label{fig:S-transient-ratio}
\end{figure}

\subsection{Full two-qubit finite-step check}

\begin{figure}[t]
 \includegraphics[width=0.94\textwidth]{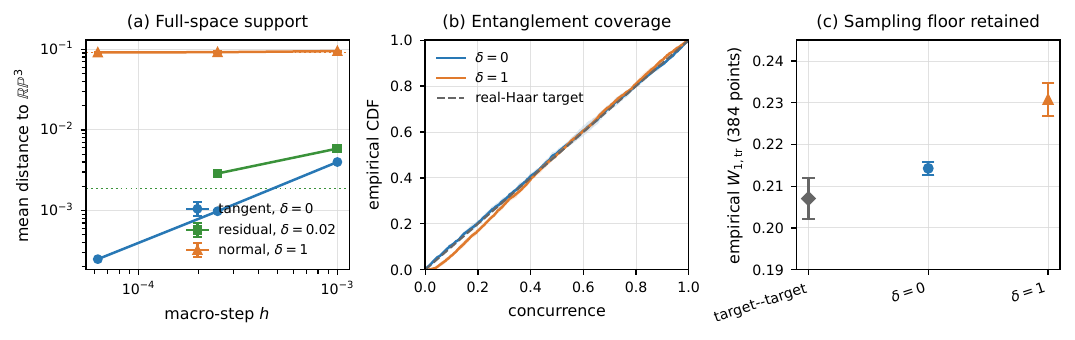}
 \caption{Frozen full-state $d=4$ diagnostics from the original full-step sample-and-hold controller.  Panel (c) uses twelve independent target--target ensembles, each containing $1024$ target states but using the first $384$ in the assignment.  This is distinct from the older qubit angular-$W_1$ sampling study.}
 \label{fig:S-twoqubit-old}
\end{figure}

The existing $d=4$ finite-step experiment implements all fifteen nonidentity two-qubit Pauli instruments on full complex-Haar initial states and stores every unprojected terminal vector.
It uses the same total depolarizing channel, $\Gamma=4$, $T=1$, and the explicit radial controller, with three independent seeds and $1024$ trajectories per condition.
At $U=8$, the mean nearest-$\mathbb{RP}^3$ trace distances were
\begin{center}
\begin{tabular}{c@{\qquad}cc}
\toprule
$h$ & tangent $\delta=0$ & normal $\delta=1$\\
\midrule
$10^{-3}$ & $0.003993\pm0.000034$ & $0.094987\pm0.001350$\\
$2.5\times10^{-4}$ & $0.000979\pm0.000020$ & $0.092039\pm0.001822$\\
$6.25\times10^{-5}$ & $0.000249\pm0.000001$ & $0.091474\pm0.000699$\\
\bottomrule
\end{tabular}
\end{center}
where the uncertainties are seed standard deviations.
The batch contains $67{,}584$ saved generated states over all conditions.
For the uniform tangent condition at $h=2.5\times10^{-4}$, the mean concurrence was $0.49334\pm0.00396$, compared with the analytic real-Haar mean $1/2$.
The $384$-point assignment distance had a target--target floor near $0.21$, so it cannot certify a population Wasserstein error at the support-distance scale.
These data test the declared controller and the full entangled target family; they do not replace the all-controller proof above or establish a scalable local implementation.

The stored controller applies a full angle $Uh$ before every fifteen-measurement macro-step.  For $0<q<Uh$ its exact unitary maps $q$ to $|q-Uh|$, so a phase spread inside this overshoot interval gives the observed mean residual close to $Uh/2$.  Deterministic checks verify branchwise real invariance, Born and macro-step normalization, the Hamiltonian norm, exact matrix exponentiation, global-phase covariance, operation order, and terminal save time.  The original policy is therefore a legal full-step sample-and-hold strategy, but it is not the stopped/limited strategy used in the continuous proof.

A separate finite-step diagnostic completes all fifteen measurements and then applies a Hamiltonian with slot amplitude $\min(U,q/h)\le U$.  With $U=8$, two seeds of $256$ trajectories at each of $h=10^{-3},2.5\times10^{-4},6.25\times10^{-5}$ all reached stable-coordinate support distance zero at the saved terminal time.  The raw vectors and step diagnostics are retained.  This result identifies the old $O(Uh)$ term as a policy overshoot rather than a physical tangent-noise floor; the post-macro-step policy belongs to the discrete cost $\C_{\delta,h}$ and is not a finite-$h$ proof of the continuous all-controller optimum.

For projected-label diagnostics, the squared coordinates of real Haar measure on $S^3$ follow a Dirichlet law whose four parameters all equal $1/2$.  Hence
\begin{equation*}
 \E x_i^2=\frac14,\qquad \E x_i^4=\frac18,\qquad
 \E x_i^2x_j^2=\frac1{24}\quad(i\ne j).
\end{equation*}
Recalculation from the unprojected frozen states gives pooled maximum deviations $(0.00309,0.00323,0.00098)$ for tangent $h=6.25\times10^{-5}$ and $(0.00716,0.00796,0.00089)$ for normal at the same $h$, compared with $(0.00474,0.00385,0.00075)$ in twelve pooled target ensembles.  These even moments are low-cost anisotropy diagnostics; the projection is used only to assign an evaluation label and does not certify the full law or alter any generated state.

\subsection{Qubit finite-clock diagnostics}

The confirmatory finite-step experiment uses the full macro-instrument of the theorem.
Each macro-step applies exact $X,Y,Z$ binary instruments with
$p_h=(1-e^{-2\kappa h})/2$, then one geodesic Hamiltonian rotates toward $\M$ by at most Fubini--Study distance $Uh$.
The normal and tangent conditions replay the same tagged uniform streams and use the same Haar states, $\kappa=0.7$, $T=1$, outcome rate, update time, and controller.
Only the $Y$ readout changes.

The refined protocol was frozen before its remote run.
It used four new seeds, $2048$ trajectories per seed and condition, the normal grids
$U=8$ with $h=2^{-11},2^{-12},2^{-13}$ and
$U=16$ with $h=2^{-12},2^{-13},2^{-14}$, and tangent controls at $h=2^{-13}$.
The continuous finite-horizon reference is Eq.~(\ref{eq:S-finite-optimum}), evaluated independently by the conservative Markov discretization described below.
Seed standard deviations are reported after the means:
\begin{table}[tb]
\caption{Finite-clock support error of the explicit complete-Pauli controller.  Values after $\pm$ are sample standard deviations across four seeds; $E^*(U,1)$ is the continuous finite-horizon optimum.}
\label{tab:S-clockpoints}
\centering
\begin{tabular}{c c c c c}
\toprule
$U$ & $h$ & $hU^2/\kappa$ & $\E\sin q_T$ & error/$E^*$\\
\midrule
$8$  & $2^{-11}$ & $0.04464$ & $0.041814\pm0.000997$ & $0.9722$\\
$8$  & $2^{-12}$ & $0.02232$ & $0.041573\pm0.000915$ & $0.9666$\\
$8$  & $2^{-13}$ & $0.01116$ & $0.043337\pm0.001028$ & $1.0076$\\
$16$ & $2^{-12}$ & $0.08929$ & $0.019934\pm0.000582$ & $0.9152$\\
$16$ & $2^{-13}$ & $0.04464$ & $0.021224\pm0.000308$ & $0.9744$\\
$16$ & $2^{-14}$ & $0.02232$ & $0.021030\pm0.000334$ & $0.9654$\\
\bottomrule
\end{tabular}
\end{table}
The first-order two-grid estimates are
$E^*(8,1)=0.0430079$ and $E^*(16,1)=0.0217825$, with successive-grid differences $1.21\times10^{-4}$ and $1.20\times10^{-4}$.
Their numerical proximity to the stationary values $L(8)=0.0430072$ and $L(16)=0.0217811$ is a relaxation result at this $T$, not an identification made by assumption.
The finest finite-step values are respectively $0.76\%$ above and $3.46\%$ below the finite-horizon references.
Adjacent means are not strictly monotone: their differences are comparable to seed variation.
Accordingly, the pre-specified frozen checks \texttt{normal\_monotone\_finer\_u8} and
\texttt{normal\_monotone\_finer\_u16} failed.
The data therefore provide resolved-window consistency while the monotone-convergence condition fails.

For the continuous computation we used the nonsingular coordinate $s=|y|=\sin(2q)$, whose Haar initial density is uniform on $[0,1]$.  Under maximal inward feedback,
\begin{equation}
 \dd s=[-4\kappa s-2U\sqrt{1-s^2}]\dd t
 +2\sqrt{\kappa(1-s^2)[\delta+(1-\delta)s^2]}\dd W+\dd K^0.
\label{eq:S-scoordinate}
\end{equation}
For $\delta>0$, $K^0$ enforces reflection; for $\delta=0$, zero is an absorbing atom and no reflecting term is used.  A positive birth--death approximation of the backward generator was propagated by implicit Euler with a factored conservative M-matrix.  The reported computation used spatial grids $N=256,512,1024,2048$ on representative convergence pairs and $4096$ time steps, plus $1024$--$8192$ step refinement.  Across these cases the maximum mass error was $2.79\times10^{-11}$, the minimum mass was nonnegative, and the maximum forward--backward duality discrepancy was $1.64\times10^{-12}$.  Tangent values below $10^{-12}$ were marked below numerical resolution rather than reported as exact zeros.  This deterministic solver evaluates the continuous reference; the finite-step points independently use the saved quantum-instrument trajectories.

\begin{figure}[t]
 \includegraphics[width=0.54\textwidth]{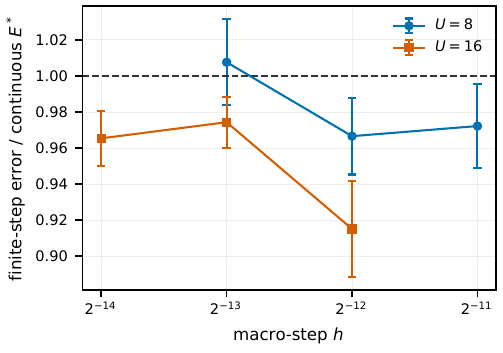}
 \caption{Finite-clock diagnostic showing all evaluated grid points.  The ordinate divides the support error of one explicit finite-step controller by the continuous finite-horizon optimum from Eq.~(\ref{eq:S-finite-optimum}).  Error bars are sample standard deviations over four seeds.  Values below one do not violate the continuous theorem because the discrete contract includes a post-triple correction slot and is a different admissible class.  Both sequences are nonmonotone and the two strict-monotonicity checks failed.}
 \label{fig:S-clockdiagnostic}
\end{figure}

We also checked the distribution along the target instead of reporting support distance alone.
Let $\widehat\mu_\Pi$ be the empirical law of each terminal state's nearest projection and let
$W_{1,\rm FS}^{\rm angle}(\widehat\mu_\Pi,\nu_R)$ be the circular one-dimensional Wasserstein distance in Fubini--Study arc length.
Samplewise nearest-point coupling and $d_{\rm tr}\le d_{\rm FS}$ imply the empirical full-law bound
\begin{equation}
 W_{1,\rm tr}(\widehat\mu,\nu_R)
 \le \frac1N\sum_{i=1}^N\sin q_i
 +W_{1,\rm FS}^{\rm angle}(\widehat\mu_\Pi,\nu_R).
\label{eq:S-empiricalfull}
\end{equation}
Table~\ref{tab:S-empiricallaw} reports the components at $h=2^{-13}$.
\begin{table}[tb]
\caption{Seedwise-averaged components of the empirical coupling upper bound at $h=2^{-13}$.  Each seed contains $2048$ trajectories.}
\label{tab:S-empiricallaw}
\centering
\begin{tabular}{c c c c c}
\toprule
instrument & $U$ & support & angular $W_1^{\rm FS}$ & full upper bound\\
\midrule
normal  & $8$  & $0.043337$ & $0.013212$ & $0.056549$\\
normal  & $16$ & $0.021224$ & $0.013571$ & $0.034796$\\
tangent & $8$  & $0$        & $0.014069$ & $0.014069$\\
tangent & $16$ & $0$        & $0.014456$ & $0.014456$\\
\bottomrule
\end{tabular}
\end{table}
The corresponding seed standard deviations of the full upper bound were
$0.00548,0.00579,0.00326,0.00337$.
Each angular Wasserstein value was computed from one seed's $2048$ trajectories and the four scalar values were then averaged; pooling all $8192$ trajectories before computing $W_1$ would define a different statistic.
The Haar initial states are common across all conditions for a fixed seed, and the tagged outcome stream is common only when both $h$ and the seed agree.

A separate sampling calculation used the same $N=2048$ seedwise statistic and the same four-seed averaging.
Across $1000$ independent four-seed groups drawn directly from the uniform target, the mean projected angular $W_1^{\rm FS}$ was $0.015934$ and the central $95\%$ range of the four-seed mean was $[0.011481,0.021417]$.
This band is a sampling reference and is not subtracted from the experimental statistic.
All four frozen angular-harmonic checks and the angular-$W_1$ gate passed.
Zero tangent support means that every stored finite-step sample was captured.
Each displayed tangent condition has four seeds times $2048$ trajectories, hence $N=8192$; the exact one-sided $95\%$ zero-failure binomial upper limit is
$1-0.05^{1/8192}=3.656\times10^{-4}$.
No samples from different conditions are combined, and this finite-sample statement is not used to infer a tail exponent.

In total, ten of twelve frozen gates passed and both failures remain in the formal summary.
The experiment tests one declared controller and neither proves a fixed-$h$ all-controller lower bound nor estimates $\C_\delta$.

\subsection{Single-channel mechanism experiment}

The analytic theorem uses the complete Pauli instrument because its normal radial diffusion is globally nondegenerate.
The finite-step mechanism experiment instead isolates one exactly solvable $Y$ channel.
Let
\begin{equation*}
 \alpha=\arcsin\tr(\rho\sigma_y)\in[-\pi/2,\pi/2].
\end{equation*}
The complex-Haar prior, $Y$ instrument, and geodesic controller are covariant under rotations about $Y$.
Conditional on $|\alpha|$, the azimuth is therefore uniform.
Mapping each state to the nearest point of the equator preserves that uniform marginal and attains the support lower bound, giving the exact identity
\begin{equation}
 W_{1,\rm tr}(\Law[\rho_T],\nu_R)
 =\E\sin\frac{|\alpha_T|}{2}.
\label{eq:S-exactW1}
\end{equation}

For $\phi=0$, the continuous latitude obeys
\begin{equation*}
 \dd\alpha_t=[v_t+2\gamma\sin\alpha_t\cos\alpha_t]\dd t
 +2\sqrt\gamma\cos\alpha_t\dd W_t,
 \qquad |v_t|\le2U.
\end{equation*}
Its noise degenerates at the Bloch poles, which is why it is not used for the global all-feedback theorem.
For $\phi=\pi/2$, every finite branch is a unitary rotation about $Y$.
The geodesic controller reduces $|\alpha|$ by $\min(2Uh,|\alpha|)$ after each step, so the tangent latitude is pathwise zero by time $T$ when $U\ge\pi/(4T)$.

The frozen remote protocol used
$\gamma=0.7$, $T=1$, $h=1/8192$,
$U\in\{0,0.5,1,2,4,8,16,32\}$,
four frozen seeds, and $4096$ trajectories per seed.
The four largest budgets are reported in Table~\ref{tab:S-singleY}.
\begin{table}[tb]
\caption{Single-$Y$ mechanism diagnostic at the four largest peak budgets.  The tangent zeros follow from the explicit pathwise controller.}
\label{tab:S-singleY}
\centering
\begin{tabular}{c c c}
\toprule
$U$ & normal $W_{1,\rm tr}$ & tangent $W_{1,\rm tr}$\\
\midrule
$4$  & $8.6239\times10^{-2}$ & $0$\\
$8$  & $4.3322\times10^{-2}$ & $0$\\
$16$ & $2.1025\times10^{-2}$ & $0$\\
$32$ & $9.2287\times10^{-3}$ & $0$\\
\bottomrule
\end{tabular}
\end{table}
Here ``zero'' is the pathwise analytic outcome of the tangent controller; the stored floating-point arrays agree to machine precision rather than defining a numerical threshold.
The normal finite-range fit has slope $-1.072$.
The interpolated peaks $4.314,8.636,16.688$ at
$\epsilon=0.08,0.04,0.02$ give slope $-0.976$.
These regressions use few accuracy levels and are reported as mechanism diagnostics, not confidence intervals for an asymptotic exponent.
The maximum seedwise relative change between the two finest grids was $0.215$, below the frozen $0.30$ threshold.
Kraus completeness and channel errors were below $1.57\times10^{-16}$, and an independent recomputation of all $64$ stored terminal arrays agreed within $5.25\times10^{-8}$, the float32 storage tolerance.

Table~\ref{tab:S-evidence} separates the conclusions supported by the analytic results and by the auxiliary diagnostics.
\begin{table}[tb]
\caption{Claim and evidence map.  Detailed hypotheses accompany each result; numerical diagnostics do not enlarge those hypotheses.}
\label{tab:S-evidence}
\centering\small
\begin{tabular}{P{0.25\textwidth}P{0.66\textwidth}}
\toprule
Result or diagnostic & Conclusion and operational scope \\
\midrule
Complete-Pauli theorem; qubit boundary case & Finite-time all-controller value for the uniform real-state law, stationary lower certificate, fixed-positive-$\delta$ accuracy class, and tangent construction.  These are continuous, unit-efficiency results; no matching tangent lower asymptotic is supplied. \\
Geometric lower bound; proper-slant example & Normal-noise obstruction under the full local and global increment assumptions.  The spin example has an orbit-supported prior; its normal matching upper bound and ambient-Haar capture remain open. \\
Sector and graph corollaries & Extensions under their stated restricted resources: all-to-all sector access or dressed-on-site controls.  Neither gives dimension-uniform total hardware costs or unrestricted global optimality for the graph case. \\
Target-law perturbation; efficiency floor & An additive Wasserstein perturbation window and a necessary mixed-state error floor.  These do not prove a feasible advantage at fixed nonunit efficiency or inverse-accuracy divergence at fixed nonzero perturbation. \\
Finite-step channel and continuum results & Exact discarded-record matching and fixed-regular-policy convergence.  A terminal discrete slot changes the admissible class; there is no uniform policy-wise limit or universal bandwidth theorem. \\
Numerical and learner diagnostics & Checks of the continuous value, particular finite-step strategies, or clipped bridge mechanisms.  They retain unresolved points and failed refinements and do not establish all-policy optimality independently of the proofs. \\
\bottomrule
\end{tabular}
\end{table}

\subsection{Additional mechanism diagnostics and evidentiary roles}

\subsubsection{Peak-clipped bridge}

For the scalar bridge
\begin{equation*}
 \dd q_t=\sigma\dd W_t,\qquad q_T\sim N(0,\epsilon^2),
\end{equation*}
we evaluated the exact Gaussian Doob drift after clipping it to $|u|\le U$.
The curves collapse under $r=U\epsilon/\sigma^2$.
At fixed normalized error $W_1/\epsilon=0.15$, the required peaks were
$4.323,9.088,18.706$ for
$\epsilon=0.08,0.04,0.02$, producing the exploratory slope $-1.057$.
The protocol used $\sigma^2=0.35$, four seeds, and $16384$ trajectories per condition.
Independent recomputation of $88$ saved arrays agreed to $8.88\times10^{-16}$.
This is a local boundary-layer diagnostic for an ideal bridge, not a channel-equivalent quantum-instrument comparison and not evidence for the all-feedback theorem.

\begin{figure}[t]
 \includegraphics[width=0.82\textwidth]{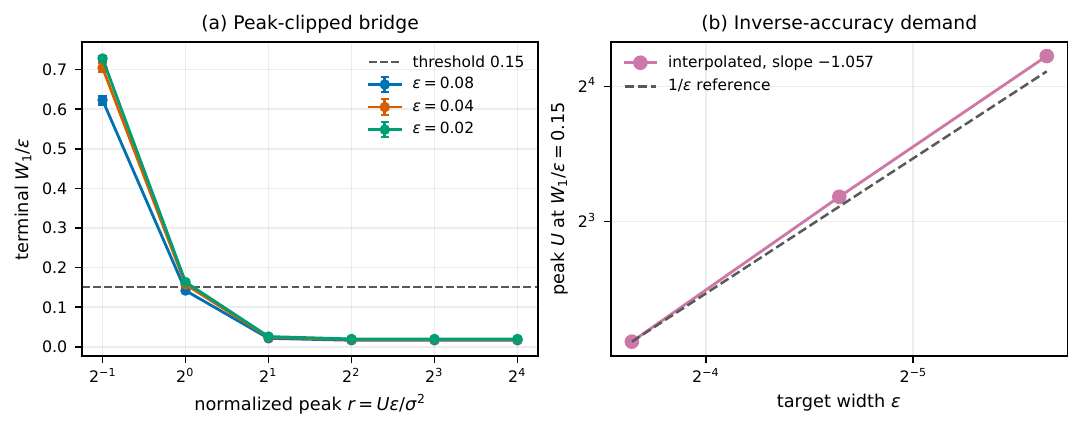}
 \caption{Peak-clipped bridge diagnostic.  The dimensionless collapse is consistent with inverse-accuracy peak demand.}
\end{figure}

\subsubsection{Same-prior \texorpdfstring{$d=4$}{d=4} double-sided bridge}

To compare against the closest fixed-diffusion construction without changing the prior, we also solved a double-sided Schr\"odinger system for the $d=4$, $\delta=1$ radial reference process with $\Gamma=4$, $\kappa=0.5$, and $T=1$.
The initial cell masses are the analytic complex-Haar masses $p_0=\pi$, while
\begin{equation}
 (p_T)_i\propto\pi_i\exp[-q_i^2/(2\epsilon^2)].
\label{eq:S-d4-bridge-target}
\end{equation}
For the implicit-Euler reference kernel $K$, Sinkhorn determines $f_0,g_T$ from
$f_0=p_0/(K^{N_t}g_T)$ and $g_T=p_T/[(K^T)^{N_t}f_0]$.
The unconstrained bridge is then advanced with the exact discrete Doob update
\begin{equation}
 p_{k+1}=h_{k+1}\odot K^T(p_k\oslash h_k),
 \qquad h_k=K^{N_t-k}g_T.
\label{eq:S-discrete-Doob}
\end{equation}
This uses precisely the same kernel as Sinkhorn and reproduces both marginals to numerical precision.
For the finite-peak diagnostic only, the continuous face demand $\kappa\,\Delta\log h/\Delta q$ is clipped to $[-U,U]$ and propagated by implicit substeps.

On the fine grid $(N,N_t)=(800,1600)$, the six unconstrained cases have $W_1/\epsilon<3.70\times10^{-11}$; the largest Sinkhorn initial and terminal L1 residuals over both grids are $1.015\times10^{-12}$ and $4.674\times10^{-15}$.
All 30 terminal arrays are finite and nonnegative, and clipped $W_1$ is nonincreasing in $U$ for every tested $\epsilon$ and grid.
Figure~\ref{fig:S-d4-bridge} shows that a fixed clipping level ceases to resolve the target as it narrows.

\begin{figure}[t]
 \includegraphics[width=0.92\textwidth]{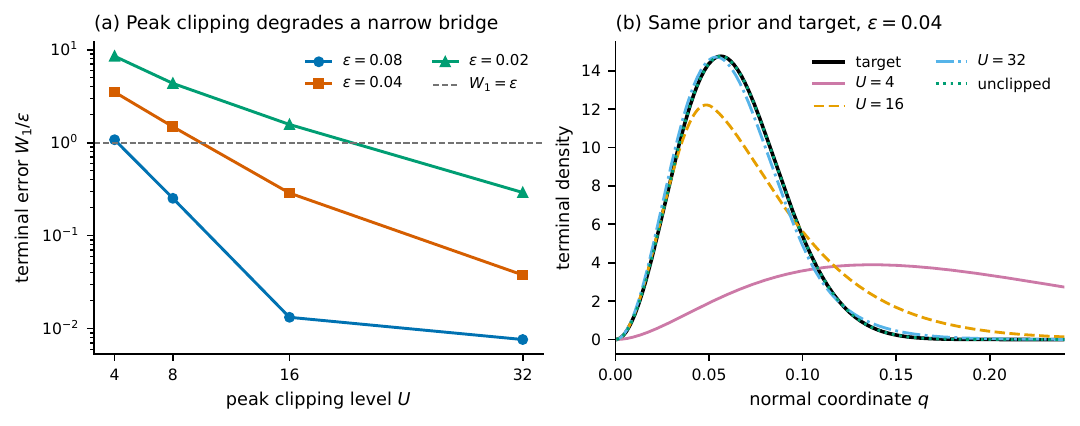}
 \caption{Same-prior $d=4$ bridge mechanism baseline on the fine grid.  (a) Terminal radial $W_1/\epsilon$ after clipping the ideal bridge demand.  The exact discrete unclipped bridge lies below $3.70\times10^{-11}$ and is omitted from the log-scale curves.  (b) Target and terminal densities for $\epsilon=0.04$.  The calculation fixes both Haar initial and target marginals; it does not optimize over quantum Hamiltonian feedback.}
 \label{fig:S-d4-bridge}
\end{figure}

Three of fifteen clipped coarse/fine comparisons exceed the specified $0.01$ threshold:
\begin{equation*}
 (\epsilon,U)=(0.04,32),(0.02,16),(0.02,32),
\end{equation*}
with $W_1/\epsilon$ differences $0.03044$, $0.03732$, and $0.03748$.
These three discrepancies preclude a continuous-limit convergence conclusion for the clipped-rate discretization.
The baseline establishes only that the unconstrained bridge demand, when subjected to the same kind of peak saturation, develops an accuracy-dependent boundary layer; it is neither a full $d=4$ quantum trajectory simulation nor an all-policy optimum.

\subsubsection{Complete-instrument learner stress test}

A separate frozen experiment used sequential $X/Y/Z$ instruments with
$\Gamma=4$, $T=1$, and $h=5\times10^{-4}$.
The Bloch-angle targets were
\begin{equation*}
 p_k(\Theta)=\frac{1+0.6\cos(k\Theta)}{2\pi},
 \qquad k=2,4,
\end{equation*}
where $\Theta=2\vartheta$ is the $2\pi$-periodic Bloch angle.
Every schedule and target had its own trained $320$-parameter controller, with three training seeds and five frozen test seeds.
At $U=8$, the normal-minus-tangent trace-$W_1$ gaps were
$0.05784$ with bootstrap interval $[0.04266,0.07301]$ for $k=2$ and
$0.05851$ with interval $[0.04434,0.07267]$ for $k=4$.
Terminal normal exposure had Spearman correlation $0.8944$ with the full-distribution gap.
This test asks whether the geometric ordering survives finite learner optimization after schedule-specific retraining.
It neither estimates the infimum in Eq.~(\ref{eq:S-cost}) nor reproduces an exact reverse-diffusion task, because the uncontrolled schedule endpoints differ.

\subsection{Data provenance and retained limitations}

The complete-Pauli resolution data are stored separately from the manuscript source.  Four seeds, $2048$ trajectories per seed, and eight conditions give $65536$ saved terminal states.  The archive contains $32$ arrays of shape $(2048,3)$; numerals such as \texttt{n16384} in array keys denote $1/h$ time steps, not trajectory counts.  Both nonmonotonic refinements in Fig.~\ref{fig:S-clockdiagnostic} are included in the displayed data.

The fourteen-file source manifest has matching hashes.  Recalculation from the terminal arrays reproduces support error within $6.65\times10^{-11}$, the empirical coupling bound within $2.44\times10^{-10}$, and the stationary integral within $2.04\times10^{-15}$; the largest float32 norm discrepancy is $4.16\times10^{-8}$.  This verifies the saved-state analysis rather than rerunning the trajectories.  Figure~\ref{fig:S-twoqubit-old}(c) uses twelve target--target ensembles of $1024$ states, with the first $384$ entering each assignment; their mean and seed standard deviation are $0.20704$ and $0.00498$.  The separate qubit angular-$W_1$ sampling baseline uses master seed $20260920$, $1000$ groups, four synthetic seeds per group, and $2048$ target samples per seed.  Scripts, raw arrays, manifests, thresholds, and unsuccessful checks are included in the reproducibility package.

The numerical scripts, saved terminal-state arrays, source manifests, thresholds, and retained unsuccessful checks underlying the reported diagnostics are preserved in a reproducibility archive.  They are available from the corresponding author, Qinglin Zhao (\href{mailto:qlzhao@must.edu.mo}{qlzhao@must.edu.mo}), upon reasonable request.  The archive's \texttt{README.md} specifies the calculation commands, the scope of each data batch, and a file-level manifest.  The revision notes distinguish calculations rerun for this version from inherited experiment records.  The archive is not attached to this arXiv version, and no public repository deposit is asserted here.

\section{Scope boundaries}

The complete-Pauli, single-excitation, and graph-dressed cost separations assume unit efficiency, a known record-conditioned pure state, a record-labelled output, and peak-bounded absolutely continuous feedback $H_t\dd t$.
The general terminal-window theorem instead states its stopped-It\^o, bracket, and finite-variation requirements explicitly and is not restricted to quantum projective space.
The finite-time exact value holds for the complete-Pauli full-space theorem at $d=2^n$; the qubit theorem separately handles its reflected lower boundary and smooth-approximation issue.
The multiqubit result uses the full $d=2^n$ system state, all $4^n-1$ nonidentity Pauli channels, a complex-Haar prior, the uniform $\mathbb{RP}^{d-1}$ target, and unrestricted global Hamiltonians.
The single-excitation result uses $m$ physical qubits, $m^2$ one- and two-body directions, a conserved $m$-dimensional sector, all-to-all pair access, and the operator norm restricted to that sector.
The graph corollary instead uses a graph-dressed product-Haar prior, the structured target $\mathcal N_G$, $3n$ bounded-support channels, and the per-site dressed-local class in Eq.~(\ref{eq:S-graph-controls}); it does not compare against unrestricted global Hamiltonians.
Unit efficiency is a quantum-limited benchmark rather than a present-hardware claim: superconducting-qubit weak measurement has reached $72\pm4\%$ efficiency \cite{Lecocq2021EfficientMeasurement}, while finite detector bandwidth has its own joint system--detector stochastic description \cite{AnnbyAndersson2022FiniteBandwidth}.
At efficiency below one, the public filter is mixed.  Theorem~\ref{thm:S-complete-Pauli-efficiency-floor} gives a positive all-controller pure-target floor for the main instruments, while their exact density-matrix-space minimax cost and a peak-dependent matching upper bound remain open.  Proposition~\ref{prop:S-efficiency-floor} is a sharper exact value only for its separately declared no-knowledge cancellation benchmark.
The score-based feedback construction of Dubey and John \cite{DubeyJohn2026ScoreFeedback} also uses the current of the same measurement increment: in their notation, $H_{\rm meas}=rA/\tau$ with $r\dd t=\langle A\rangle\dd t+\sqrt\tau\dd W$.  Its Hamiltonian increment therefore contains an $A\dd W/\sqrt\tau$ term.  This continuous-time resource differs from the almost-sure bound on the ordinary amplitude $H_t$ in $H_t\dd t$ used here; the present lower bound does not apply to that singular current-feedback class.
Same-increment no-knowledge feedback can alter the diffusion tensor \cite{Szigeti2014NoKnowledge,Saiphet2021DelayedFeedback}; its sample paths have unbounded variation and require a joint peak-bandwidth cost.
Delayed records, auxiliary system-reset or adaptive-ancilla channels beyond the fixed instrument dilation, postselection, and locality classes beyond the specific dressed-on-site contract are separate operational problems.

The two closest control formulations differ from the present resource question
along independent axes \cite{WisemanDoherty2005,OhzekiJordan2026MeasuredBridge}:
\begin{table}[tb]
\caption{Boundary to the closest measurement--control and bridge formulations.
The distinction is the optimized object and resource functional, not a claim
that the prior formulations cannot vary measurements or prepare endpoints.}
\label{tab:S-nearest-work}
\centering
\small
\begin{tabular}{@{}P{0.18\textwidth}P{0.36\textwidth}P{0.36\textwidth}@{}}
\toprule
Prior work & Established question & Boundary of the present theorem \\
\midrule
Wiseman--Doherty & For fixed linear system--environment dynamics, choose the environment measurement/unravelling that optimizes stationary linear-quadratic-Gaussian feedback performance; the result covers state-based and Markovian current feedback. & Fix one nonlinear terminal state-law task, prior, horizon, public-record interface, and $L^\infty$ Hamiltonian class; then compare exact channel-equivalent instruments and prove different accuracy exponents after optimizing over every declared causal controller. \\
Measured quantum Schr\"odinger bridge & For one chosen monitored reference diffusion, use a Doob--Sinkhorn potential to impose an endpoint law or terminal effect and synthesize a quadratic-cost Hamiltonian control signal. & Vary the physical reference instrument while preserving its discarded-record channel, and compare the minimum peak-resource scaling.  The all-controller normal lower bound is not obtained by solving a bridge. \\
\bottomrule
\end{tabular}
\end{table}
Wiseman--Doherty already show explicitly, in their Eq.~(24), that the optimized stationary quadratic cost depends on the unravelling through its conditional covariance.  The new assertion here is the nonlinear terminal-law peak--accuracy separation and finite-time all-controller value, not the general dependence of optimized control performance on monitoring.

Table~\ref{tab:S-diffusion-bridge} states the separate boundary to QuDDPM, and Proposition~\ref{prop:S-discrete-continuum-bridge} gives only fixed-policy convergence.  Joint measurement--control optimization also appears in the oscillator endpoint-path problem of Karmakar and Jordan \cite{KarmakarJordan2026Pontryagin}.  The present distinction requires all three ingredients: identical discarded-record channels, a common terminal-law task, and the minimum peak strength after optimization over all $\A_U$.
Control-limit results bound fidelity to a single target under decoherence or monitored feedback \cite{KobayashiYamamoto2019ControlLimit,OConnorMaGenoni2025FidelityBound}; optimal environment measurements have also been chosen for linear-quadratic objectives \cite{WisemanDoherty2005}, and reverse-time quantum diffusions construct ensemble generators for fixed monitored noise \cite{Bompais2026Reverse}.
Neither establishes a nonlinear terminal-law $L^\infty$ cost separation between channel-equivalent instruments.
Universal continuous-distribution generation by parameterized quantum circuits \cite{BartheEtAl2025ContinuousGenerators} concerns classical output distributions driven by circuit parameters; it does not compare record-conditioned quantum-state laws under a fixed Lindbladian.
Optimized unravelings for classical simulation \cite{ChenBaoChoi2024} minimize a computational trajectory objective, while gauge freedoms classify when monitorings give the same conditional trajectories \cite{Brown2025GaugeUnravellings}; the cost in Eq.~(\ref{eq:S-cost}) instead prices a physical Hamiltonian needed to realize a terminal ensemble.
The reflected-diffusion comparison and boundary-density ideas themselves belong to classical stochastic analysis \cite{LionsSznitman1984,KruhnerXu2023Density}; the novelty claim is the exact quantum-instrument matching, the radial reduction, and the ensuing peak-cost and calibration laws.
The two records here have the same public alphabet and rate but different state informativeness.  Equations~(\ref{eq:S-record-Fisher})--(\ref{eq:S-record-Fisher-delta}) quantify the fixed-state local record information per unit time.  They are neither the Fisher information of an entire closed-loop path nor record--target mutual information.  In this local sense the present one-parameter family does not contain an information-matched pair with distinct diagonal normal noise.
Accordingly, the theorem attributes the separation to the full conditional instrument---its noise orientation and information structure---rather than to a Lindblad generator alone.  Separating those two instrument properties at fixed full record Fisher information remains open.
The remaining laboratory idealizations are therefore explicit: complete record collection and sufficiently fast state estimation; the full-Haar theorem has an exponential monitor count and global control; the single-excitation result uses all-to-all sector hardware; and the graph corollary omits the circuit, ancilla, total-norm, efficiency, and delay costs of its bounded-support observables.\end{document}